\documentclass[12pt]{amsart}

\usepackage{amsmath,amsthm,amsfonts,amssymb,eucal, bbm}
\usepackage{scalerel}
\usepackage{enumitem}
\usepackage{color}
\usepackage{dsfont}

\newcommand{\g}{\gamma}

\newcommand{\real}{\mathbb{R}}
\newcommand{\naturals}{\mathbb{N}}

\newcommand{\complex}{\mathbb{C}}

\newcommand{\oq}{\ {\raise 7pt\hbox{${\scriptstyle\circ}$}}
	\kern -7pt{
		\hbox{$Q$}}}

\newcommand{\R}{ \mathbb R}

\newcommand {\bx}{\mathbf x}
\newcommand {\hbx}{\hat{\bx}}

\newcommand {\by}{\mathbf y}

\newcommand {\bnu}{\boldsymbol\nu}

\newcommand {\bal}{\boldsymbol\alpha}
\newcommand{\bbeta}{\boldsymbol\beta}

\newcommand {\bsig}{\boldsymbol\sigma}

\newcommand\norm[1]{\left\lVert#1\right\rVert}

\newcommand{\2}{\1\!\1}

\newtheorem{thm}{Theorem}[section]
\newtheorem{cor}[thm]{Corollary}
\newtheorem{lem}[thm]{Lemma}
\newtheorem{prop}[thm]{Proposition}

\theoremstyle{definition}

\newtheorem*{remark}{Remark}
\newtheorem{rem}[thm]{Remark}

\numberwithin{equation}{section}

\newcommand{\bee}{\begin{equation}}
	\newcommand{\ene}{\end{equation}}
\newcommand{\bees}{\begin{equation*}}
	\newcommand{\enes}{\end{equation*}}
\newcommand{\bes}{\begin{split}}
	\newcommand{\ens}{\end{split}}

\newcommand{\bet}{\begin{thm}}
	\newcommand{\ent}{\end{thm}}
\newcommand{\bel}{\begin{lem}}
	\newcommand{\enl}{\end{lem}}
\newcommand{\bec}{\begin{cor}}
	\newcommand{\enc}{\end{cor}}
\newcommand{\bep}{\begin{proof}}
	\newcommand{\enp}{\end{proof}}
\newcommand{\ber}{\begin{rem}}
	\newcommand{\enr}{\end{rem}}

\newcommand{\al}{\alpha}

\newcommand{\1}{\mathbbm 1}

\begin{document}
	\hoffset -4pc

\title
[Fifth Derivative of the One-Particle Density Matrix]
{Boundedness of the Fifth Derivative for the One-Particle Coulombic Density Matrix at the Diagonal}
\author{Peter Hearnshaw}
\address{Centre for the Mathematics of Quantum Theory\\ University of Copenhagen\\
	Universitetsparken 5\\ DK-2100 Copenhagen \O\\ Denmark}
\email{ph@math.ku.dk}

\begin{abstract}	 
Boundedness is demonstrated for the fifth derivative of the one-particle reduced density matrix for non-relativistic Coulombic wavefunctions in the vicinity of the diagonal. To prove this result, strong pointwise bounds are obtained for cluster derivatives of wavefunctions involving multiple clusters.
\end{abstract}


\maketitle

\section{Introduction and results}

We consider the non-relativistic quantum system of $N \ge 2$ electrons among $N_0 \ge 1$ nuclei, of fixed position, which represents the system of an atom or molecule. For simplicity we surpress spin coordinates and restrict ourselves to the case of an atom ($N_0 = 1$), although all results are applicable to the case where spin is included and also readily generalise to the molecular case. The electrons therefore have coordinates $\bx = (x_1, \dots, x_N), x_k \in \real^3$, $k=1, \dots, N$, and the nucleus has charge $Z>0$ whose position we choose fixed at the origin in $\real^3$. The corresponding Schr\"{o}dinger operator is
\begin{equation}
\label{eq:so}
H = -\Delta + V
\end{equation}
where $\Delta = \sum_{k=1}^N \Delta_{x_k}$ is the Laplacian in $\real^{3N}$, i.e. $\Delta_{x_k}$ refers to the Laplacian applied to the variable $x_k$, and $V$ is the Coulomb potential given by
\begin{equation}
\label{coulomb}
V(\bold{x}) = -\sum_{k=1}^N \frac{Z}{|x_k|} + \sum_{1\le j<k \le N}\frac{1}{|x_j - x_k|}
\end{equation}
for $\bold{x} \in \real^{3N}$. This operator acts in $L^2(\real^{3N})$ and is self-adjoint on $H^2(\real^{3N})$, see for example \cite[Theorem X.16]{rs2}. We consider solutions to the eigenvalue problem for $H$ in the operator sense, namely
\begin{equation}
\label{eq:se}
H\psi = E\psi
\end{equation}
for $\psi \in H^2(\real^{3N})$ and $E \in \real$. 

There is a long history of regularity and differentiablity properties of wavefunctions $\psi$, a few selected references include \cite{kato}, \cite{HOS_81}, \cite{densities_atoms}, \cite{yse_02}, \cite{analytic_rep}, \cite{ammann12}, \cite{coulomb_estimates}. In recent years there has been growing interest in the regularity of the reduced density matrices, which are objects derived from such wavefunctions. 
Of particular interest is the non-smoothness at the ``diagonal'', which is a signature of electron-electron interation. In previous work by A.V. Sobolev and the current author, it was shown that the density matrix remains bounded at the diagonal upon taking up to four derivatives, \cite{hearn_sob2}, and it was also suggested that this should remain true for five derivatives, see also \cite{cio20}, \cite{cio22}. The current work demonstrates this to be correct.

We now introduce notation and define precisely our objects of study. For each $j=1, \dots, N$, we represent
\begin{align*}
\hat\bx_j = (x_1, \dots, x_{j-1}, x_{j+1},\dots, x_N), \quad (x, \hat\bx_j) = (x_1, \dots, x_{j-1}, x, x_{j+1},\dots, x_N)
\end{align*}
for $x \in \real^3$. We define the \textit{one-particle reduced density matrix}, or simply \textit{density matrix}, by
\begin{align}
\label{eq:gamma}
\g(x, y) = \int_{\R^{3N-3}}
\psi(x, \hat\bx)  \overline{\psi(y, \hat\bx)} \, d\hat\bx,\ 
\quad \hat\bx = \bold{\hat{x}}_1.
\end{align}
for $x,y \in \real^3$.
More commonly, the one-particle reduced density matrix is defined as the function
\begin{align}\label{eq:den}
\tilde\g(x, y) = \sum_{j=1}^N\int\limits_{\R^{3N-3}}\psi(x, \hat\bx_j) 
 \overline{\psi(y, \hat\bx_j)}\,  d\hat\bx_j,\quad x,y\in\R^3. 
\end{align}
However, since we are interested only in regularity properties we need only study one term of (\ref{eq:den}), hence our use of the definition (\ref{eq:gamma}). In fact, we have $\tilde\g(x,y) = N\gamma(x,y)$ whenever $\psi$ is totally symmetric or antisymmetric, although in this work we will assume no symmetry of $\psi$. An important related function is the \textit{one-particle density}, or simply the \textit{density}, which is defined here as
\begin{align}\label{eq:dens}
\rho(x) = \g(x, x) = \int_{\R^{3N-3}}
|\psi(x, \hat\bx)|^2 \, d\hat\bx, \quad x \in \real^3.
\end{align}

Now suppose $\psi$ is any eigenfunction obeying (\ref{eq:se}) with $\norm{\psi}_{L^2(\real^{3N})} = 1$, and let $\rho$ and $\gamma$ be the corresponding functions as defined in (\ref{eq:dens}) and (\ref{eq:gamma}) respectively. In \cite{hearn_sob}, see also \cite{jecko22}, real analyticity was proven for $\gamma(x,y)$ as a function of two variables in the set
\begin{equation*}
\mathcal{D} = \{(x,y) \in \real^3 \times \real^3 : x \ne 0, y \ne 0, x \ne y \}.
\end{equation*}
In particular, real analyticity was not shown across the diagonal, that is where $x=y$. Despite this, real analyticity holds for the density $\rho$ on the set $\real^3\backslash\{0\}$ as shown in \cite{density_analytic}, see also \cite{jecko10}. Therefore, there is real analyticity of $\gamma$ in one variable along the diagonal $x=y$, excluding the point $x=y=0$. However, there is no smoothness of $\gamma$ across the diagonal, as discussed in \cite[Remark 1.2(7)]{hearn_sob2} in relation to the polynomial decay of eigenvalues of the density matrix operator in \cite{sob_estimates} and \cite{sob_asymptotics}, see also \cite{sob_kinetic}. We also note that non-smoothness has also been demonstrated directly for the ($N-1$)-particle reduced density matrix in \cite{jecko23}.

The first indication of non-smoothness at the diagonal seems to be from \cite{cio20}, see also \cite{cio22}. Here, quantum chemistry calculations indicated the existence of a fifth-order cusp, of the form $|x-y|^5$, at the diagonal. To elucidate the structure at the diagonal rigorously, pointwise derivative estimates for $\gamma$ on the set $\mathcal{D}$ were then given in \cite{hearn_sob2}. Due to its importance in the context of the current result, we state the following bounds which arise from \cite[Theorem 1.1]{hearn_sob2}. 
For $b \ge 0$, $t >0$, first define
\begin{align*}
h_b(t) =
\begin{cases}
t^{\min\{0, 5-b \}} &\text{if } b \ne 5\\
\log(t^{-1}+2) &\text{if } b = 5.
\end{cases}
\end{align*}
We denote  $\naturals_0 = \naturals \cup \{0\}$. For all $R>0$ and $\alpha,\beta \in \naturals_0^3$ with $|\alpha|,|\beta| \ge 1$ there exists $C$ such that
\begin{multline}
\label{eq:hearnsob}
|\partial^\alpha_x \partial_y^{\beta} \gamma(x,y)| \le C\big(1 + |x|^{2-|\alpha|-|\beta|} + |y|^{2-|\alpha|-|\beta|} \\+ h_{|\alpha|+|\beta|}(|x-y|)\big) \norm{\rho}^{1/2}_{L^{1}(B(x, R))} \norm{\rho}^{1/2}_{L^{1}(B(y, R))},
\end{multline}
and for all $|\alpha| \ge 1$ there exists $C$ such that
\begin{multline}
\label{eq:hearnsob2}
|\partial^\alpha_x \gamma(x,y)| + |\partial^\alpha_y \gamma(x,y)| \le C\big(1 + |x|^{1-|\alpha|} + |y|^{1-|\alpha|} \\+ h_{|\alpha|}(|x-y|)\big) \norm{\rho}^{1/2}_{L^{1}(B(x, R))} \norm{\rho}^{1/2}_{L^{1}(B(y, R))},
\end{multline}
for all $x,y \in \real^3$ with $x \ne 0$, $y \ne 0$ and $x \ne y$. The notation $\partial_{x}^{\alpha}$ refers to the $\alpha$-partial derivative in the $x$ variable. The constant $C$ depends on $\alpha, \beta, R, N$ and $Z$. The right-hand side is finite because $\psi$ is normalised and hence $\rho \in L^1(\real^3)$. The density matrix $\g$ can be seen to be locally Lipschitz continuous in view of the local boundedness of the first derivative of $\g$ on $\real^6$. In addition, there is local boundedness of up to four derivatives at the diagonal.

Finally, we remark upon a recent result by T. Jecko, \cite{jecko24}, where it is shown that a formula can be produced in the special case of the ($N-1$)-particle density matrix which involves a term with a fifth order cusp along with a smoother remainder term. Due to the method of proof involving the analytic representation of $\psi$ in \cite{analytic_rep}, this result cannot be extended in any straightforward manner to the general $k$-particle case, including the one-particle case.
 
Our main result is as follows.
\begin{thm}
\label{thm:5}
Let $\psi$ be an eigenfunction of (\ref{eq:se}). Define $m(x,y) = \min\{1, |x|, |y| \}$. Then for all $|\alpha|+|\beta| =5$ we have $C$, depending on $Z,N$ and $E$, such that
\begin{align}
\label{eq:5}
|\partial_x^{\alpha}\partial_y^{\beta}\gamma(x,y)| \le 
\begin{cases}
Cm(x,y)^{-3}
\norm{\rho}^{1/2}_{L^{1}(B(x, 1))} \norm{\rho}^{1/2}_{L^{1}(B(y, 1))} \quad &\text{if } |\al|,|\beta|\ge 1 \\
Cm(x,y)^{-4}
\norm{\rho}^{1/2}_{L^{1}(B(x, 1))} \norm{\rho}^{1/2}_{L^{1}(B(y, 1))} \quad &\text{otherwise}\\
\end{cases}
\end{align}
for all $x,y \in \real^3$ obeying $0<|x-y| \le m(x,y)/2$.
\end{thm}

\begin{rem}
\begin{enumerate}
\item 
The bound (\ref{eq:5}) naturally complements (\ref{eq:hearnsob}) and (\ref{eq:hearnsob2}) for $|\alpha|+|\beta| \ne 5$.
\item 
As a consequence of \cite[Proposition A.1]{hearn_sob2} and Sobolev embeddings, the inequality (\ref{eq:5}) shows that $\g\in C^{4, 1}
\big((\R^3\backslash\{0\})\times(\R^3\backslash\{0\})\big)$. The definition of this space is given after Remark \ref{rem:ddpsi} below.
\item
In view of \cite[Remark 1.2(7)]{hearn_sob2}, the boundedness of the fifth derivative at the diagonal is consistent with the asymptotics of eigenvalue decay of the density matrix operator given in \cite{sob_asymptotics}. In addition, such analysis indicates the sixth derivative will in general be unbounded.
\end{enumerate}
\end{rem}

The idea of the proof is to turn partial derivatives of the density matrix $\gamma(x,y)$, with a suitable cutoffs, into certain directional derivatives, known as cluster derivatives, under the integral. The resulting integrals can then be estimated using pointwise bounds to cluster derivatives of $\psi$. This method builds upon similar strategies employed in \cite{density_smooth}, \cite{density_analytic}, \cite{coulomb_estimates} for differentiating the density $\rho$, and \cite{hearn_sob}, \cite{hearn_sob2} for differentiating the density matrix. However, proving boundedness of the fifth derivatives of the density matrix requires a significant refinement of these techniques.

We will consider derivatives of $\gamma(x,y)$, not just in $x$ and $y$ but also in the direction $x+y$. Such derivatives will turn out not to contribute to the singularity at the diagonal. To show this, we will require bounds to cluster derivatives of $\psi$ involving multiple clusters and which, crucially, will distinguish between the orders of each cluster derivative. The result is Theorem \ref{thm:ddpsi} and is proven in Section \ref{chpt:2}. Cutoffs will be introduced in Section \ref{chpt:3}, these are similar to those used in \cite{hearn_sob}, \cite{hearn_sob2} but more sophisticated since particles are grouped up, not only into different so-called clusters associated with $x$, $y$, but crucially also associated to $x+y$. In Section \ref{chpt:4}, we will show that, due to the results on the $(x+y)$-derivatives, it suffices to show boundedness at the diagonal for derivatives $\partial_x^{\alpha}\partial_y^{\beta}\gamma$ with $|\alpha|, |\beta| \ge 1$. For such derivatives we require a representation of cluster derivatives of $\psi$ in terms of a ``good'' term and a ``bad'' term, see (\ref{eq:gp_def}) of Theorem \ref{thm:ddpsi}. The bounds for the ``good'' term is stronger than in previous results, for example those obtained in \cite[Proposition 1.12]{coulomb_estimates} and \cite[Theorem 4.3]{hearn_sob2}. The additional strengthening is proven using the optimal multiplicative factors of \cite{fsho_c11}. The ``bad'' term cannot be bounded sufficiently, but it can be handled using a series of steps involving integration by parts.

\text{ }\\
\textbf{Notation.} 
%
%
%
%
As mentioned earlier, we use a standard notation whereby $\bold{x} = (x_1, \dots, x_N) \in \real^{3N}$, $x_j \in \real^3$, j=1, \dots, N, and where $N$ is the number of electrons. In addition, define for $1 \le j,k \le N$, $j \ne k$,
\begin{align}
\label{eq:bx1}
\bold{\hat{x}}_j &= (x_1, \dots, x_{j-1}, x_{j+1}, \dots, x_N)\\
\bold{\hat{x}}_{j,k} &= (x_1, \dots, x_{j-1}, x_{j+1}, \dots, x_{k-1}, x_{k+1}, \dots,  x_N)
\end{align}
with obvious modifications if either $j,k$ equals $1$ or $N$, and if $k<j$.
We define $\bold{\hat{x}} = \bold{\hat{x}}_1$, which will be used throughout. Variables placed before $\bold{\hat{x}}_j$ and $\bold{\hat{x}}_{j,k}$ will be placed in the removed slots as follows, for any $x,y \in \real^3$ we have
\begin{align}
(x,\bold{\hat{x}}_j) &= (x_1, \dots, x_{j-1}, x, x_{j+1}, \dots, x_N),\\
\label{eq:bx4}
(x,y,\bold{\hat{x}}_{j,k}) &= (x_1, \dots, x_{j-1}, x, x_{j+1}, \dots, x_{k-1}, y, x_{k+1}, \dots,  x_N).
\end{align}
In this way, $\bold{x} = (x_j, \bold{\hat{x}}_j) = (x_j,x_k,\bold{\hat{x}}_{j,k})$.

For variables $x,y \in \real^3$, the directional derivative $\partial_{x+y}$ will be defined by
\begin{align}
\label{eq:dd}
\partial^{\al}_{x+y} = \partial_x^{\al} + \partial_y^{\al}  \quad \text{for } \alpha \in \naturals_0^3,\, |\alpha|=1.
\end{align}
Higher order derivatives $\partial^{\al}_{x+y}$ are defined by successive first order derivatives.

We define a \textit{cluster} to be any subset $P \subset \{1, \dots, N\}$. Denote $P^c = \{1, \dots, N\}\backslash P$, $P^* = P\backslash\{1\}$. We will also need \textit{cluster sets}, $\bold{P} = (P_1, \dots, P_M)$, where $M \ge 1$ and $P_1, \dots, P_M$ are clusters.


First-order \textit{cluster derivatives} are defined, for a non-empty cluster $P$, by
\begin{equation}
\label{eq:cd}
D^{\alpha}_P = \sum_{j \in P} \partial_{x_j}^{\alpha} \quad \text{for } \alpha \in \naturals_0^3,\, |\alpha|=1.
\end{equation}
For $P=\emptyset$, $D^{\alpha}_P$ is defined as the identity. Higher order cluster derivatives, for $\alpha = (\alpha', \alpha'', \alpha''') \in \naturals_0^3$ with $|\alpha| \ge 2$, are defined by successive application of first-order cluster derivatives as follows,
\begin{equation}
\label{eq:cd2}
D^{\alpha}_P = (D_P^{e_1})^{\alpha'}(D_P^{e_2})^{\alpha''}(D_P^{e_3})^{\alpha'''}
\end{equation}
where $e_1, e_2, e_3$ are the standard unit basis vectors of $\real^3$. Let $\bold{P} = (P_1, \dots, P_M)$ and $\bal = (\alpha_1, \dots, \alpha_M)$, $\alpha_j \in \naturals_0^3$, $1 \le j \le M$, then we define the \textit{multicluster derivative} (often simply referred to as cluster derivative) by
\begin{equation}
D_{\bold{P}}^{\bal} = D_{P_1}^{\alpha_1} \dots D_{P_M}^{\alpha_M}.
\end{equation}
It can readily be seen that cluster derivatives obey a Leibniz rule.

Throughout, the letter $C$ refers to a positive constant whose value is unimportant and may depend on $Z$, $N$ and the eigenvalue $E$.
\text{ }\\\\
\textbf{Distance function notation and elementary results.}
For non-empty cluster $P$, define
\begin{align}
\label{sigmap2}
\Sigma_P = \Big\{ \bold{x} \in \real^{3N} : \prod_{j \in P} |x_j| \prod_{\substack{l \in P \\ m \in P^c}} |x_l - x_m| = 0 \Big\}.
\end{align}
For $P= \emptyset$ we set $\Sigma_{P} := \emptyset$. Denote $\Sigma_P^c = \real^{3N}\backslash\Sigma_P$. For each $P$ we have $\Sigma_P \subset \Sigma$ where
\begin{equation}
\Sigma = \Big\{ \bold{x} \in \real^{3N} : \prod_{1 \le j \le N} |x_j| \prod_{1 \le l<m \le N} |x_l - x_m| = 0 \Big\}
\end{equation}
is the set of singularities of the Coulomb potential $V$, (\ref{coulomb}). For any cluster $P$ we can define the following distances
\begin{align}
\label{dp_def}
d_P(\bold{x}) &:= \text{dist}( \bold{x}, \Sigma_{P} ) = \min\big\{|x_j|, \, 2^{-1/2}|x_j - x_k| : j \in P, k \in P^c \big\}\\
\label{lam_def}
\lambda_{P}(\bold{x}) &:= \min\{ 1,\, d_P(\bold{x}) \}\\
\label{eq:nup}
\nu_P(\bx) &:= \min\{1; |x_j| : j \in P\}
\end{align}
for all $\bold{x} \in \real^{3N}$.

Using the formula (\ref{dp_def}), see for example \cite[Lemma 4.2]{hearn_sob2}, it can be shown that
\begin{equation}
\label{lam_lip}
|d_P(\bold{x}) - d_P(\bold{y})|, |\lambda_P(\bold{x}) - \lambda_P(\bold{y})|, |\nu_P(\bx) - \nu_P(\by)| \le |\bold{x} - \bold{y}|
\end{equation}
for all $\bold{x}, \bold{y} \in \real^{3N}$.

Let $\bold{P} = (P_1, \dots, P_M)$ be a cluster set and $\bal = (\alpha_1, \dots, \alpha_M) \in \naturals_0^{3M}$ be a multiindex. Define
\begin{equation}
\label{eq:sigma_alpha}
\Sigma_{\boldsymbol{\alpha}} = \bigcup_{j \,:\, \alpha_j \ne 0} \Sigma_{P_j}
\end{equation}
for non-zero $\boldsymbol{\alpha}$ and when $\boldsymbol{\alpha} = 0$ we set
$\Sigma_{\boldsymbol{\alpha}} = \emptyset$. Denote $\Sigma_{\bal}^c = \real^{3N}\backslash\Sigma_{\bal}$. For non-zero $\boldsymbol{\alpha}$ we can also define the distances
\begin{align}
d_{\boldsymbol{\alpha}}(\bold{x}) &= \min\{d_{P_j}(\bold{x}) : \alpha_j \ne 0,\, j=1, \dots, M \},\\
\label{eq:lam_alpha}
\lambda_{\boldsymbol{\alpha}}(\bold{x}) &= \min\{\lambda_{P_j}(\bold{x}) : \alpha_j \ne 0,\, j=1, \dots, M \},
\end{align}
for all $\bold{x} \in \real^{3N}$. Notice we also have the identity
\begin{align*}
d_{\boldsymbol{\alpha}}(\bold{x}) = \text{dist}(\bold{x}, \Sigma_{\boldsymbol{\alpha}}).
\end{align*}
Indeed, since $\Sigma_{P_j} \subset \Sigma_{\boldsymbol{\alpha}}$ whenever $j$ is such that $\alpha_j \ne 0$, we have $\text{dist}(\bold{x}, \Sigma_{\boldsymbol{\alpha}}) \le d_{P_j}(\bold{x})$ and hence $\text{dist}(\bold{x}, \Sigma_{\boldsymbol{\alpha}}) \le d_{\boldsymbol{\alpha}}(\bold{x})$. Conversely for each $\boldsymbol{\xi} \in \Sigma_{\boldsymbol{\alpha}}$ we have $\boldsymbol{\xi} \in \Sigma_{P_j}$ for some $j$ with $\alpha_j \ne 0$, and hence $|\bold{x} - \boldsymbol{\xi}| \ge d_{P_j}(\bold{x})$. Therefore, $|\bold{x} - \boldsymbol{\xi}| \ge d_{\boldsymbol{\alpha}}(\bold{x})$. Taking infimum over $\boldsymbol{\xi} \in \Sigma_{\boldsymbol{\alpha}}$ we obtain the reverse inequality. 

\paragraph{}
We will also use the following related quantities involving maxima of relevant distances for non-zero $\boldsymbol{\alpha}$,
\begin{align}
q_{\boldsymbol{\alpha}}(\bold{x}) &= \max\{d_{P_j}(\bold{x}) : \alpha_j \ne 0,\, j=1, \dots, M \}\\
\label{eq:mu_alpha}
\mu_{\boldsymbol{\alpha}}(\bold{x}) &= \max\{\lambda_{P_j}(\bold{x}) : \alpha_j \ne 0,\, j=1, \dots, M \}.
\end{align}
For $\boldsymbol{\alpha} = 0$ we set $d_{\boldsymbol{\alpha}}, q_{\boldsymbol{\alpha}} \equiv 0$ and $\lambda_{\boldsymbol{\alpha}}, \mu_{\boldsymbol{\alpha}} \equiv 1$.

%

In order to state the results we first define the following for an arbitrary function $u$ and any $r>0$,
\begin{align}
\label{eq:f_inf_p}
f_{\infty}(\bold{x}; r; u) := \norm{\nabla u}_{L^{\infty}(B(\bold{x}, r))} + \norm{u}_{L^{\infty}(B(\bold{x}, r))}
\end{align}
for $\bold{x} \in \real^{3N}$. The ball $B(\bold{x}, r)$ is considered in $\real^{3N}$. Largely, this notation will be used for $u=\psi$, and in this case we have the notation
\begin{align}
\label{eq:f_inf}
f_{\infty}(\bold{x}; r) :=f_{\infty}(\bold{x}; r; \psi).
\end{align}
Furthermore, later we will not track the dependence on $r$, hence it is convenient to define
\begin{align}
\label{eq:f_inf2}
f_{\infty}(\bold{x}) :=f_{\infty}\Big(\bold{x}; \frac{1}{2}; \psi\Big).
\end{align}
\text{ }\\
\textbf{A pointwise cluster derivative bound.}
To prove Theorem \ref{thm:5} we will state and prove a new pointwise bound to cluster derivatives of eigenfunctions $\psi$, which itself is of independent interest. It will be shown by elliptic regularity that for all $\bal$ the weak cluster derivatives $D^{\bal}_{\bold{P}} \psi$ exist in the set $\Sigma_{\bal}^c$. We consider how such cluster derivatives behave as the set $\Sigma_{\bal}$ is approached.

Previously, S. Fournais and T. \O. S\o rensen have given bounds to local $L^p$-norms of cluster derivatives of $\psi$ for a single cluster $P$. Indeed, in \cite[Proposition 1.12]{coulomb_estimates} it is shown that for any multiindex $0 \ne \alpha \in \naturals_0^{3}$, $p \in (1, \infty]$ and any $0<r<R<1$ there exists $C$, depending on $r, R, p$ and $\alpha$, such that
\begin{equation}
\label{eq:fs_p}
\norm{D_{P}^{\alpha} \psi}_{L^{p}(B(\bold{x}, r\lambda_{P}(\bold{x})))} \le C \lambda_{P}(\bold{x})^{1-|\alpha|} \big( \norm{\nabla \psi}_{L^p(B(\bx, R\lambda_P(\bx)))} + \norm{\psi}_{L^p(B(\bx, R\lambda_P(\bx)))}\big)
\end{equation}
for all $\bold{x} \in \Sigma^c_{P}$. Notice that for every $\bx \in \Sigma_P^c$, we have $B(\bx, r\lambda_P(\bx)) \subset \Sigma^c_{P}$ by the definition of $\lambda_P(\bx)$. The exponent $1-|\al|$ is natural since both $\psi$ and $\nabla\psi$ are locally bounded, therefore we only get a singularity for $|\al| \ge 2$.

The objective of the following theorem is to extend the bounds (\ref{eq:fs_p}) in the case of $p=\infty$ and for cluster sets $\bold{P}$ which may contain more than one cluster. Bounds for cluster sets are present in \cite[Theorem 4.3]{hearn_sob2}, but here we obtain estimates which depend on the order of derivative on each cluster in $\bold{P}$. This will be vital in proving the boundedness of the fifth derivative of the density matrix at the diagonal. The symbol $\nabla$ denotes the gradient operator in $\real^{3N}$. 

\begin{thm}
\label{thm:ddpsi}
For every cluster set $\bold{P} = (P_1, \dots, P_M)$, multiindex $\boldsymbol{\alpha} \in \naturals_0^{3M}$ and any $R>0$ there exists $C$, depending on $\bal$ and $R$, such that for $k=0,1$,
\begin{equation}
\label{eq:ddpsi}
|D_{\bold{P}}^{\boldsymbol{\alpha}}\nabla^{k} \psi(\bx)| \le C \lambda_{\boldsymbol{\alpha}}(\bold{x})^{1-k} \lambda_{P_1}(\bold{x})^{-|\alpha_1|} \dots \lambda_{P_M}(\bold{x})^{-|\alpha_M|}
f_{\infty}(\bold{x}; R)
\end{equation}
for all $\bold{x} \in \Sigma^c_{\boldsymbol{\alpha}}$.

Furthermore, for each $|\bal| \ge 1$ there exists a function $G_{\bold{P}}^{\boldsymbol{\alpha}} : \Sigma^c_{\bal} \to \complex^{3N}$ such that
\begin{equation}
\label{eq:gp_def}
D_{\bold{P}}^{\boldsymbol{\alpha}}\nabla \psi = G_{\bold{P}}^{\boldsymbol{\alpha}} + \psi\, D_{\bold{P}}^{\bal}\nabla F_c,
\end{equation}
where $F_c$ is defined later in (\ref{f_def}), and for each $0\le b<1$ there exists $C$, depending on $\bal, R$ and $b$, such that
\begin{equation}
\label{eq:gp_bnd}
|G_{\bold{P}}^{\boldsymbol{\alpha}}(\bold{x})| \\\le C \mu_{\bal}(\bold{x})^b \lambda_{P_1}(\bold{x})^{-|\alpha_1|} \dots \lambda_{P_M}(\bold{x})^{-|\alpha_M|} f_{\infty}(\bold{x}; R )
\end{equation}
for all $\bold{x} \in \Sigma^c_{\boldsymbol{\alpha}}$.

\end{thm}

\begin{rem}
\label{rem:ddpsi}
\begin{enumerate}
\item
In the case of a single cluster and $k=0$, the bound (\ref{eq:ddpsi}) reestablishes (\ref{eq:fs_p}) in the case of $p = \infty$, albeit with a larger radius in the $L^{\infty}$-norms on the right-hand side, which doesn't concern us.
\item
The presence of a single power of $\lambda_{\bal}(\bx)$ in the bound (\ref{eq:ddpsi}), for $k=0$, will cancel a single negative power of the smallest $\lambda_{P_j}(\bx)$ for $j$ such that $\alpha_j \ne 0$. Notice that the appropriate $j$ will depend on $\bx$.
\item
The bound in (\ref{eq:gp_bnd}) is stronger than that of (\ref{eq:ddpsi}) with $k=1$. This is because a positive power of $\mu_{\bal}(\bx)$ will partially cancel a single negative power of the largest $\lambda_{P_j}(\bx)$ for $j$ such that $\alpha_j \ne 0$. The term $\psi\, D_{\bold{P}}^{\bal}\nabla F_c$ in (\ref{eq:gp_def}) only obeys a bound as in (\ref{eq:ddpsi}) for $k=1$, therefore represents the part of $D_{\bold{P}}^{\boldsymbol{\alpha}}\nabla \psi$ which inhibits an improved derivative bound.
\end{enumerate}
\end{rem}


The proof of Theorem \ref{thm:ddpsi} will adopt a similar strategy to that used originally in \cite[Proposition 1.12]{coulomb_estimates}, and developed in \cite[Theorem 4.3]{hearn_sob2}. An additional result will be required to prove (\ref{eq:gp_bnd}). This result, \cite{fsho_c11}, shows that $\psi$ can be made $W^{2,\infty}_{loc}(\real^{3N})$ upon multiplication by a factor, universal in the sense that the factor depends only on $N$ and $Z$.

We will require an elliptic regularity result, stated below, which will be used in the proofs. Beforehand, we clarify the precise form of definitions which we will be using. Let $\Omega$ be open, $\theta \in (0,1]$ and $k=\naturals_0$. We formally define the $\theta$-H\"older seminorms for a function $f$ by
\begin{align*}
&[f]_{\theta, \Omega} = \sup_{\substack{x,y \in \Omega \\ x \ne y}} \frac{|f(x)-f(y)|}{|x-y|^{\theta}},\\
&[\nabla^k f]_{\theta, \Omega} = \sup_{|\alpha|=k}[\partial^\alpha f]_{\theta, \Omega}.
\end{align*}
The space $C^{k,\theta}(\Omega)$ is defined as all $f \in C^k(\Omega)$ where $[\nabla^k f]_{\theta, \Omega'}$ is finite for each $\Omega'$ compactly contained in $\Omega$. In addition, the space $C^{k,\theta}(\overline{\Omega})$ is defined as all $f \in C^k(\overline{\Omega})$ where $[\nabla^k f]_{\theta, \Omega}$ is finite. This space has a norm given by
\begin{align*}
\norm{f}_{C^{k,\theta}(\overline{\Omega})} = \norm{f}_{C^{k}(\overline{\Omega})} + [\nabla^k f]_{\theta, \Omega}.
\end{align*}
 
For open $\Omega \subset \real^n$ we can consider the following elliptic equation,
\begin{equation}
\label{eq:pde}
Lu := -\Delta u + \bold{c} \cdot \nabla u + du = g
\end{equation}
for some $\bold{c} : \Omega \to \complex^n$ and $d,g : \Omega \to \complex$. The corresponding bilinear form for operator $L$ is defined formally as
\begin{equation*}
\mathcal{L}(u, \chi) = \int_{\Omega} \big( \nabla u \cdot \nabla \chi + (\bold{c} \cdot \nabla u) \chi + d u \chi \big) \, dx
\end{equation*}
for all $u \in H^1_{loc}(\Omega)$ and $\chi \in C^{\infty}_c(\Omega)$. We say that a function $u \in H^{1}_{loc}(\Omega)$ is a \textit{weak solution} to the equation (\ref{eq:pde}) in $\Omega$ if $\mathcal{L}(u,\chi) = \int_{\Omega} g\chi \,dx$ for every $\chi \in C^{\infty}_c(\Omega)$.

The following theorem is a restatement of \cite[Proposition 3.1]{hearn_sob2} (\cite[Proposition A.2]{coulomb_estimates} is similar), with additional H\"older regularity which follows from the proof.
\begin{thm}
\label{thm:c1}
Let $x_0 \in \real^{n}$, $R>0$ and $\bold{c}, d, g \in L^{\infty}(B(x_0, R))$ and $u \in H^1(B(x_0, R))$ be a weak solution to (\ref{eq:pde}) then for each $\theta \in [0,1)$ we have $u \in C^{1, \theta}(B(x_0, R)) \cap H_{loc}^2(B(x_0, R))$, and for any $r \in (0,R)$ we have
\begin{equation}
\label{c1_er_eqn}
\norm{u}_{C^{1,\theta}(\overline{B(x_0, r)})} \le C(\norm{u}_{L^2(B(x_0, R))} + \norm{g}_{L^{\infty}(B(x_0, R))})
\end{equation}
for $C=C(n,  K, r, R, \theta)$ where
\begin{equation*}
\norm{\bold{c}}_{L^{\infty}(B(x_0, R))} + \norm{d}_{L^{\infty}(B(x_0, R))} \le K.
\end{equation*}
\end{thm}

\section{Proof of Theorem \ref{thm:ddpsi}}
\label{chpt:2}
The proof will involve the use of a multiplicative factor, frequently called a \textit{Jastrow factor} in mathematical literature, to improve the regularity of $\psi$. This is a strategy that has been used successfully in, for example, \cite{density_smooth}, \cite{densities_atoms} to elucidate regularity properties of $\psi$. We start by defining a function $F = F(\bx)$, depending only on $N$ and $Z$, such that the function $e^{-F}\psi$ solves an elliptic equation with bounded coefficients. These coefficients behave suitably well under the action of cluster derivatives. This allows us to use elliptic regularity to produce bounds to the cluster derivatives of $e^{-F}\psi$ and hence to $\psi$ itself.

\subsection{Jastrow factors}
We begin by defining
\begin{align}
\label{f_def}
F(\bx) &= F_c(\bx) - F_s(\bx),\quad F_c(\bold{x}) = F_c^{(en)}(\bx) + F_c^{(ee)}(\bx)
\end{align}
where
\begin{align}
F^{(en)}_c(\bold{x}) = -\frac{Z}{2}\sum_{1 \le j \le N}|x_j|, \quad F^{(ee)}_c(\bx) = \frac{1}{4}\sum_{1\le l<k \le N} |x_l - x_k|,\\
\label{f2_def}
F_s(\bold{x}) = -\frac{Z}{2}\sum_{1 \le j \le N}\sqrt{|x_j|^2 + 1} + \frac{1}{4}\sum_{1 \le l<k \le N} \sqrt{|x_l - x_k|^2 + 1},
\end{align}
for $\bold{x} \in \real^{3N}$.
The function $F$ was used in \cite{coulomb_estimates}. We now detail some basic facts regarding these functions. Firstly,
\begin{equation}
\label{f1_v}
\Delta F_c = V
\end{equation}
where $V$ is the Coulomb potential, (\ref{coulomb}). The function $F_s$ has the same behaviour as $F_c$ at infinity. However, $F_s \in C^{\infty}(\real^{3N})$ and
\begin{equation}
\label{f2_linf}
\partial^{\alpha} F_s \in L^{\infty}(\real^{3N}) \quad \text{for all }\, \alpha \in \naturals_0^{3N}, |\alpha| \ge 1.
\end{equation}
Finally, the function $F$ obeys
\begin{equation}
\label{f_linf}
F, \nabla F \in L^{\infty}(\real^{3N}).
\end{equation}

The function $F$ is used to define the following object, which will be used throughout the proof:
\begin{equation}
\label{phi_def}
\phi = e^{-F}\psi.
\end{equation}
Using (\ref{eq:se}) the following elliptic equation can be shown to hold for $\phi$ a weak solution,
\begin{equation}
\label{phi_pde_eq}
-\Delta \phi -2 \nabla F \cdot \nabla \phi + (\Delta F_s - |\nabla F|^2 - E)\phi = 0.
\end{equation}
Since all coefficients are bounded in $\real^{3N}$ we can see from Theorem \ref{thm:c1} that $\phi \in C^{1,\theta}(\real^{3N})$ for each $\theta \in [0,1)$. Furthermore, by the same theorem, for any pair $0<r<R$ we can find constants $C, C'$, dependent only on $N, Z, E, r, R$ and $\theta$, such that
\begin{equation}
\label{eq:phi_er}
\norm{\phi}_{C^{1,\theta}(\overline{B(\bold{x}, r)})} \le C\norm{\phi}_{L^{2}(B(\bold{x}, R))} \le C'\norm{\phi}_{L^{\infty}(B(\bold{x}, R))}
\end{equation}
for all $\bold{x} \in \real^{3N}$.

From the definition of $\phi$, (\ref{phi_def}), along with (\ref{f_linf}) there exists new constants $C<C'$, dependent only on $Z$ and $N$ (hence independent of $R>0$), such that
\begin{equation}
\label{eq:phi_psi}
C\norm{\phi}_{L^{\infty}(B(\bold{x}, R))} \le \norm{\psi}_{L^{\infty}(B(\bold{x}, R))} \le C'\norm{\phi}_{L^{\infty}(B(\bold{x}, R))}
\end{equation}
for all $\bold{x} \in \real^{3N}$.
\subsection{Derivatives of F}

Informally, our objective is to take cluster derivatives of the elliptic equation (\ref{phi_pde_eq}) and apply elliptic regularity. To do so, we require bounds to the cluster derivatives of the coefficients present in this equation. This is the objective of the current section. To begin, we state and prove the following preparatory lemma involving the distances introduced in (\ref{lam_def}), (\ref{eq:lam_alpha}) and (\ref{eq:mu_alpha}).
\begin{lem}
\label{mu_lemma}
For any $\boldsymbol{\sigma} = (\sigma_1, \dots, \sigma_M) \in \naturals_0^{3M}$ we have for $k=0,1$,
\begin{equation}
\label{mu_1sig}
\mu_{\boldsymbol{\sigma}}(\bold{y})^{k-|\boldsymbol{\sigma}|} \le \lambda_{\boldsymbol{\sigma}}(\bold{y})^{k} \lambda_{P_1}(\bold{y})^{-|\sigma_1|} \dots \lambda_{P_M}(\bold{y})^{-|\sigma_M|}
\end{equation}
for all $\bold{y} \in \real^{3N}$. Furthermore, for some $n \ge 1$ let $\boldsymbol{\beta}^{(1)}, \dots, \boldsymbol{\beta}^{(n)}$ be an arbitrary collection of multiindices in $\naturals_0^{3M}$ such that $\boldsymbol{\beta}^{(1)} + \dots + \boldsymbol{\beta}^{(n)} = \boldsymbol{\sigma}$, then
\begin{equation}
\label{lam_prod}
\prod_{j=1}^n  \lambda_{\boldsymbol{\beta}^{(j)}}(\bold{y}) \le \lambda_{\boldsymbol{\sigma}}(\bold{y})
\end{equation}
for all $\bold{y} \in \real^{3N}$.
\end{lem}
\begin{proof}
The results are trivial in the case of $\boldsymbol{\sigma} = 0$ therefore we assume in the following that $\boldsymbol{\sigma}$ is non-zero. First, observe for all $j=1, \dots, M$,
\begin{equation*}
\begin{split}
\mu_{\boldsymbol{\sigma}}(\bold{y})^{-|\sigma_j|} &\le \lambda_{P_j}(\bold{y})^{-|\sigma_j|}\\
\mu_{\boldsymbol{\sigma}}(\bold{y})^{1-|\sigma_j|} &\le \lambda_{P_j}(\bold{y})^{1-|\sigma_j|} \quad \text{ if } \sigma_j \ne 0
\end{split}
\end{equation*}
by the definition of $\mu_{\boldsymbol{\sigma}}$. Now perform the following trivial expansion of the product,
\begin{equation*}
\begin{split}
\mu_{\boldsymbol{\sigma}}(\bold{y})^{-|\boldsymbol{\sigma}|} &= \mu_{\boldsymbol{\sigma}}(\bold{y})^{-|\sigma_1|} \dots \mu_{\boldsymbol{\sigma}}(\bold{y})^{-|\sigma_M|}\\
&\le \lambda_{P_1}(\bold{y})^{-|\sigma_1|} \dots \lambda_{P_M}(\bold{y})^{-|\sigma_M|}
\end{split}
\end{equation*}
which proves (\ref{mu_1sig}) for $k=0$. For $k=1$, consider that for each $\bold{y}$ we can find $l = 1, \dots, M$ such that $\lambda_{\boldsymbol{\sigma}}(\bold{y}) = \lambda_{P_l}(\bold{y})$ and $\sigma_l \ne 0$. Then,
\begin{equation*}
\begin{split}
\mu_{\boldsymbol{\sigma}}(\bold{y})^{1-|\boldsymbol{\sigma}|} &= \mu_{\boldsymbol{\sigma}}(\bold{y})^{-|\sigma_1|} \dots \mu_{\boldsymbol{\sigma}}(\bold{y})^{1-|\sigma_l|} \dots \mu_{\boldsymbol{\sigma}}(\bold{y})^{-|\sigma_M|}\\
&\le  \lambda_{P_1}(\bold{y})^{-|\sigma_1|}  \dots \lambda_{P_l}(\bold{y})^{1-|\sigma_l|} \dots \lambda_{P_M}(\bold{y})^{-|\sigma_M|}\\
&=  \lambda_{\boldsymbol{\sigma}}(\bold{y}) \lambda_{P_1}(\bold{y})^{-|\sigma_1|} \dots \lambda_{P_M}(\bold{y})^{-|\sigma_M|}
\end{split}
\end{equation*}
as required.

\paragraph{}
Finally we prove (\ref{lam_prod}). As above, take arbitrary $\bold{y}$ and a corresponding $l$ such that $\lambda_{\boldsymbol{\sigma}}(\bold{y}) = \lambda_{P_l}(\bold{y})$ with $\sigma_l \ne 0$. For each $1 \le j \le n$ we denote the $\naturals_0^3$-components of $\boldsymbol{\beta}^{(j)}$ as $\boldsymbol{\beta}^{(j)} = (\beta^{(j)}_1, \dots, \beta^{(j)}_M)$. We know $\boldsymbol{\beta}^{(1)} + \dots + \boldsymbol{\beta}^{(n)} = \boldsymbol{\sigma}$ so in particular, $\beta_l^{(1)} + \dots + \beta_l^{(n)} = \sigma_l$. Since $\sigma_l \ne 0$ there exists at least one $1 \le r \le n$ such that $\beta^{(r)}_{l} \ne 0$. Hence by the definition of $\lambda_{\boldsymbol{\beta}^{(r)}}$,
\begin{equation*}
\lambda_{\boldsymbol{\beta}^{(r)}}(\bold{y}) = \min\{\lambda_{P_s}(\bold{y}) : \beta_s^{(r)} \ne 0 \text{, } s = 1, \dots, M \} \le \lambda_{P_l}(\bold{y}) = \lambda_{\boldsymbol{\sigma}}(\bold{y}).
\end{equation*}
The remaining factors $\lambda_{\boldsymbol{\beta}^{(j)}}(\bold{y})$, for $j \ne r$, can each be bounded above by one.

\end{proof}

The following lemma will be useful in proving results about taking cluster derivatives of $F$, as defined in (\ref{f_def}). It is a similar, but stronger version of \cite[Lemma 4.4]{hearn_sob2}. Later, we will apply it using $f(x)$ as the function $|x|$ for $x \in \real^3$, or derivatives thereof.
\begin{lem}
\label{der_F}
Let $f \in C^{\infty}(\real^3 \backslash \{0 \})$ and $k \in \naturals_0$ be such that for each $\sigma \in \naturals_0^3$ there exists $C$ such that
\begin{equation}
\label{dsf_ineq}
|\partial^{\sigma} f(x)| \le C|x|^{k-|\sigma|} \text{ for all } x \ne 0. 
\end{equation}
Then for any $\boldsymbol{\alpha} \ne 0$ with $|\boldsymbol{\alpha}| \ge k$ there exists some new $C$ such that for any $l,m = 1, \dots, N,$ the weak derivatives $D_{\bold{P}}^{\boldsymbol{\alpha}}(f(x_l))$ and $D_{\bold{P}}^{\boldsymbol{\alpha}}(f(x_l - x_m))$ are both smooth in $\Sigma_{\boldsymbol{\alpha}}^c$ and obey
\begin{equation*}
|D_{\bold{P}}^{\boldsymbol{\alpha}}(f(x_l))|,\, |D_{\bold{P}}^{\boldsymbol{\alpha}}(f(x_l - x_m))| \le Cq_{\boldsymbol{\alpha}}(\bold{x})^{k-|\boldsymbol{\alpha}|}
\end{equation*}
for all $\bold{x} \in \Sigma_{\boldsymbol{\alpha}}^c$.
\end{lem}
\begin{proof}
Take any $j = 1, \dots, M$ with $\alpha_j \ne 0$, then we have $D^{\alpha_j}_{P_j}(f(x_l)) \equiv 0$ for each $l \in P_j^c$. Therefore, for $D^{\boldsymbol{\alpha}}_{\bold{P}} (f(x_l))$ to not be identically zero we require that $l \in P_j$ for each $j$ with $\alpha_j \ne 0$. For such $l$ we have $x_l \ne 0$ since $\bold{x} \in \Sigma^c_{\boldsymbol{\alpha}}$, and
\begin{equation*}
|x_l| \ge d_{P_j}(\bold{x})
\end{equation*}
for each $j$ with $\alpha_j \ne 0$ by (\ref{dp_def}). Therefore, for constant $C$ in (\ref{dsf_ineq}), we have
\begin{equation*}
|D^{\boldsymbol{\alpha}}_{\bold{P}} (f(x_l))| = |\partial^{\alpha_1 + \dots + \alpha_M}f(x_l)| \le C|x_l|^{k-|\boldsymbol{\alpha}|} \le Cq_{\boldsymbol{\alpha}}(\bold{x})^{k-|\boldsymbol{\alpha}|}
\end{equation*}
because $|\boldsymbol{\alpha}| \ge k$. Similarly, for each $j = 1, \dots, M$, with $\alpha_j \ne 0$ we have $D^{\alpha_j}_{P_j} (f(x_l - x_m)) \equiv 0$ if either $l, m \in P_j$ or $l, m \in P_j^c$. Therefore, for $D^{\boldsymbol{\alpha}}_{\bold{P}} (f(x_l - x_m))$ to not be identically zero we require that
\begin{equation*}
(l,m) \in \bigcap_{j \,:\, \alpha_j \ne 0} \big((P_j \times P_j^c) \cup (P_j^c \times P_j)\big).
\end{equation*}
For such $(l,m)$ we have $x_l \ne x_m$ since $\bold{x} \in \Sigma^c_{\boldsymbol{\alpha}}$ and
\begin{equation*}
 |x_l - x_m| \ge \sqrt{2}\, d_{P_j}(\bold{x})
\end{equation*}
for each $j$ with $\alpha_j \ne 0$ by (\ref{dp_def}). Therefore, for some constant $C'$,
\begin{equation*}
|D^{\boldsymbol{\alpha}}_{\bold{P}} (f(x_l - x_m))| = |\partial^{\alpha_1 + \dots + \alpha_M}f(x_l - x_m)| \le C|x_l - x_m|^{k - |\boldsymbol{\alpha}|} \le C' q_{\boldsymbol{\alpha}}(\bold{x})^{k-|\boldsymbol{\alpha}|}
\end{equation*}
because $|\boldsymbol{\alpha}| \ge k$.
\end{proof}

The following lemma provides pointwise bounds to cluster derivatives of functions involving $F$.
\begin{lem}
\label{der_gradF}
For any cluster set $\bold{P}$ and any $|\bsig| \ge 1$ there exists $C$, which depends on $\boldsymbol{\sigma}$, such that for $k=0,1$,
\begin{equation}
\label{dd_f_ef}
\big|D_{\bold{P}}^{\boldsymbol{\sigma}}\nabla^k F(\bold{y})\big|, \, \big|D_{\bold{P}}^{\boldsymbol{\sigma}} \nabla^k (e^{F})(\bold{y})\big| \le  C \lambda_{\boldsymbol{\sigma}}(\bold{y})^{1-k} \lambda_{P_1}(\bold{y})^{-|\sigma_1|} \dots \lambda_{P_M}(\bold{y})^{-|\sigma_M|}
\end{equation}
\begin{equation}
\label{dd_f_2}
\big|D_{\bold{P}}^{\boldsymbol{\sigma}} |\nabla F(\bold{y})|^2\big| \le C\lambda_{P_1}(\bold{y})^{-|\sigma_1|} \dots \lambda_{P_M}(\bold{y})^{-|\sigma_M|}
\end{equation}
for all $\bold{y} \in \Sigma^c_{\boldsymbol{\sigma}}$. The bound to the first quantity in (\ref{dd_f_ef}) also holds when $F$ is replaced by $F_c$.
\end{lem}
\begin{proof}
Let $\tau$ be the function defined as $\tau(x) = |x|$ for $x \in \real^{3}$. Then, by definition (\ref{f_def}) we can write
\begin{equation}
\label{dpf1}
F_c(\bold{y}) = -\frac{Z}{2}\sum_{1 \le j \le N}\tau(y_j) + \frac{1}{4}\sum_{1\le l<m \le N} \tau(y_l - y_m).
\end{equation}
For each $r=1,\dots,N$, denote by $\nabla_{y_r}$ the three-dimensional gradient in the variable $y_r$. We then have
\begin{equation}
\label{ddpf1}
\nabla_{y_r} F_c(\bold{y}) = -\frac{Z}{2} \nabla \tau(y_r) +  \frac{1}{4}\sum_{\substack{m=1 \\ m \ne r}}^{N} \nabla \tau(y_r - y_m).
\end{equation}
We apply the $D^{\boldsymbol{\sigma}}_{\bold{P}}$-derivative to each of (\ref{dpf1}) and (\ref{ddpf1}), using Lemma \ref{der_F} with $f = \tau$ and $f = \nabla \tau$ respectively. This shows that for $k=0,1$ the functions $D_{\bold{P}}^{\boldsymbol{\sigma}} \nabla^{k} F_c$ exist and are smooth on the set $\Sigma^c_{\boldsymbol{\sigma}}$. Furthermore, there exists $C$, depending on $\boldsymbol{\sigma}$, such that
\begin{align}
\label{ddpf1_ineq}
|D_{\bold{P}}^{\boldsymbol{\sigma}} \nabla^{k} F_c(\bold{y})| &\le Cq_{\boldsymbol{\sigma}}(\bold{y})^{1-k-|\boldsymbol{\sigma}|} \le C\mu_{\boldsymbol{\sigma}}(\bold{y})^{1-k-|\boldsymbol{\sigma}|}
\end{align}
for all $\bold{y} \in \Sigma^c_{\boldsymbol{\sigma}}$. In the second inequality we used that $\mu_{\boldsymbol{\sigma}}(\bold{y}) \le q_{\boldsymbol{\sigma}}(\bold{y})$.
Recall that $\nabla F_s$ is smooth with all derivatives bounded, see (\ref{f2_linf}). Since $F = F_c-F_s$ and $\mu_{\boldsymbol{\sigma}} \le 1$, we have that for each $\boldsymbol{\sigma} \ne 0$ and $k=0,1$, there exists $C$, depending on $\boldsymbol{\sigma}$, such that
\begin{equation}
\label{dd_f}
|D_{\bold{P}}^{\boldsymbol{\sigma}}\nabla^k F(\bold{y})| \le C\mu_{\boldsymbol{\sigma}}(\bold{y})^{1-k-|\boldsymbol{\sigma}|}
\end{equation}
for all $\bold{y} \in \Sigma^c_{\boldsymbol{\sigma}}$. By a straightforward application of Lemma \ref{mu_lemma}, this proves the first bound in (\ref{dd_f_ef}), for $k=0,1$.

\paragraph{}
We proceed to prove the second bound in (\ref{dd_f_ef}). For every $\boldsymbol{\eta} \le \boldsymbol{\sigma}$ we have that $\Sigma_{\boldsymbol{\eta}} \subset \Sigma_{\boldsymbol{\sigma}}$ and therefore $D_{\bold{P}}^{\boldsymbol{\eta}} \nabla^k F$ exists and is smooth in $\Sigma_{\boldsymbol{\sigma}}^c$, for $k=0,1$. We use the chain rule for weak derivatives to show that $D^{\boldsymbol{\sigma}}_{\bold{P}}(e^F)$ exists in $\Sigma_{\boldsymbol{\sigma}}^c$ and is equal to a sum of terms, each of the form
\begin{equation}
\label{dp_ef}
e^F \prod_{1 \le j \le n} D^{\boldsymbol{\beta}^{(j)}}_{\bold{P}}F
\end{equation}
for some $1 \le n \le |\boldsymbol{\sigma}|$ and some collection $0 \ne \boldsymbol{\beta}^{(j)} \in \naturals_0^{3M}$ for $j=1,\dots,n$, where $\boldsymbol{\beta}^{(1)} + \dots + \boldsymbol{\beta}^{(n)} = \boldsymbol{\sigma}$. 
For each $j=1,\dots,n$, we write $\boldsymbol{\beta}^{(j)} = (\beta^{(j)}_1, \dots, \beta^{(j)}_{M})$ to denote the $\naturals_0^3$-components of $\boldsymbol{\beta}^{(j)}$.
Similarly, the weak derivative $D^{\boldsymbol{\sigma}}_{\bold{P}} \nabla(e^F)$ exists in $\Sigma_{\boldsymbol{\sigma}}^c$ and the gradient of the general term (\ref{dp_ef}) is equal to
\begin{equation}
\label{ddp_ef}
e^F \, \nabla F \prod_{1 \le j \le n} D_{\bold{P}}^{\boldsymbol{\beta}^{(j)}}F + e^F \sum_{r=1}^n \Big(D_{\bold{P}}^{\boldsymbol{\beta}^{(r)}} \nabla F\, \prod_{\substack{1\le j \le n \\ j\ne r}} D_{\bold{P}}^{\boldsymbol{\beta}^{(j)}} F\Big).
\end{equation}
We will now find bounds for the above expressions. To do this we will use (\ref{dd_f}) to bound cluster derivatives of $F$. This, along with (\ref{mu_1sig}) and (\ref{lam_prod}) of Lemma \ref{mu_lemma}, allows us to bound the following product
\begin{align*}
\Big|\prod_{j=1}^n D_{\bold{P}}^{\boldsymbol{\beta}^{(j)}} F (\bold{y})\Big| &\le C \prod_{j=1}^n \mu_{\boldsymbol{\beta}^{(j)}}(\bold{y})^{1-|\boldsymbol{\beta}^{(j)}|}\\
&\le C \prod_{j=1}^n  \lambda_{\boldsymbol{\beta}^{(j)}}(\bold{y}) \lambda_{P_1}(\bold{y})^{-|\beta^{(j)}_1|} \dots \lambda_{P_M}(\bold{y})^{-|\beta^{(j)}_M|} \\
&= C \Big(\prod_{j=1}^n  \lambda_{\boldsymbol{\beta}^{(j)}}(\bold{y})\Big) \Big( \prod_{l=1}^M \lambda_{P_l}(\bold{y})^{-|\beta^{(1)}_l|} \dots \lambda_{P_l}(\bold{y})^{-|\beta^{(n)}_l|}\Big) \\
&\le C\lambda_{\boldsymbol{\sigma}}(\bold{y}) \lambda_{P_1}(\bold{y})^{-|\sigma_1|} \dots \lambda_{P_M}(\bold{y})^{-|\sigma_M|}
\end{align*}
where $C$ is some constant, depending on $\boldsymbol{\sigma}$, and all $\bold{y} \in \Sigma^c_{\boldsymbol{\sigma}}$. We used that for each $l=1,\dots,M$, we have $|\beta^{(1)}_l| + \dots + |\beta^{(n)}_l| = |\sigma_l|$ as a consequence of $\boldsymbol{\beta}^{(1)} + \dots + \boldsymbol{\beta}^{(n)} = \boldsymbol{\sigma}$. We now bound the following product in a similar manner using (\ref{dd_f}) and Lemma \ref{mu_lemma}. For any $r=1, \dots, n$,
\begin{align*}
\Big| D^{\boldsymbol{\beta}^{(r)}}_{\bold{P}} \nabla F(\bold{y}) \prod_{\substack{1\le j\le n \\ j \ne r}} D^{\boldsymbol{\beta}^{(j)}}_{\bold{P}} F(\bold{y}) \Big| &\le C \lambda_{P_1}(\bold{y})^{-|\sigma_1|} \dots \lambda_{P_M}(\bold{y})^{-|\sigma_M|}
\end{align*}
for some constant $C$, depending on $\boldsymbol{\sigma}$, and all $\bold{y} \in \Sigma^c_{\boldsymbol{\sigma}}$. Using the fact that $e^F$ and $\nabla F$ are bounded in $\real^{3N}$, it is now straightforward to bound (\ref{dp_ef}) and (\ref{ddp_ef}) appropriately. This completes the proof of (\ref{dd_f_ef}).

\paragraph{}
Finally, we prove the inequality (\ref{dd_f_2}). It can be shown that the following Leibniz rule holds
\begin{equation*}
D_{\bold{P}}^{\boldsymbol{\sigma}} |\nabla F|^2 = \sum_{\boldsymbol{\beta} \le \boldsymbol{\sigma}} {\boldsymbol{\sigma} \choose \boldsymbol{\beta}} D^{\boldsymbol{\beta}}_{\bold{P}} \nabla F \cdot D^{\boldsymbol{\sigma}-\boldsymbol{\beta}}_{\bold{P}} \nabla F
\end{equation*}
in $\Sigma^c_{\boldsymbol{\sigma}}$. For every $\boldsymbol{\beta} \le \boldsymbol{\sigma}$ we can bound $D^{\boldsymbol{\beta}}_{\bold{P}} \nabla F$ by a constant if $\boldsymbol{\beta} = 0$, and by (\ref{dd_f}) otherwise. This gives some constant $C$, depending on $\boldsymbol{\sigma}$, such that
\begin{align*}
\big|D_{\bold{P}}^{\boldsymbol{\sigma}} |\nabla F(\bold{y})|^2\big| &\le C\sum_{\boldsymbol{\beta} \le \boldsymbol{\sigma}}  \mu_{\boldsymbol{\beta}}(\bold{y})^{-|\boldsymbol{\beta}|} \mu_{\boldsymbol{\sigma} - \boldsymbol{\beta}}(\bold{y})^{-|\boldsymbol{\sigma}|+|\boldsymbol{\beta}|}
\end{align*}
for all $\bold{y} \in \Sigma^c_{\boldsymbol{\sigma}}$. The required bound is obtained after an application of (\ref{mu_1sig}) of Lemma \ref{mu_lemma}.

\end{proof}

\paragraph{}
The following result extends the bounds given in Lemma \ref{der_gradF} to give bounds to $L^{\infty}$-norms in balls. It is useful to note that if $\bold{x} \in \Sigma_{\boldsymbol{\sigma}}^c$, for $\boldsymbol{\sigma} \ne 0$, then $B(\bold{x}, \nu\lambda_{\boldsymbol{\sigma}}(\bold{x})) \subset \Sigma_{\boldsymbol{\sigma}}^c$ for all $\nu \in (0,1)$.

\begin{lem}
\label{ab_lem}
For any $|\bsig| \ge 1$ and any $\nu \in (0,1)$ there exists $C$, depending on $\boldsymbol{\sigma}$ and $\nu$, such that for $k=0,1$,
\begin{multline}
\label{dd_f_ef_2}
\norm{D_{\bold{P}}^{\boldsymbol{\sigma}} \nabla^k F}_{L^{\infty}(B(\bold{x}, \nu\lambda_{\boldsymbol{\sigma}}(\bold{x})))},\, \norm{D_{\bold{P}}^{\boldsymbol{\sigma}}  \nabla^k (e^F)}_{L^{\infty}(B(\bold{x}, \nu\lambda_{\boldsymbol{\sigma}}(\bold{x})))}\\ \le C\lambda_{\boldsymbol{\sigma}}(\bold{x})^{1-k} \lambda_{P_1}(\bold{x})^{-|\sigma_1|} \dots \lambda_{P_M}(\bold{x})^{-|\sigma_M|}
\end{multline}
\begin{equation}
\label{dd_f_2_2}
\norm{D_{\bold{P}}^{\boldsymbol{\sigma}} |\nabla F|^2}_{L^{\infty}(B(\bold{x}, \nu\lambda_{\boldsymbol{\sigma}}(\bold{x})))} \le C\lambda_{P_1}(\bold{x})^{-|\sigma_1|} \dots \lambda_{P_M}(\bold{x})^{-|\sigma_M|}
\end{equation}
for all $\bold{x} \in \Sigma_{\boldsymbol{\sigma}}^c$. The bound to the first norm in (\ref{dd_f_ef_2}) also holds when $F$ is replaced by $F_c$.
\end{lem}
\begin{proof}
Take some $\bold{x} \in \Sigma_{\boldsymbol{\sigma}}^c$. By (\ref{lam_lip}), for each $\bold{y} \in B(\bold{x}, \nu\lambda_{\boldsymbol{\sigma}}(\bold{x}))$ we have
\begin{equation*}
|\lambda_{P_j}(\bold{x}) - \lambda_{P_j}(\bold{y})| \le |\bold{x} - \bold{y}| \le \nu \lambda_{\boldsymbol{\sigma}}(\bold{x}) \le \nu \lambda_{P_j}(\bold{x})
\end{equation*}
for each $j=1, \dots, M$ with $\sigma_j \ne 0$. The final inequality above uses the definition of $\lambda_{\boldsymbol{\sigma}}$ in (\ref{eq:lam_alpha}). By rearrangement, we have
\begin{equation}
\label{lam_xy_eq}
(1-\nu)\lambda_{P_j}(\bold{x}) \le \lambda_{P_j}(\bold{y}).
\end{equation}
Therefore,
\begin{equation}
\label{lamlam1}
\lambda_{P_1}(\bold{y})^{-|\sigma_1|} \dots \lambda_{P_M}(\bold{y})^{-|\sigma_M|} \le (1-\nu)^{-|\boldsymbol{\sigma}|} \lambda_{P_1}(\bold{x})^{-|\sigma_1|} \dots \lambda_{P_M}(\bold{x})^{-|\sigma_M|}
\end{equation}
for all $\bold{y} \in B(\bold{x}, \nu\lambda_{\boldsymbol{\sigma}}(\bold{x}))$. We now prove an analogous inequality. First, we see that $\lambda_{\boldsymbol{\sigma}}(\bold{x}) = \lambda_{P_l}(\bold{x})$ for some $1 \le l \le M$ with $\sigma_l \ne 0$ which will depend on the choice of $\bold{x}$. We also note that $\lambda_{\boldsymbol{\sigma}}(\bold{y}) \le \lambda_{P_l}(\bold{y})$ for all $\bold{y} \in \real^{3N}$ which follows from $\sigma_l \ne 0$ and the definition of $\lambda_{\boldsymbol{\sigma}}$. Therefore,
\begin{align}
\nonumber
\lambda_{\boldsymbol{\sigma}}(\bold{y})\lambda_{P_1}(\bold{y})^{-|\sigma_1|} \dots \lambda_{P_M}(\bold{y})^{-|\sigma_M|}
&\le  \lambda_{P_l}(\bold{y})^{1-|\sigma_l|}\prod_{\substack{j=1 \\ j \ne l}}^M \lambda_{P_j}(\bold{y})^{-|\sigma_j|}\\
\nonumber
&\le (1-\nu)^{1-|\boldsymbol{\sigma}|} \lambda_{P_l}(\bold{x})^{1-|\sigma_l|}\prod_{\substack{j=1 \\ j \ne l}}^M \lambda_{P_j}(\bold{x})^{-|\sigma_j|}\\
\label{lamlam2}
&= (1-\nu)^{1-|\boldsymbol{\sigma}|} \lambda_{\boldsymbol{\sigma}}(\bold{x}) \lambda_{P_1}(\bold{x})^{-|\sigma_1|} \dots \lambda_{P_M}(\bold{x})^{-|\sigma_M|}
\end{align}
for all $\bold{y} \in B(\bold{x}, \nu\lambda_{\boldsymbol{\sigma}}(\bold{x}))$. In the second step we applied (\ref{lam_xy_eq}). The required bounds then arise from Lemma \ref{der_gradF} followed by either (\ref{lamlam1}) or (\ref{lamlam2}) as appropriate.

\end{proof}

\subsection{Cluster derivatives of $\phi$ and $\psi$}

\paragraph{}
We introduce the following notation. For cluster set $\bold{P} = (P_1, \dots, P_M)$ and $\boldsymbol{\alpha} \in \naturals_0^{3M}$ define
\begin{equation*}
\phi_{\boldsymbol{\alpha}} = D_{\bold{P}}^{\boldsymbol{\alpha}}\phi.
\end{equation*}
We have previously obtained estimates for the coefficients, and their derivatives, in equation (\ref{phi_pde_eq}). By taking cluster derivatives of this equation, we obtain elliptic equations whose weak solutions are $\phi_{\boldsymbol{\alpha}}$. By a process of rescaling, using a method developed in \cite{coulomb_estimates}, we will apply elliptic regularity to give quantative bounds to the $L^{\infty}$-norms of $\phi_{\boldsymbol{\alpha}}$ and $\nabla \phi_{\boldsymbol{\alpha}}$. From here it is straightforward to obtain corresponding bounds for $D_{\bold{P}}^{\boldsymbol{\alpha}}\psi$ and $D_{\bold{P}}^{\boldsymbol{\alpha}}\nabla \psi$.

\begin{lem}
\label{diff_pde_lem}
For every $|\bal| \ge 1$ we have $\phi_{\boldsymbol{\alpha}}$ is a weak solution to
\begin{multline}
\label{phi_alpha_eq}
-\Delta \phi_{\boldsymbol{\alpha}} - 2\nabla F \cdot \nabla \phi_{\boldsymbol{\alpha}} + (\Delta F_s - |\nabla F|^2 - E) \phi_{\boldsymbol{\alpha}} \\= \sum_{\substack{\boldsymbol{\sigma} \le \boldsymbol{\alpha} \\ \boldsymbol{\sigma} \ne 0}} {\boldsymbol{\alpha} \choose \boldsymbol{\sigma}} \big( 2 D_{\bold{P}}^{\boldsymbol{\sigma}} \nabla F \cdot \nabla \phi_{\boldsymbol{\alpha}-\boldsymbol{\sigma}} - D_{\bold{P}}^{\boldsymbol{\sigma}}(\Delta F_s - |\nabla F|^2) \phi_{\boldsymbol{\alpha}-\boldsymbol{\sigma}} \big) =: g_{\bal}
\end{multline}
in $\Sigma^c_{\bal}$, and therefore $\phi_{\boldsymbol{\alpha}} \in C^1(\Sigma_{\bal}^c) \cap H^2_{loc}(\Sigma_{\bal}^c)$.
\end{lem}
\begin{proof}
We prove by induction, starting with $|\boldsymbol{\alpha}| = 1$. Let $\mathcal{L}(\cdot,\cdot)$ be the bilinear form corresponding to the operator acting on $\phi_{\boldsymbol{\alpha}}$ on the left-hand side of (\ref{phi_alpha_eq}), as defined in (\ref{eq:pde}). Since $\phi \in H^{2}_{loc}(\real^{3N})$ we have $\phi_{\boldsymbol{\alpha}} \in H^1_{loc}(\real^{3N})$. We use integration by parts to show that, for any $\chi \in C^{\infty}_c(\Sigma^c_{\bal})$,
\begin{equation}
\label{eq:l_phi}
\mathcal{L}(\phi_{\boldsymbol{\alpha}} , \chi) = -\mathcal{L}(\phi, D_{\bold{P}}^{\boldsymbol{\alpha}}\chi) +  \int_{\Sigma^c_{\bal}} g_{\boldsymbol{\alpha}} \chi\,d\bold{x}.
\end{equation}
Now, $\mathcal{L}(\phi, D_{\bold{P}}^{\boldsymbol{\alpha}}\chi) = 0$ since $\phi$ is a weak solution to (\ref{phi_pde_eq}), hence (\ref{phi_alpha_eq}) holds. By Lemma \ref{der_gradF} and that $\phi \in C^1(\real^{3N})$ we see that $g_{\bal} \in L_{loc}^{\infty}(\Sigma^c_{\bal})$. Hence, by Theorems \ref{thm:c1}, $\phi_{\bal} \in C^1(\Sigma^c_{\bal}) \cap H_{loc}^2(\Sigma^c_{\bal})$

Assume the hypothesis holds for all multiindices $1 \le |\boldsymbol{\alpha}| \le k-1$ for some $k\ge 2$. Take some arbitrary multiindex $|\boldsymbol{\alpha}| = k$. We will prove the induction hypothesis for $\bal$. First, we state the useful fact that for any $\bsig \le \bal$ we have $\Sigma_{\bsig} \subset \Sigma_{\bal}$. Take some $\boldsymbol{\eta} \le \bal$ with $|\boldsymbol{\eta}| = 1$. Notice, then, that $\phi_{\bal-\boldsymbol{\eta}} \in H^2_{loc}(\Sigma^c_{\bal})$ by the induction hypothesis, and hence $\phi_{\boldsymbol{\alpha}} \in H^1_{loc}(\Sigma^c_{\bal})$. This allows $\mathcal{L}(\phi_{\bal}, \chi)$ to be defined for any $\chi \in C_c^{\infty}(\Sigma^c_{\bal})$. Integration by parts is then used on the $D^{\boldsymbol{\eta}}_{\bold{P}}$-derivative, in a similar way to (\ref{eq:l_phi}). Along with the induction hypothesis and further applications of integration by parts, we find that (\ref{phi_alpha_eq}) holds for $\phi_{\bal}$ as a weak solution. Regularity for $\phi_{\bal}$ is shown in a similar way to the $|\bal|=1$ case. Indeed, the induction hypothesis can be used to show $g_{\bal} \in L^{\infty}_{loc}(\Sigma^c_{\bal})$.

\end{proof}

\paragraph{}
We now apply the $C^1$-elliptic regularity estimates of Theorem \ref{thm:c1}, via a scaling procedure, to the equations (\ref{phi_alpha_eq}).
\begin{lem}
\label{er_phi_lem}
For all $|\boldsymbol{\alpha}| \ge 1$ and $0<r<R<1$ there exists $C$, dependent only on $E, Z, N, r$ and $R$, such that
\begin{multline}
\label{er_phi_eq}
\norm{\nabla \phi_{\boldsymbol{\alpha}}}_{L^{\infty}(B(\bold{x}, r\lambda_{\boldsymbol{\alpha}}(\bold{x})))} \\\le C \lambda_{\boldsymbol{\alpha}}(\bold{x})^{-1} \big(\norm{\phi_{\boldsymbol{\alpha}}}_{L^{\infty}(B(\bold{x}, R\lambda_{\boldsymbol{\alpha}}(\bold{x})))} + \lambda_{\boldsymbol{\alpha}}(\bold{x})^2 \norm{g_{\boldsymbol{\alpha}}}_{L^{\infty}(B(\bold{x}, R\lambda_{\boldsymbol{\alpha}}(\bold{x})))} \big)
\end{multline}
for all $\bold{x} \in \Sigma_{\boldsymbol{\alpha}}^c$. The function $g_{\boldsymbol{\alpha}}$ was given by (\ref{phi_alpha_eq}).
\end{lem}
\begin{proof}
Fix any $\bold{x} \in \Sigma_{\boldsymbol{\alpha}}^c$ and denote $\lambda = \lambda_{\boldsymbol{\alpha}}(\bold{x})$. The proof proceeds via a rescaling. Define
\begin{equation*}
\begin{split}
w(\bold{y}) &= \phi_{\boldsymbol{\alpha}}(\bold{x} + \lambda \bold{y})\\
\bold{c}(\bold{y}) &= -2\nabla F(\bold{x} + \lambda \bold{y})\\
d(\bold{y}) &= \Delta F_s(\bold{x} + \lambda \bold{y}) - |\nabla F (\bold{x} + \lambda \bold{y})|^2 - E \\
f(\bold{y}) &= g_{\boldsymbol{\alpha}}(\bold{x} + \lambda \bold{y})
\end{split}
\end{equation*}
for all $\bold{y} \in B(0,1)$. Then by Lemma \ref{diff_pde_lem} we have that $w$ is a weak solution to
\begin{equation}
\label{rescaled_pde}
-\Delta w + \lambda \bold{c} \cdot \nabla w + \lambda^2 d w = \lambda^2 f
\end{equation}
in $B(0,1)$. By (\ref{f_linf}) and (\ref{f2_linf}), and that $\lambda \le 1$ by definition, we obtain
\begin{equation*}
\lambda \norm{\bold{c}}_{L^{\infty}(B(0,1))} + \lambda^2 \norm{d}_{L^{\infty}(B(0,1))} \le K
\end{equation*}
for some $K$, dependent only on $E, Z$ and $N$ and in particular is independent of our choice of $\bold{x}$. Therefore, by Theorem \ref{thm:c1}, with $\theta=0$, we get some $C$, dependent only on $N, K, r$ and $R$, such that
\begin{equation}
\label{w_c1_eq}
\norm{w}_{C^1(\overline{B(0,r)})} \le C(\norm{w}_{L^{\infty}(B(0,R))} + \lambda^2 \norm{f}_{L^{\infty}(B(0,R))}).
\end{equation}
Finally, we use $\nabla w(\bold{y}) = \lambda \nabla \phi_{\boldsymbol{\alpha}}(\bold{x} + \lambda \bold{y})$ by the chain rule, and rewrite (\ref{w_c1_eq}) to give (\ref{er_phi_eq}).

\end{proof}

The following proposition uses induction to write estimates for both $\nabla \phi_{\bal}$ and $\phi_{\bal}$ with a bound involving only the zeroth and first-order derivatives of $\phi$. Corollary \ref{cor:d2phi} of the appendix is used in the proof to improve the regularity of the second derivative of $\phi$, and this benefit is passed on through induction to all higher orders of derivative. The function $f_{\infty}(\,\cdot\,;R;\phi)$ was defined in (\ref{eq:f_inf_p}).

\begin{prop}
\label{dd_phi_prop}
For any $|\bal| \ge 1$, any $0<r<R<1$ and any $0 \le b < 1$ there exists $C$, dependent on $\bal, r, R$ and $b$, such that for $k=0,1$, with $k+|\bal| \ge 2$,
\begin{align*}
\norm{\nabla^k \phi_{\boldsymbol{\alpha}}}_{L^{\infty}(B(\bold{x}, r\lambda_{\boldsymbol{\alpha}}(\bold{x})))} &\le C \lambda_{\boldsymbol{\alpha}}(\bold{x})^{1-k} \lambda_{P_1}(\bold{x})^{-|\alpha_1|} \dots \lambda_{P_M}(\bold{x})^{-|\alpha_M|} \mu_{\bal}(\bold{x})^{b} f_{\infty}(\bold{x}; R; \phi )
\end{align*}
for all $\bold{x} \in \Sigma_{\boldsymbol{\alpha}}^c$. The inequality also holds for $|\bal|=1$ and $k=0$ if we take $b=0$.
\end{prop}
\begin{proof}
Suppose $|\bal| = 1$ and $k=1$. Let $l$ be the index $1 \le l \le M$ such that $\alpha_l \ne 0$. Then $\mu_{\bal} = \lambda_{P_l}$ and $\Sigma_{\bal} = \Sigma_{P_l}$. 
The bound then follows from Corollary \ref{cor:d2phi} with $P=P_l$. Indeed, for each $s>0$ we can find some constant $C_s$ such that $\log(t) \le C_st^s$ for all $t >0$. We use this fact along with the simple inequality $\lambda_{P_l} \le \nu_{P_l}$.

For $|\bal| \ge 2$ we prove by induction. 
Assume the hypothesis holds for all multiindices $\boldsymbol{\alpha}$ with $1 \le |\boldsymbol{\alpha}| \le m-1$ for some $m \ge 2$. Take any $\boldsymbol{\alpha}$ with $|\boldsymbol{\alpha}| = m$. We prove the $k=0$ and $k=1$ cases in turn. It is useful, here, to state the following fact: for all $\boldsymbol{\sigma} \le \boldsymbol{\alpha}$ we have $\Sigma_{\boldsymbol{\sigma}} \subset \Sigma_{\boldsymbol{\alpha}}$ and hence $\lambda_{\boldsymbol{\alpha}}(\bold{x}) \le \lambda_{\boldsymbol{\sigma}}(\bold{x})$.

The $k=0$ case is a straightforward application of the induction hypothesis and is described as follows. 
Firstly, for each $\bx \in \Sigma_{\bal}^c$ it is clear that there exists some $1 \le l \le M$ such that $\lambda_{\bal}(\bold{x}) = \lambda_{P_l}(\bold{x})$ and where $\alpha_l \ne 0$. Therefore we can find some $\boldsymbol{\eta} \le \bal$ with $|\boldsymbol{\eta}| = 1$ and $\eta_j = 0$ for each $j \ne l$ where $1 \le j \le M$. Now, from the definition of cluster derivatives (\ref{eq:cd2}),
\begin{align}
\label{ua_bound}
\norm{\phi_{\bal}}_{L^{\infty}(B(\bold{x}, r\lambda_{\bal}(\bold{x})))} &\le N\norm{\nabla \phi_{\bal-\boldsymbol{\eta}}}_{L^{\infty}(B(\bold{x}, r\lambda_{\boldsymbol{\alpha}}(\bold{x})))}.
\end{align}
We can then use the induction assumption to show existence of $C$ such that
\begin{multline}
\label{du_ae_bound}
\norm{\nabla \phi_{\bal-\boldsymbol{\eta}}}_{L^{\infty}(B(\bold{x}, r\lambda_{\boldsymbol{\alpha}}(\bold{x})))} \\\le C \lambda_{P_l}(\bold{x}) \lambda_{P_1}(\bold{x})^{-|\alpha_1|} \dots \lambda_{P_M}(\bold{x})^{-|\alpha_M|}\mu_{\bal-\boldsymbol{\eta}}(\bold{x})^{b} f_{\infty}(\bold{x}; R; \phi),
\end{multline}
where we can then use $\lambda_{P_l}(\bold{x}) = \lambda_{\boldsymbol{\alpha}}(\bold{x}) $ and  $\mu_{\bal-\boldsymbol{\eta}}(\bold{x}) \le \mu_{\bal}(\bold{x})$ to complete the required bound for our choice of $\bal$ in the case of $k=0$. 

The $k=1$ case follows from the induction hypothesis and Lemma \ref{er_phi_lem}. Let $r' = (r+R)/2$.
Firstly, by the definition of $g_{\bal}$, (\ref{phi_alpha_eq}), along with (\ref{f2_linf}) and Lemma \ref{ab_lem} we can find $C$, depending on $\bal$ and $r'$, such that
\begin{multline*}
\norm{g_{\bal}}_{L^{\infty}(B(\bold{x}, r'\lambda_{\boldsymbol{\alpha}}(\bold{x})))}
\le C \sum_{\substack{\boldsymbol{\sigma} \le \boldsymbol{\alpha} \\ \boldsymbol{\sigma} \ne 0}} \lambda_{P_1}(\bold{x})^{-|\sigma_1|} \dots \lambda_{P_M}(\bold{x})^{-|\sigma_M|}\big(\norm{\phi_{\boldsymbol{\alpha}-\boldsymbol{\sigma}}}_{L^{\infty}(B(\bold{x}, r'\lambda_{\boldsymbol{\alpha}}(\bold{x})))}\\+ \norm{\nabla \phi_{\boldsymbol{\alpha}-\boldsymbol{\sigma}}}_{L^{\infty}(B(\bold{x}, r'\lambda_{\boldsymbol{\alpha}}(\bold{x})))}\big).
\end{multline*}
It follows by the induction hypothesis, using $b=0$, along with the definition of $f_{\infty}$, (\ref{eq:f_inf}), that we can find some $C$, depending on $\bal, r'$ and $R$, such that
\begin{equation}
\label{eq:ga_linf}
\norm{g_{\boldsymbol{\alpha}}}_{L^{\infty}(B(\bold{x}, r'\lambda_{\boldsymbol{\alpha}}(\bold{x})))}
\le C \lambda_{P_1}(\bold{x})^{-|\alpha_1|} \dots \lambda_{P_M}(\bold{x})^{-|\alpha_M|} f_{\infty}(\bold{x}; R; \phi)
\end{equation}
for all $\bold{x} \in \Sigma_{\bal}^c$. Next, we apply Lemma \ref{er_phi_lem} to obtain some constant $C$, depending on $r$ and $r'$, such that
\begin{align*}
\norm{\nabla \phi_{\boldsymbol{\alpha}}}_{L^{\infty}(B(\bold{x}, r\lambda_{\boldsymbol{\alpha}}(\bold{x})))} &\le C\big( \lambda_{\bal}(\bold{x})^{-1}\norm{\phi_{\bal}}_{L^{\infty}(B(\bold{x}, r'\lambda_{\bal}(\bold{x})))} + \lambda_{\boldsymbol{\alpha}}(\bold{x}) \norm{g_{\boldsymbol{\alpha}}}_{L^{\infty}(B(\bold{x}, r'\lambda_{\boldsymbol{\alpha}}(\bold{x})))} \big).
\end{align*}
We use the $k=0$ case, proven above, to bound the $L^{\infty}$-norm of $\phi_{\bal}$ in the first term on the right-hand side of the above inequality. To the second term, we apply (\ref{eq:ga_linf}) and the simple inequality $\lambda_{\bal}(\bold{x}) \le \mu_{\bal}(\bold{x})$. Together, these give the required bound for $k=1$ and our choice of $\bal$. This completes the induction.

\end{proof}

We now obtain bounds to the cluster derivatives of the eigenfunction $\psi$ using those for $\phi$ in the above proposition.
\begin{proof}[Proof of Theorem \ref{thm:ddpsi}]
It is clear that (\ref{eq:ddpsi}) holds when $\bal = 0$. Therefore, consider $|\bal| \ge 1$. We first prove (\ref{eq:ddpsi}) for $k=0$. Take $\bold{x} \in \Sigma_{\bal}^c$, then by the Leibniz rule for cluster derivatives we have
\begin{align}
\label{eq:psi_leibniz}
D_{\bold{P}}^{\boldsymbol{\alpha}}\psi &= \sum_{\boldsymbol{\beta} \le \boldsymbol{\alpha}} {\boldsymbol{\alpha} \choose \boldsymbol{\beta} } D_{\bold{P}}^{\boldsymbol{\beta}}\big(e^F \big) \phi_{\boldsymbol{\alpha}-\boldsymbol{\beta}}
\end{align}
in $B(\bold{x}, r\lambda_{\bal}(\bold{x}))$. Now, for each $\boldsymbol{\beta} \le \boldsymbol{\alpha}$ there exists some $C$, independent of $\bx$, such that
\begin{equation*}
\norm{D_{\bold{P}}^{\boldsymbol{\beta}}\big(e^F \big) \phi_{\boldsymbol{\alpha}-\boldsymbol{\beta}}}_{L^{\infty}(B(\bold{x}, r\lambda_{\boldsymbol{\alpha}}(\bold{x})))} \le C\lambda_{\bal}(\bold{x}) \lambda_{P_1}(\bold{x})^{-|\alpha_1|} \dots \lambda_{P_M}(\bold{x})^{-|\alpha_M|} f_{\infty}(\bold{x}; R; \phi).
\end{equation*}
To prove this, we use Lemma \ref{ab_lem} and Proposition \ref{dd_phi_prop} with $b=0$. In addition, if $\bbeta = 0$ we use (\ref{f_linf}) and if $\bbeta = \bal$ we use the definition (\ref{eq:f_inf}). If $\bbeta \notin \{0, \bal \}$ we use the inequality $\lambda_{\boldsymbol{\beta}}(\bold{x}) \lambda_{\boldsymbol{\alpha}-\boldsymbol{\beta}}(\bold{x}) \le \lambda_{\boldsymbol{\alpha}}(\bold{x})$, which is obtained from Lemma \ref{mu_lemma}. The required bound (\ref{eq:ddpsi}) for $k=0$ then follows from (\ref{eq:psi_leibniz}) and (\ref{eq:phi_er})-(\ref{eq:phi_psi}).

Take $\bx \in \Sigma^c_{\bal}$. By the definition $F= F_c - F_s$, the equality $\nabla \psi = \psi \nabla F + e^F \nabla \phi$ and the Leibniz rule, we have (\ref{eq:gp_def}) with
\begin{align}
\label{eq:gp_def_full}
G^{\bal}_{\bold{P}} &= \sum_{\substack{\bbeta \le \bal \\ |\bbeta| \ge 1}} {\boldsymbol{\alpha} \choose \boldsymbol{\beta} } D_{\bold{P}}^{\bbeta}\psi\, D_{\bold{P}}^{\bal - \bbeta}\nabla F +
\sum_{\bbeta \le \bal} {\boldsymbol{\alpha} \choose \boldsymbol{\beta}} \nabla \phi_{\bbeta}\, D_{\bold{P}}^{\bal - \bbeta} \big(e^F \big) - \psi\, D_{\bold{P}}^{\bal}\nabla F_s.
\end{align}
We bound each term in $G^{\bal}_{\bold{P}}$ as in (\ref{eq:gp_bnd}).
Take any $\bbeta \le \bal$ with $|\bbeta| \ge 1$. Notice that $\Sigma_{\bbeta} \subset \Sigma_{\bal}$ and hence $\lambda_{\bal}(\bx) \le \lambda_{\bbeta}(\bx)$. Now, there exists $C$, independent of $\bx$, such that
\begin{align}
\label{eq:gp_term1}
\norm{D_{\bold{P}}^{\bbeta}\psi\, D_{\bold{P}}^{\bal -\bbeta} \nabla F }_{L^{\infty}(B(\bold{x}, r\lambda_{\boldsymbol{\alpha}}(\bold{x})))} &\le C\lambda_{\bbeta}(\bold{x}) \lambda_{P_1}(\bold{x})^{-|\alpha_1|} \dots \lambda_{P_M}(\bold{x})^{-|\alpha_M|} f_{\infty}(\bold{x}; R).
\end{align}
To prove this inequality, we bound derivatives of $F$ using Lemma \ref{ab_lem} if $\bbeta \ne \bal$ and (\ref{f_linf}) if $\bbeta = \bal$. The derivatives $D_{\bold{P}}^{\bbeta}\psi$ are then bounded using (\ref{eq:ddpsi}) with $k=0$, which was proven above. The right-hand side of (\ref{eq:gp_term1}) can be bounded as in (\ref{eq:gp_bnd}) by using the simple bound $\lambda_{\bbeta} \le \mu_{\bal}$.  Next, still considering $|\bbeta| \ge 1$ we can find $C$, independent of $\bx$, such that
\begin{align*}
\norm{ \nabla \phi_{\bbeta}\, D_{\bold{P}}^{\bal - \bbeta} \big(e^F \big) }_{L^{\infty}(B(\bold{x}, r\lambda_{\boldsymbol{\alpha}}(\bold{x})))} &\le C\mu_{\bbeta}(\bold{x})^b \lambda_{\bal - \bbeta}(\bold{x}) \lambda_{P_1}(\bold{x})^{-|\alpha_1|} \dots \lambda_{P_M}(\bold{x})^{-|\alpha_M|} f_{\infty}(\bold{x}; R; \phi).
\end{align*}
This is proven using Proposition \ref{dd_phi_prop} to bound the derivatives $\nabla \phi_{\bbeta}$, and Lemma \ref{ab_lem} and (\ref{f_linf}) to bound derivatives of $e^F$. The right-hand side of the above inequality can be bounded as in (\ref{eq:gp_bnd}) after use of (\ref{eq:phi_er})-(\ref{eq:phi_psi}) and the inequalities $\mu_{\bbeta} \le \mu_{\bal}$ and $\lambda_{\bal - \bbeta} \le 1$. 
In addition, by Lemma \ref{ab_lem} there exists $C$ such that
\begin{align*}
\norm{ \nabla \phi\, D_{\bold{P}}^{\bal} \big(e^F \big) }_{L^{\infty}(B(\bold{x}, r\lambda_{\boldsymbol{\alpha}}(\bold{x})))} &\le C\lambda_{\bal}(\bold{x}) \lambda_{P_1}(\bold{x})^{-|\alpha_1|} \dots \lambda_{P_M}(\bold{x})^{-|\alpha_M|} f_{\infty}(\bold{x}; R; \phi).
\end{align*}
Use of (\ref{eq:phi_er})-(\ref{eq:phi_psi}), as before, and the inequality $\lambda_{\bal} \le \mu_{\bal}$ give the correct bound. Finally, the last term in (\ref{eq:gp_def_full}) is readily bounded appropriately using (\ref{f2_linf}). In view of (\ref{eq:gp_def_full}), this completes the proof of (\ref{eq:gp_bnd}).

To prove (\ref{eq:ddpsi}) for $k=1$ we use (\ref{eq:gp_def}), (\ref{eq:gp_bnd}) and the following inequality,
\begin{equation*}
\norm{\psi D_{\bold{P}}^{\bal}\nabla F_c}_{L^{\infty}(B(\bold{x}, r\lambda_{\boldsymbol{\alpha}}(\bold{x})))} \le C\lambda_{P_1}(\bold{x})^{-|\alpha_1|} \dots \lambda_{P_M}(\bold{x})^{-|\alpha_M|} f_{\infty}(\bold{x}; R)
\end{equation*}
which holds for some $C$ as a direct consequence of Lemma \ref{ab_lem}.

\end{proof}

\section{Cutoffs}
\label{chpt:3}
In the proof of Theorem \ref{thm:5}, partial derivatives of the density matrix will become cluster derivatives of $\psi$ under the integral. In this section we introduce cutoff functions which will included in the density matrix. These can be made to form a partition of unity, and are present to ensure that the cluster derivatives of $\psi$ are supported only away from their singularities.

The cutoffs introduced in this section are an extension of similar cutoffs used in \cite{density_analytic}, \cite{hearn_sob} and \cite{hearn_sob2}. They involve the variables (i.e. ``particles'') $x,y$ and $x_j$, $2 \le j \le N$. For a given cutoff, we will define corresponding clusters $P$ and $S$. Particles in $P$ (or $S$) will be close to $x$ (or $y$, respectively) on the support of the cutoff. The cluster $P$ (or $S$) will eventually faciliate taking $x$- (or $y$-)derivatives of the density matrix.

We will also need to take derivatives of the density matrix in the variable $x+y$. For this, it will be natural to consider a larger cluster $Q$ which contains both $P$ and $S$, but can include other particles. The particles in $Q$ are held close to both $x$ and $y$, but potentially looser than how particles in $P$ and $S$ are held to $x$ and $y$ respectively. 



\subsection{Definition of $\Phi$}

To begin, take some $\xi \in C_c^{\infty}(\real)$, $0 \le \xi(s) \le 1$, with
\begin{align}
\label{eq:chi_def}
\xi(s) &=
\begin{cases}
1 &\text{ if } |s| \le 1\\
0 &\text{ if } |s| \ge 2.
\end{cases}
\end{align}
for $s \in \real$. For each $t>0$ we can define the following two \textit{cutoff factors} $\zeta_{t} = C_c^{\infty}(\real^3)$ and $\theta_{t} \in C^{\infty}(\real^3)$ by
\begin{equation}
\label{cutoff_factors}
\zeta_t(z) = \xi\Big(\frac{4N|z|}{t}\Big), \qquad \theta_t(z) = 1-\zeta_{t}(z), \qquad z \in \real^3.
\end{equation}
We have the following \textit{support criteria} for cutoff factors. For any $z \in \real^3$ and $t>0$,
\begin{itemize}
\item if $\zeta_t(z) \ne 0$ then $|z| < (2N)^{-1}t$,
\item if $\theta_t(z) \ne 0$ then $|z| > (4N)^{-1}t$.
\end{itemize}

Take any $\delta,\epsilon$ with $0<2\delta \le \epsilon$. We define our cutoff, which depends on $\delta$ and $\epsilon$ as parameters, as a function $\Phi = \Phi_{\delta,\epsilon}(x, y,\bold{\hat{x}})$ given by
\begin{equation}
\label{eq:pphi_def}
\Phi_{\delta,\epsilon}(x, y,\bold{\hat{x}}) = \prod_{2 \le j \le N} g^{(1)}_{j}(x - x_j) \prod_{2 \le j \le N}  g^{(2)}_{j}(y -x_j) \prod_{2 \le k<l \le N} f_{k,l}(x_k - x_l)
\end{equation}
for $x,y \in \real^3$ and $\bold{\hat{x}} \in \real^{3N-3}$, and where $g^{(1)}_j, g^{(2)}_j, f_{k,l} \in \{\zeta_{\delta}, \theta_{\delta}\zeta_{\epsilon}, \theta_{\epsilon} \}$ for $2 \le j \le N$ and $2 \le k<l \le N$. In addition, for $2\le k<l \le N$ we define $f_{l,k} = f_{k,l}$.

The idea is that on the support of $\Phi$ we have upper and/or lower bounds to the distances between various pairs of particles. If we take the particles $x_j$ and $x_k$, for example, we have that $\zeta_{\delta}(x_j - x_k)$ holds the two particles close relative to a distance $\delta$,
$\theta_{\epsilon}(x_j - x_k)$ holds them apart relative to a distance $\epsilon$, and $(\theta_{\delta}\zeta_{\epsilon})(x_j - x_k)$ keeps them within an intermediate distance apart.

The cutoffs, $\Phi$, defined in (\ref{eq:pphi_def}) form a partition of unity as shown in the following lemma.
\begin{lem}
\label{pou}
There exists a collection $\{\Phi^{(j)} \}_{j=1}^{J}$ of cutoffs (\ref{eq:pphi_def}), with integer $J$ depending only on $N$, such that whenever $0< 2\delta \le \epsilon$ we have
\begin{equation*}
\sum_{j=1}^{J} \Phi^{(j)}(x,y,\bold{\hat{x}}) = 1
\end{equation*}
for all $x,y \in \real^3$, $\bold{\hat{x}} \in \real^{3N-3}$.
\end{lem}
\begin{proof}
From the definitions (\ref{cutoff_factors}) and support criteria it is straightforward to verify that
\begin{align}
\theta_{\delta}(z)\theta_{\epsilon}(z) = \theta_{\epsilon}(z), \qquad
\zeta_{\delta}(z)\theta_{\epsilon}(z) = 0, \qquad
\zeta_{\delta}(z)\zeta_{\epsilon}(z) = \zeta_{\delta}(z)
\end{align}
for all $z \in \real^3$. It follows that on $\real^3$,
\begin{equation}
1 \equiv \big(\zeta_{\delta} + \theta_{\delta}\big) \big(\zeta_{\epsilon} + \theta_{\epsilon}\big) = \zeta_{\delta} + \theta_{\delta}\zeta_{\epsilon} + \theta_{\epsilon}.
\end{equation}

\paragraph{}
For any $A,B \subset \{2, \dots, N \}$ with $A \cap B = \emptyset$ we define
\begin{equation*}
\tau_{A,B}(\bold{x}) = \prod_{j \in A} \zeta_{\delta}(x_1-x_j) \prod_{j \in B}(\theta_{\delta}\zeta_{\epsilon})(x_1-x_j) \prod_{j\in \{2,\dots,N\}\backslash(A\cup B)} \theta_{\epsilon}(x_1-x_j)
\end{equation*}
for $\bold{x} = (x_1, \dots, x_N) \in \real^{3N}$. Therefore,
\begin{equation*}
\sum_{\substack{A \subset \{2,\dots,N\} \\ B \subset \{2,\dots,N \}\backslash A}} \tau_{A,B}(\bold{x}) = \prod_{2 \le j \le N} \big\{\zeta_{\delta}(x_1-x_j) + (\theta_{\delta}\zeta_{\epsilon})(x_1-x_j) + \theta_{\epsilon}(x_1-x_j) \big\} = 1
\end{equation*}
for all $\bold{x} \in \real^{3N}$.
Let $\Xi = \{(j,k) : 2 \le j <k \le N \}$. For each $Y, Z \subset \Xi$ with $Y \cap Z = \emptyset$ we define
\begin{equation*}
T_{Y,Z}(\bold{\hat{x}}) =  \prod_{(j,k) \in Y} \zeta_{\delta}(x_j-x_k) \prod_{(j,k) \in Z}(\theta_{\delta}\zeta_{\epsilon})(x_j-x_k) \prod_{(j,k) \,\in\, \Xi\backslash(Y\cup Z)} \theta_{\epsilon}(x_j-x_k)
\end{equation*}
for $\bold{\hat{x}} = (x_2, \dots, x_N) \in \real^{3N-3}$. Therefore,
\begin{equation*}
\sum_{\substack{Y \,\subset\, \Xi \\ Z \,\subset\, \Xi\backslash Y}} T_{Y,Z}(\bold{\hat{x}}) = \prod_{2\le j<k \le N} \big(\zeta_{\delta}(x_j-x_k) + (\theta_{\delta}\zeta_{\epsilon})(x_j-x_k) + \theta_{\epsilon}(x_j-x_k) \big) = 1
\end{equation*}
for all $\bold{\hat{x}} \in \real^{3N-3}$. Overall, for all $x,y \in \real^3$, $\bold{\hat{x}} \in \real^{3N-3}$,
\begin{equation*}
\sum_{\substack{A \subset \{2,\dots,N\} \\ B \subset \{2,\dots,N\}\backslash A}} \quad \sum_{\substack{C \subset \{2,\dots,N\} \\ D \subset \{2,\dots,N\}\backslash C}} \quad \sum_{\substack{Y \,\subset\, \Xi \\ Z \,\subset\, \Xi\backslash Y}}   \tau_{A,B}(x, \bold{\hat{x}})  \tau_{C,D}(y, \bold{\hat{x}})  T_{Y,Z}(\bold{\hat{x}}) = 1,
\end{equation*}
which is a sum of cutoffs of the form (\ref{eq:pphi_def}).

\end{proof}

\subsection{Clusters corresponding to $\Phi$}
\label{phi_clusters}
To each $\Phi$ we now introduce three corresponding clusters. To define these, we first introduce index sets based solely on the choice of cutoff factors $g^{(1)}_j, g^{(2)}_j$ and $f_{k,l}$ in the definition (\ref{eq:pphi_def}) of $\Phi$. The index sets and clusters are therefore not dependent on $x,y,\bold{\hat{x}}$ nor on the scaling parameters $\delta$ and $\epsilon$.

We define the index set $L \subset \{(j,k) \in \{1,\dots,N\}^2 : j\ne k \}$ as follows.
\begin{itemize}
\item We have $(j,k) \in L$ if $f_{j,k} \ne \theta_{\epsilon}$ for $j,k =2,\dots,N$. Also $(1,j), (j,1) \in L$ if $(g^{(1)}_j, g^{(2)}_j) \ne (\theta_{\epsilon},\theta_{\epsilon})$ for $j=2,\dots,N$.
\end{itemize}
Furthermore, we define two more index sets $J,K \subset \{(j,k) \in \{1,\dots,N\}^2 : j\ne k \}$ as follows.
\begin{itemize}
\item We have $(j,k) \in J$ if $f_{j,k} = \zeta_{\delta}$ for $j,k =2,\dots,N$. Also $(1,j), (j,1) \in J$ if $g^{(1)}_j = \zeta_{\delta}$ for $j=2,\dots,N$.
\item We have $(j,k) \in K$ if $f_{j,k} = \zeta_{\delta}$ for $j,k =2,\dots,N$. Also $(1,j), (j,1) \in K$ if $g^{(2)}_j = \zeta_{\delta}$ for $j=2,\dots,N$.
\end{itemize}
These index sets obey $J,K \subset L$.
For an arbitrary index set $I \subset \{(j,k) \in \{1,\dots,N\}^2 : j\ne k \}$ we say that two indices $j,k \in \{1,\dots,N\}$ are $I$-\textit{linked} if either $j=k$, or $(j,k) \in I$, or there exist pairwise distinct indices $j_1, \dots, j_s$ for $1 \le s \le N-2$, all distinct from $j$ and $k$ such that $(j,j_1), (j_1, j_2),\dots, (j_s, k) \in I$.

\paragraph{}
The cluster $Q=Q(\Phi)$ is defined as the set of all indices $L$-linked to $1$. The cluster $P = P(\Phi)$ is defined as the set of all indices $J$-linked to $1$. The cluster $S=S(\Phi)$ is defined as the set of all indices $K$-linked to $1$. Since $J,K \subset L$ we see that $P, S \subset Q$.

\subsection{Support of $\Phi$}
The following lemma shows that on the support of $\Phi$, the cluster $P^* := P\backslash\{1\}$ represents a set of particles whose positions $x_j$, $j \in P^*$, are close to $x$, and the cluster $S^* := S\backslash\{1\}$ represents a set of particles whose positions $x_j$, $j \in S^*$, are close to $y$. In both cases this closeness is with respect to the parameter $\delta$. On the support of $\Phi$, the cluster $Q^* := Q\backslash\{1\}$ (recall $P,S \subset Q$) represents a set of particles whose positions $x_j$, $j \in Q^*$, are close to both $x$ and $y$, albeit held potentially looser since this closeness is with respect to the larger parameter $\epsilon$.

\begin{lem}
\label{lem:support}
We have $\Phi(x,y,\bold{\hat{x}}) = 0$ unless
\begin{align}
\label{eq:lam_b1}
|x - x_k|,\,|y-x_k|,\, |x_j - x_k| &> (4N)^{-1}\epsilon \qquad &&j \in Q^*, k\in Q^c,\\
\label{eq:lam_b3}
|x -x_j|, |y -x_k| &> (4N)^{-1}\delta \quad &&j \in P^{c},\, k \in S^c,\\
\label{eq:lam_b4}
|x_r -x_s| &> (4N)^{-1}\delta \quad && r \in P^{*},\, s \in P^{c} \text{  or  } r \in S^{*},\, s \in S^{c}.
\end{align}
\end{lem}
\begin{proof}
Take $j \in Q^*$ and $k \in Q^c$. By the definition of the cluster $Q$ we have $f_{j,k} = \theta_{\epsilon}$ and  $g^{(1)}_k = g^{(2)}_k = \theta_{\epsilon}$ in the formula (\ref{eq:pphi_def}) defining $\Phi$. If we assume $\Phi(x,y,\hbx) \ne 0$, then the support criteria for $\theta_{\epsilon}(x-x_k), \theta_{\epsilon}(y-x_k)$ and $\theta_{\epsilon}(x_j-x_k)$ give (\ref{eq:lam_b1}). The inequalities (\ref{eq:lam_b3}) and (\ref{eq:lam_b4}) are proven in a similar way.
\end{proof}

\begin{lem}
\label{particle_ineqs}
Let $0<2\delta \le \epsilon$. Then $\Phi(x,y,\hat{\bold{x}}) = 0$ unless
\begin{align}
|x-x_k| < \delta/2 \quad\text{ for }\quad k \in P^*,\\
|y-x_k| < \delta/2 \quad\text{ for }\quad k \in S^*,
\end{align} 
and if $|x-y| \ge \delta$ then $\Phi(x,y,\hat{\bold{x}}) = 0$ unless
\begin{align}
|y-x_k| > \delta/2 \quad\text{ for }\quad k \in P^*,\\
|x-x_k| > \delta/2 \quad\text{ for }\quad k \in S^*.
\end{align} 
Therefore, if $P^* \cap S^* \ne \emptyset$ then $\Phi(x,y,\bold{\hat{x}}) = 0$ for all $|x-y| \ge \delta$ and $\bold{\hat{x}} \in \real^{3N-3}$.

Furthermore, if $|x|,|y|\ge \epsilon$ then $\Phi(x,y,\hat{\bold{x}}) = 0$ unless
\begin{align}
\label{eq:particle_ineq3}
\min\{|x-x_k|,\, |y-x_k|\} < \epsilon/2 \quad\text{ for }\quad k \in Q^*,\\
\label{eq:lam_b2}
|x_k| > \epsilon/2 \quad\text{ for }\quad k \in Q^*.
\end{align}
%

%
\end{lem}


%
\begin{proof}
By definition, if $k \in P^{*}$ either $g^{(1)}_k = \zeta_{\delta}$ or there exist pairwise distinct $j_1, \dots, j_s \in \{2,\dots,N\}$ with $1\le s \le N-2$ such that $g^{(1)}_{j_1} = \zeta_{\delta}$ and $f_{j_1,j_2}, f_{j_2, j_3}, \, \dots,\, f_{j_s, k} = \zeta_{\delta}$. In the former case, support criteria for $\zeta_{\delta}(x-x_k)$ gives $|x - x_k| < (2N)^{-1}\delta$. In the latter case, support conditions give $|x - x_{j_1}|,\, |x_{j_1} - x_{j_2}|,\, \dots,\, |x_{j_s} - x_k|  < (2N)^{-1}\delta$ and so by the triangle inequality $|x - x_k| < \delta/2$. Now, $|y-x_k| \ge |x-y|-|x-x_k| > \delta/2$ by the reverse triangle inequality. The case of $k \in S^*$ is analogous.

\paragraph{}
Now let $k \in Q^*$. First, we consider the case where either $g^{(1)}_k \ne \theta_{\epsilon}$ or $g^{(2)}_k \ne \theta_{\epsilon}$, or both. Without loss, assume $g^{(1)}_k \ne \theta_{\epsilon}$. Then either $g^{(1)}_k = \zeta_{\delta}$ or  $g^{(1)}_k = \theta_{\delta}\zeta_{\epsilon}$, but both give the inequality
\begin{equation*}
|x-x_k| < (2N)^{-1}\epsilon
\end{equation*} 
by support criteria. Hence (\ref{eq:particle_ineq3}) holds in this case.
Now suppose that $g^{(1)}_k = g^{(2)}_k = \theta_{\epsilon}$. Then there must exist pairwise distinct $j_1, \dots, j_s \in \{2,\dots,N\}$ with $1\le s \le N-2$ such that $f_{j_1,j_2}, f_{j_2,j_3},\dots,f_{j_s, k} \ne \theta_{\epsilon}$ and either $g^{(1)}_{j_1} \ne \theta_{\epsilon}$ or $g^{(2)}_{j_1} \ne \theta_{\epsilon}$. Without loss take $g^{(1)}_{j_1} \ne \theta_{\epsilon}$, so that, as before, $|x-x_{j_1}| < (2N)^{-1}\epsilon$. Similarly, we get $|x_{j_1} - x_{j_2}|,\, \dots,\, |x_{j_s} - x_k|  < (2N)^{-1}\epsilon$. Therefore, by the triangle inequality, $|x-x_k| < \epsilon/2$. Hence completing the proof of (\ref{eq:particle_ineq3}). Finally, (\ref{eq:lam_b2}) follows from (\ref{eq:particle_ineq3}) and $|x|,|y| \ge \epsilon$.
\end{proof}

\subsection{Factorisation of cutoffs}
Let $\Phi$ be given by (\ref{eq:pphi_def}). We can define a \textit{partial product} of $\Phi$ as a function of the form
\begin{equation}
\label{beprime}
\Phi'(x,y,\bold{\hat{x}}) = \prod_{j \in R_1} g^{(1)}_{j}(x - x_j) \prod_{j \in R_2}  g^{(2)}_{j}(y -x_j) \prod_{(k,l) \in R_3} f_{k,l}(x_k - x_l)
\end{equation}
where $R_1, R_2 \subset \{2, \dots, N \}$, $R_3 \subset \{(k,l) : 2 \le k<l \le N \}$.

We now define classes of partial products of $\Phi$ which corresponding to a cluster.  Let $T$ be an arbitrary cluster with $1 \in T$. Then,
\begin{align}
\label{eq:pphi_t}
\Phi(x,y,\bold{\hat{x}}; T) &= \prod_{j \in T^*} g^{(1)}_j(x-x_j) \prod_{j \in T^*} g^{(2)}_j(y-x_j) \prod_{\substack{k,l \in T^* \\ k < l}} f_{k,l}(x_k - x_l),\\ 
\label{eq:pphi_tc}
\Phi(x,y,\bold{\hat{x}}; T^c) &= \prod_{\substack{k,l \in T^c \\ k < l}} f_{k,l}(x_k - x_l),\\
\label{eq:pphi_ttc}
\Phi(x,y,\bold{\hat{x}}; T, T^c) &= \prod_{j \in T^c} g^{(1)}_j(x-x_j) \prod_{j \in T^c} g^{(2)}_j(y-x_j) \prod_{\substack{k \in T^* \\ l \in T^c}} f_{k,l}(x_k - x_l).
\end{align}
Formulae (\ref{eq:pphi_t}) and (\ref{eq:pphi_tc}) are separate definitions in the cases where the cluster contains $1$ and does not contain $1$, respectively.

If we consider both $x$ and $y$ to adopt the role of particle $1$ then these functions can be interpreted as $\Phi(\,\cdot\,;T)$ involving only pairs of particles in $T$, $\Phi(\,\cdot\,;T^c)$ involving only pairs in $T^c$, and $\Phi(\,\cdot\,;T,T^c)$ involving pairs where one particle lies in $T$ and the other lies in $T^c$.

\begin{lem}
\label{phi_factors_lem}
Given a cutoff $\Phi$ of the form (\ref{eq:pphi_def}) and any cluster $T$ with $1 \in T$ we have
\begin{equation}
\label{phi_factors}
\Phi(x,y,\bold{\hat{x}}) =  \Phi(x,y,\bold{\hat{x}}; T)\Phi(x,y,\bold{\hat{x}};T,T^c)\Phi(x,y,\bold{\hat{x}}; T^c).
\end{equation}
\end{lem}
\begin{proof}
This identity follows from the definitions (\ref{eq:pphi_t})-(\ref{eq:pphi_ttc}) and the equality
\begin{equation*}
\prod_{2 \le k<l \le N} f_{k,l}(x_k - x_l) = \prod_{\substack{k,l \in T^* \\ k < l}} f_{k,l}(x_k - x_l) \prod_{\substack{k \in T^* \\ l \in T^c}} f_{k,l}(x_k - x_l) \prod_{\substack{k,l \in T^c \\ k < l}} f_{k,l}(x_k - x_l),
\end{equation*}
since for any $2 \le k<l \le N$ we have $f_{k,l} = f_{l,k}$ by definition.
\end{proof}

Let $Q = Q(\Phi)$. Then the partial product $\Phi(\,\cdot\,; Q,Q^c)$ consists of only $\theta_{\epsilon}$ cutoff factors, as shown in the following lemma.
\begin{lem}
\label{lem:pphi_qqc}
For $Q = Q(\Phi)$,
\begin{equation}
\label{eq:phiqqc}
\Phi(x,y,\bold{\hat{x}}; Q,Q^c) =  \prod_{j \in Q^c} \theta_{\epsilon}(x - x_j) \prod_{j \in Q^c} \theta_{\epsilon}(y - x_j) \prod_{\substack{k \in Q^* \\ l \in Q^c}} \theta_{\epsilon}(x_k - x_l).
\end{equation}
\end{lem}
\begin{proof}
By the definition of the cluster $Q(\Phi)$, we have $g^{(1)}_j = g^{(2)}_j = \theta_{\epsilon}$ for each $j \in Q^c$, and $f_{k,l} = \theta_{\epsilon}$ for each $k \in Q^*$ and $l \in Q^c$. The formula then follows from (\ref{eq:pphi_ttc}).
\end{proof}

\subsection{Derivatives}
The definition of cluster derivatives, (\ref{eq:cd}), is modestly extended to allow action on cutoffs. For any cluster $T$ we define the following three cluster derivatives which can act on functions of $x$, $y$ and $\bold{\hat{x}}$, such as $\Phi$. We set
\begin{align}
\label{eq:dxbc}
D^{\alpha}_{x,T} = \partial^{\alpha}_x +  \sum_{j \in T^*} \partial^{\alpha}_{x_j} \quad \text{and}\quad D^{\alpha}_{y,T} = \partial^{\alpha}_y +  \sum_{j \in T^*} \partial^{\alpha}_{x_j} \quad \text{for } \alpha \in \naturals_0^3,\, |\alpha|=1,
\end{align}
\begin{align}
\label{eq:dxybc}
D^{\alpha}_{x,y,T} = \partial^{\alpha}_x + \partial^{\alpha}_y +  \sum_{j \in T^*} \partial^{\alpha}_{x_j}\quad \text{for } \alpha \in \naturals_0^3,\, |\alpha|=1,
\end{align}
which are extended to higher order multiindices $\alpha \in \naturals_0^3$ by successive application of first-order derivatives, as in (\ref{eq:cd2}).

The following lemma gives partial derivative estimates of the cutoff factors. We will require the following elementary result. For each $\sigma \in \naturals_0^3$ and $s \in \real$ there exists $C>0$ such that for any $x_0 \in \real^{3}$ we get $\big|\partial_x^{\sigma}|x+x_0|^{s}\big| \le C|x+x_0|^{s-|\sigma|}$ for all $x \in \real^3$, $x \ne -x_0$. We use the standard notation $\mathds{1}_{S}$ to denote the indicator function on a set $S$.

\begin{lem}
\label{dcutoff_factors}
For any $\sigma \in \naturals_0^3$ with $|\sigma| \ge 1$ and any $t > 0$ there exists $C$, depending on $\sigma$ but independent of $t$, such that
\begin{equation}
\label{dcutoff_factors1}
|\partial^{\sigma} \zeta_{t}(x)|,\, |\partial^{\sigma} \theta_{t}(x)| \le C t^{-|\sigma|} \mathds{1}_{\{(4N)^{-1}t < |x| < (2N)^{-1}t\}}(x)
\end{equation}
for all $x \in \real^3$.
\end{lem}
\begin{proof}
Without loss we consider the case of $\zeta_{t}$, the case of $\theta_{t}$ being similar. Recall $\xi = \xi(s)$ was defined in (\ref{eq:chi_def}) and we denote by $\xi^{(m)}$ the $m$-th (univariate) derivative of $\xi$. Since $|\sigma| \ge 1$ the chain rule shows that $\partial^{\sigma} \zeta_{t}(x)$ can be written as a sum of terms of the form
\begin{equation}
\label{dtheta}
(4Nt^{-1})^m \, \xi^{(m)}\big(4N|x|t^{-1}\big)\, \partial^{\sigma_1} |x| \, \dots \,\partial^{\sigma_m} |x|
\end{equation}
where $1 \le m \le |\sigma|$, and $\sigma_1, \dots, \sigma_m \in \naturals_0^3$ are non-zero multiindices obeying
\begin{equation*}
\sigma_1 + \dots + \sigma_m = \sigma.
\end{equation*}
Since $m \ge 1$, we have that if $\xi^{(m)}(s) \ne 0$ then $s \in (1,2)$. Therefore, for any term (\ref{dtheta}) to be non-zero we require that
\begin{equation}
\label{l_bound}
(4N)^{-1}t < |x| <(2N)^{-1}t.
\end{equation}
By the remark preceeding the current lemma, there exists $C$, dependent on $\sigma_1, \dots, \sigma_m$, such that
\begin{align*}
\partial^{\sigma_1} |x| \, \dots \,\partial^{\sigma_m} |x| \le  C|x|^{m-|\sigma|} \le C(4N)^{|\sigma|-m} t^{m-|\sigma|},
\end{align*}
using (\ref{l_bound}). Therefore, the terms (\ref{dtheta}) can readily be bounded to give the desired result.
\end{proof}

We now give bounds for the cluster derivatives (\ref{eq:dxbc})-(\ref{eq:dxybc}) acting on cutoffs.
\begin{lem}
\label{dqpq_lem}
Let $Q=Q(\Phi)$. Then $D^{\alpha}_{x,y,Q}\Phi(\,\cdot\,;Q) \equiv 0$ for all $\alpha \in \naturals_0^3$ with $|\alpha| \ge 1$.
\end{lem}
\begin{proof}
By the chain rule, each function in the product (\ref{eq:pphi_t}) for $\Phi(\,\cdot\,;Q)$ has zero derivative upon action of $D^{\alpha}_{x,y,Q}$.
\end{proof}

For the next result we will use the following function, defined for all $t>0$ by
\begin{align}
\label{eq:mt}
&M_{t}(x,y,\bold{\hat{x}}) = \sum_{2 \le j \le N} \mathds{1}_{\{(4N)^{-1}t < |x-x_j| < (2N)^{-1}t\}}(x-x_j)\\ 
\nonumber
&\quad+ \sum_{2 \le j \le N} \mathds{1}_{\{(4N)^{-1}t < |y-x_j| < (2N)^{-1}t\}}(y-x_j) + \sum_{2 \le k < l \le N} \mathds{1}_{\{(4N)^{-1}t < |x_k-x_l| < (2N)^{-1}t\}}(x_k-x_l)
\end{align}
for $x,y \in \real^3$ and $\bold{\hat{x}} \in \real^{3N-3}$.

\begin{lem}
\label{lem:dpphi_prime}
Let $0<2\delta \le \epsilon$ and $\Phi = \Phi_{\delta,\epsilon}$ be a cutoff of the form (\ref{eq:pphi_def}). Let $Q=Q(\Phi)$. Then for any multiindex $\bal \in \naturals_0^{3N+3}$ there exists $C$, dependent on $\bal$ but independent of $\delta$ and $\epsilon$, such that for any partial products $\Phi'$ of $\Phi$ we have
\begin{align}
\label{partial_a}
|\partial^{\bal}\Phi'(x,y,\bold{\hat{x}})| \le
\begin{cases}
C \epsilon^{-|\bal|} &\text{ if } \Phi' = \Phi(\,\cdot\,;Q,Q^c)\\
C \big(\epsilon^{-|\bal|} + \delta^{-|\bal|} M_{\delta}(x,y,\bold{\hat{x}})\big) &\text{ otherwise.}\\
\end{cases}
\end{align}
\end{lem}
\begin{proof}
%
%
%
Lemma \ref{dcutoff_factors} gives bounds for the partial derivatives of the functions $\zeta_{\delta}, \theta_{\delta}, \zeta_{\epsilon}, \theta_{\epsilon}$. Considering $\theta_{\delta}\zeta_{\epsilon}$, we apply the Leibniz rule with $\sigma \in \naturals_0^3$, $|\sigma| \ge 1$, to obtain, for $z \in \real^3$
\begin{align}
\label{dcutoff_factor_final3}
\partial^{\sigma} (\theta_{\delta}\zeta_{\epsilon})(z) &= \sum_{\mu \le \sigma} {\sigma \choose \mu} \partial^{\mu} \theta_{\delta}(z) \partial^{\sigma-\mu} \zeta_{\epsilon}(z)
= \partial^{\sigma} \theta_{\delta}(z) + \partial^{\sigma} \zeta_{\epsilon}(z),
\end{align}
since, for each $\mu \le \sigma$ with $\mu \ne 0$ and $\mu \ne \sigma$ we have
\begin{equation*}
\partial^{\mu} \theta_{\delta}(z)\, \partial^{\sigma-\mu}\zeta_{\epsilon}(z) \equiv 0
\end{equation*}
by Lemma \ref{dcutoff_factors} and that $2\delta \le \epsilon$.
\paragraph{}
Now, to evaluate $\partial^{\bal}\Phi'$ for a general $\Phi'$ we apply the Leibniz rule to the product (\ref{beprime}) using (\ref{dcutoff_factor_final3}) where appropriate. Differentiated cutoff factors $\zeta_\delta, \theta_\delta, \zeta_\epsilon, \theta_\epsilon$ are bounded by (\ref{dcutoff_factors1}). The indicator function in this bound is not required in the case of $\zeta_\epsilon$ or $\theta_\epsilon$, whereas in the case of $\zeta_\delta$ or $\theta_\delta$ the indicator function is bounded above by $M_t$. Any remaining undifferentiated cutoff factors are bounded above by 1. If $\Phi' = \Phi(\,\cdot\,;Q,Q^c)$, all cutoff factors are of the form $\theta_{\epsilon}$ by Lemma \ref{lem:pphi_qqc} and therefore we need only use the bounds in (\ref{dcutoff_factors1}) with $t=\epsilon$.

\end{proof}

The derivative $D_{x,y,Q}$ acting on $\Phi$ is special in that it contributes only powers of $\epsilon$ (and not $\delta$) to the bounds. This is shown in the next lemma.

\begin{lem}
\label{lem:dphi}
Let $0<2\delta \le \epsilon$ and $\Phi=\Phi_{\delta,\epsilon}$ be a cutoff of the form (\ref{eq:pphi_def}). Let $Q=Q(\Phi)$. For any multiindices $\alpha \in \mathbb{N}_0^3$ and $\bsig \in \mathbb{N}_0^{3N+3}$ there exists $C$, independent of $\epsilon$ and $\delta$, such that
\begin{equation*}
|\partial^{\bsig} D^{\alpha}_{x,y,Q} \Phi(x,y, \bold{\hat{x}})| \le C \epsilon^{-|\alpha|}\big( \epsilon^{-|\bsig|} + \delta^{-|\bsig|} M_{\delta}(x,y,\bold{\hat{x}})\big)
\end{equation*}
for all $x,y,\bold{\hat{x}}$.
\end{lem}
\begin{proof}
First, set $\Phi' = \Phi(\,\cdot\,;Q) \, \Phi(\,\cdot\,;Q^c)$ and $\Phi'' = \Phi(\,\cdot\,;Q,Q^c)$. We then have
\begin{equation*}
D^{\alpha}_{x,y,Q} \Phi = \Phi'\, D^{\alpha}_{x,y,Q} \Phi''
\end{equation*}
which follows from Lemmas \ref{phi_factors_lem} and \ref{dqpq_lem}, and that $\Phi(\,\cdot\,;Q^c)$ is not dependent on variables involved in the $D^{\alpha}_{x,y,Q}$-derivative. By the definition (\ref{eq:dxybc}), the derivative $D^{\alpha}_{x,y,Q} \Phi''$ can be written as a sum of partial derivatives of the form $\partial^{\bal} \Phi''$ where $\bal \in \naturals_0^{3N+3}$ obeys $|\bal| = |\alpha|$. Now, by the Leibniz rule and Lemma \ref{lem:dpphi_prime} there exists some constants $C$ and $C'$, independent of $\delta$ and $\epsilon$, such that
\begin{align*}
|\partial^{\bsig} (\Phi' \, \partial^{\bal} \Phi'' )| &\le \sum_{\boldsymbol{\tau} \le \bsig} {\bsig \choose \boldsymbol{\tau}} |\partial^{\boldsymbol{\tau}}\Phi'| \, |\partial^{\bsig-\boldsymbol{\tau}+\bal} \Phi''|\\
&\le C\sum_{\boldsymbol{\tau} \le \bsig} (\epsilon^{-|\boldsymbol{\tau}|} + \delta^{-|\boldsymbol{\tau}|}M_{\delta}) \epsilon^{-|\bsig|+|\boldsymbol{\tau}|-|\bal|}\\
&\le C'\epsilon^{-|\al|}(\epsilon^{-|\bsig|} + \delta^{-|\bsig|}M_{\delta}),
\end{align*}
completing the proof.
\end{proof}

\subsection{Integrals involving $f_{\infty}$}

The following proposition is a restatement of \cite[Lemma 5.1]{hearn_sob2} and is proven in that paper. The function $M_t$ defined in (\ref{eq:mt}) has a slightly different form to the corresponding function used in the paper but this does not affect the proof. Recall $f_{\infty}$ was defined in (\ref{eq:f_inf_p}).

\begin{prop}
\label{prop:ff}
Given $R >0$, there exists $C$ such that
\begin{equation}
\label{eq:ff_int}
\int_{\real^{3N-3}} f_{\infty}(x,\bold{\hat{x}}; R) f_{\infty}(y,\bold{\hat{x}};R)\, d\bold{\hat{x}} \le C \norm{\rho}^{1/2}_{L^{1}(B(x, 2R))} \norm{\rho}^{1/2}_{L^{1}(B(y, 2R))} 
\end{equation}
for all $x,y \in \real^3$. In addition, given $G \in L^1(\real^3)$ there exists $C$, independent of $G$, such that
\begin{multline}
\label{eq:gff_int}
\int_{\real^{3N-3}} \big(|G(x_j - x_k)| + |G(z- x_k)| + |G(x_j)|\big) f_{\infty}(x,\bold{\hat{x}}; R) f_{\infty}(y,\bold{\hat{x}};R)\, d\bold{\hat{x}} \\\le C\norm{G}_{L^1(\real^3)} \norm{\rho}^{1/2}_{L^{1}(B(x, 2R))} \norm{\rho}^{1/2}_{L^{1}(B(y, 2R))} 
\end{multline}
for all $x,y, z \in \real^3$, and $j, k =2, \dots, N$, $j \ne k$. In particular, for any $t>0$ there exists $C$, independent of $t$, such that
\begin{equation}
\label{eq:mt_int}
\int_{\real^{3N-3}} M_{t}(x,y,\bold{\hat{x}}) f_{\infty}(x,\bold{\hat{x}}; R) f_{\infty}(y,\bold{\hat{x}};R)\, d\bold{\hat{x}} \le C t^{3} \norm{\rho}^{1/2}_{L^{1}(B(x, 2R))} \norm{\rho}^{1/2}_{L^{1}(B(y, 2R))} 
\end{equation}
for all $x,y \in \real^3$.
\end{prop}
%


Recall from (\ref{eq:f_inf2}) that $f_{\infty}(\,\cdot\,)$ was, for convenience, defined as $f_{\infty}(\,\cdot\,; R)$ with $R=1/2$.
\begin{lem}
\label{lem:lam_b}
Let $\Phi = \Phi_{\delta,\epsilon}$ be a cutoff of the form (\ref{eq:pphi_def}), and let $j,k,r,s$ be as in (\ref{eq:lam_b3}) and (\ref{eq:lam_b4}). Then for any $b = b_1+b_2+b_3$, $b\ne 3$, $b_1,b_2,b_3 \ge 0$ there exists $C$, independent of $x,y$ and $\delta,\epsilon$, such that
\begin{multline}
\label{eq:lam_b5}
\int_{\textnormal{supp}\,\Phi(x,y,\,\cdot\,)} |x-x_j|^{-b_1}|y-x_k|^{-b_2}|x_r-x_s|^{-b_3} f_{\infty}(x,\bold{\hat{x}}) f_{\infty}(y,\bold{\hat{x}}) \,d\bold{\hat{x}} \\\le
C(1+\delta^{\min\{0, 3-b\}}) \norm{\rho}^{1/2}_{L^1(B(x,1))}\norm{\rho}^{1/2}_{L^1(B(y,1))}.
\end{multline}

\end{lem}

\begin{proof}
By Young's inequality,
\begin{align*}
|x-x_j|^{-b_1}|y-x_k|^{-b_2}|x_r-x_s|^{-b_3} \le \frac{1}{b}\big(b_1|x-x_j|^{-b} + b_2|y-x_k|^{-b} + b_3|x_r-x_s|^{-b} \big).
\end{align*}
It therefore suffices to prove (\ref{eq:lam_b5}) where only one of $b_1, b_2$ and $b_3$ is non-zero. We consider the case where this is $b_1$, the other cases are similar. We split the integral into one where $|x-x_j| \ge 1$ and one where $|x-x_j| < 1$. In the former case, we can bound using (\ref{eq:ff_int}) of Proposition \ref{prop:ff}, whereas in the latter we use (\ref{eq:gff_int}) of the same proposition with $G(x-x_j) = \mathds{1}_{\{(4N)^{-1}\delta < |x-x_j| < 1\}}(x-x_j)|x-x_j|^{-b}$. The lower bound can be included in the indicator function due to (\ref{eq:lam_b3}) of Lemma \ref{lem:support}.

\end{proof}

\section{Proof of Theorem \ref{thm:5}}
\label{chpt:4}

We begin by introducing auxiliary functions related to the density matrix, $\g(x,y)$, defined in (\ref{eq:gamma}). For $l, m \in \naturals_0^3$ with $|l|, |m| \le 1$ define,
\begin{equation}
\label{eq:gam_lm}
\gamma_{l,m}(x,y) = \int_{\real^{3N-3}} \partial_x^{l} \psi(x, \bold{\hat{x}}) \overline{\partial_y^m \psi(y, \bold{\hat{x}})}\,d\bold{\hat{x}}.
\end{equation}
In this notation, it is clear that $\g = \g_{0,0}$. By differentiation under the integral, we have
\begin{align}
\label{eq:du_g}
\partial^{l}_x \partial^m_y \g(x,y) &= \g_{l,m}(x,y).
\end{align}
Furthermore, for any cutoff $\Phi$ of the form (\ref{eq:pphi_def}), we set
\begin{align*}
\gamma_{l,m}(x,y;\Phi) &= \int_{\real^{3N-3}} \partial_x^l \psi(x, \bold{\hat{x}}) \overline{\partial_y^m \psi(y, \bold{\hat{x}})} \Phi(x,y,\bold{\hat{x}}) \,d\bold{\hat{x}},
\end{align*}
and define $\gamma(\,\cdot\,;\Phi) = \gamma_{0,0}(\,\cdot\,;\Phi)$.

To produce the required bound for derivatives $\partial_x^{\al}\partial_y^{\beta}\g(x,y)$ for $|\al|+|\beta|=5$ we consider separately the cases where $|\al|,|\beta|\ge 1$ and where at least one of $\al, \beta$ is zero. In the former case, we will immediately differentiate under the integral once in both $x$ and $y$. This is because Proposition \ref{prop:glm}, see below, is not affected by whether $l$ and $m$ in the proposition are zero or represent single derivatives. This itself is a consequence of the fact that $\psi$ and $\nabla\psi$ are both bounded functions, the function $\psi$ only becoming singular after two derivatives.

When one of $\al, \beta$ is zero, we can rewrite the derivatives into an expression involving $(x+y)$-derivatives as follows. Without loss, suppose $\beta = 0$ and hence $|\al|=5$. Let $l,m \in \naturals_0^3$, $|l|=|m|=1$ be such that $l+m \le \al$. Then we can use
$\partial^{l+m}_{x+y} = (\partial_x^l + \partial_y^l)(\partial_x^m + \partial_y^m)$, see (\ref{eq:dd}), to write the identity
\begin{align}
\label{eq:2dd}
\partial_x^{l+m} = \partial^{l+m}_{x+y} - \partial_y^{l+m} - \partial_x^l\partial_y^m - \partial_x^m\partial_y^l,
\end{align}
which will hold on sufficiently smooth functions. We note that rearrangement would give us the corresponding expression for $\partial_y^{l+m}$. The final three terms above introduce $y$-derivatives, and since we already have the three remaining derivatives of $\al$, we can differentiate once in both $x$ and $y$, as in the previous case. The first term in the equality above involves an $(x+y)$-derivative. It will be shown that such derivatives do not contribute to the singularity at the diagonal for the density matrix.

\subsection{Differentiating the density matrix - general derivatives}

The following lemma gives bounds to differentiating $\g_{l,m}(\,\cdot\,; \Phi)$. We will then use that the cutoffs form a partition of unity to obtain bounds to derivatives of $\g_{l,m}$. As before, for a given $\Phi$ we will set $P = P(\Phi), S=S(\Phi), Q=Q(\Phi)$ as the corresponding clusters as defined in the previous section.

\begin{lem}
\label{lem:glm}
Let $\Phi = \Phi_{\delta,\epsilon}$ be a cutoff of the form (\ref{eq:pphi_def}) with $P^* \cap S^* = \emptyset$. For all $l,m \in \naturals_0^3$, $|l|,|m| \le 1$, and all $\mu_1, \mu_2, \mu_3 \in \naturals_0^3$ which obey either
\begin{enumerate}[label=(\roman*)]
\item $|\mu_2|+|\mu_3| \ne 3$ with $\mu_1$ arbitrary, or
\item $|\mu_2|+|\mu_3| = 3$ and $\mu_1 = 0$,
\end{enumerate}
there exists a constant $C$ such that for all $0<2\delta \le \epsilon \le 1$ we have
\begin{multline}
\label{eq:dglm_int2}
|\partial_{x+y}^{\mu_1}\partial_x^{\mu_2}\partial_y^{\mu_3} \gamma_{l,m}(x,y; \Phi_{\delta,\epsilon})| \\\le C \epsilon^{-|\mu_1|-|\mu_2|-|\mu_3|}\delta^{\min\{0,\,3- |\mu_2|-|\mu_3|\}}
\norm{\rho}^{1/2}_{L^{1}(B(x, 1))} \norm{\rho}^{1/2}_{L^{1}(B(y, 1))}
\end{multline}
for all $|x|, |y| \ge \epsilon$ and $|x-y| \le 2\delta$.
\end{lem}
The proof of this lemma is quite long so is postponed until after the proof of Theorem \ref{thm:5}. Here, as before, $m(x,y) = \min\{1, |x|, |y| \}$.

\begin{prop}
\label{prop:glm}
Take any $l,m \in \naturals_0^3$, $|l|,|m| \le 1$, and any $\mu_1, \mu_2, \mu_3$ as in Lemma \ref{lem:glm}. Then there exists $C$ such that
\begin{equation}
\label{eq:prop_glm}
|\partial_{x+y}^{\mu_1}\partial_x^{\mu_2}\partial_y^{\mu_3}\gamma_{l,m}(x,y)| \le Cm(x,y)^{-|\mu_1|-|\mu_2|-|\mu_3|} |x-y|^{\min\{0,3-|\mu_2|-|\mu_3|\}}
\norm{\rho}^{1/2}_{L^{1}(B(x, 1))} \norm{\rho}^{1/2}_{L^{1}(B(y, 1))}
\end{equation}
for all $x,y \in \real^3$ obeying $0<|x-y| \le m(x,y)/2$.
\end{prop}
\begin{proof}
Firstly, by Lemma \ref{pou} there exists a finite collection of cutoffs, $\Phi^{(j)}$, $j=1,\dots,J$, for some $J$, such that
\begin{equation}
\label{eq:g_pou}
\gamma_{l,m} = \sum_{j=1}^J \gamma_{l,m}\big(\,\cdot\,;\Phi_{\delta,\epsilon}^{(j)}\big)
\end{equation}
holds everywhere for all choices of $0< 2\delta \le \epsilon$.

Fix any $x,y$ such that $0<|x-y| \le m(x,y)/2$. Then set $\delta = |x-y|/2$ and $\epsilon = m(x,y)/2$. Index by $j_k$, where $k=1,\dots, K$ for some $K \le J$, the $j$'s such that $P(\Phi^{(j)})^* \cap S(\Phi^{(j)})^* = \emptyset$. The $j_k$'s do not depend on the choice of $\delta$ and $\epsilon$. By Lemma \ref{particle_ineqs},
\begin{equation}
\gamma_{l,m}(x',y') = \sum_{k=1}^{K} \gamma_{l,m}\big(x',y';\Phi_{\delta,\epsilon}^{(j_k)}\big)
\end{equation}
holds for all $|x'-y'| \ge |x-y|/2$. Then by Lemma \ref{lem:glm}, for each $k=1, \dots, K$, we have
\begin{align*}
&|\partial_{x'+y'}^{\mu_1}\partial_{x'}^{\mu_2}\partial_{y'}^{\mu_3}\gamma_{l,m}(x',y')| \le \sum_{k=1}^K\big|\partial_{x'+y'}^{\mu_1}\partial_{x'}^{\mu_2}\partial_{y'}^{\mu_3}\gamma_{l,m}\big(x',y';\Phi_{\delta,\epsilon}^{(j_k)}\big)\big|\\ 
&\qquad\le Cm(x,y)^{-|\mu_1|-|\mu_2|-|\mu_3|}|x-y|^{\min\{0,3-|\mu_2|-|\mu_3|\}}\norm{\rho}^{1/2}_{L^{1}(B(x', 1))} \norm{\rho}^{1/2}_{L^{1}(B(y', 1))}
\end{align*}
for all $x', y'$ such that $|x-y|/2 < |x'-y'| \le |x-y|$ and $|x'|,|y'| \ge m(x,y)/2$. In particular, it holds for $x'=x$ and $y'=y$. The constant $C$ does not depend on the choice of $\delta$ and $\epsilon$, therefore the bound holds for all required $x$ and $y$.

\end{proof}

The proof of our main theorem is an immediate consequence of this proposition.
\begin{proof}[Proof of Theorem \ref{thm:5}.]
If $|\al|,|\beta| \ge 1$ then take any $l \le \al$ and $m \le \beta$ with $|l|=|m|=1$. We can then write
\begin{align*}
\partial_x^{\al} \partial_y^{\beta}\gamma = \partial_x^{\al-l} \partial_y^{\beta-m}\gamma_{l,m}.
\end{align*}
It is then straightforward to obtain the required bound by Proposition \ref{prop:glm} with $\mu_1 = 0$, $\mu_2 = \al - l$ and $\mu_3 = \beta - m$.

We now consider the case where either $\al$ or $\beta$ is zero. Without loss, assume $\beta = 0$. Let $l \le \al$ and $m \le \al-l$ obey $|l|=|m|=1$. We then have by (\ref{eq:2dd}),
%
\begin{align*}
\partial^{\al}_x\gamma &= \partial^{l+m}_{x+y}\partial_x^{\al-l-m}\gamma - \partial_x^{\al-l-m}\partial_y^{l+m}\gamma - \partial_x^{\al-m}\partial_y^m\gamma - \partial_x^{\al-l}\partial_y^l\gamma
\end{align*}
It suffices to bound each term separately. The final three terms have at least one derivative on each $x$ and $y$, hence we be treated as before. Next, the first term can be rewritten as $\partial^{l+m}_{x+y}\partial_x^{\al-l-m-r}\gamma_{r,0}$ for some $r \le \al-l-m$ with $|r|=1$. Use of Proposition \ref{prop:glm} with $\mu_1 = l+m$, $\mu_2 = \al-l-m-r$ and $\mu_3 = 0$ then gives the required bound.
\end{proof}

\subsection{Proof of Lemma \ref{lem:glm}}
The cutoffs introduced in Section \ref{chpt:3} facilitate taking simultaneous derivatives in $x$, $y$ and $x+y$. The strategy in proving Lemma \ref{lem:glm} will be to turn partial derivatives $\g_{l,m}(\,\cdot\,;\Phi)$ into cluster derivatives under the integral. Theorem \ref{thm:ddpsi} then gives the required pointwise bounds for these cluster derivatives, and in fact distinguishes between derivatives of different clusters

Partial derivatives in $x$ and $y$ will produce cluster derivatives for clusters $P,S,P^*$ and $S^*$, whereas derivatives in $x+y$ will produce cluster derivatives for the cluster $Q$. We wish to separate the contributions from $x$- and $y$-derivatives and from ($x+y$)-derivatives. To do this it will be convenient to group together the distances $\lambda_P$ etc. which will later appear when Theorem \ref{thm:ddpsi} is applied to the differentiated density matrix. This is what we define now as the functions $\lambda$ and $\pi$.

First, we fix a cutoff $\Phi = \Phi_{\delta,\epsilon}$ to be used throughout this section, and let $P,S,Q$ be the corresponding clusters. We will assume that $P^* \cap S^* = \emptyset$ and $0<2\delta \le \epsilon \le 1$.

For any $x,y \in \real^3$ and $\bold{\hat{x}} \in \real^{3N-3}$, we define
\begin{align}
\label{eq:lam_xyx}
\lambda(x,y,\bold{\hat{x}}) &= \min\{\lambda_{P}(x,\bold{\hat{x}}),\, \lambda_{S^*}(x,\bold{\hat{x}}),\, \lambda_{P^*}(y,\bold{\hat{x}}),\, \lambda_{S}(y,\bold{\hat{x}}) \},\\
\label{eq:mu_xyx}
\pi(x,y,\bold{\hat{x}}) &= \min\{\lambda_{Q}(x,\bold{\hat{x}}),\, \lambda_{Q}(y,\bold{\hat{x}}) \}.
\end{align}
%
%
Recall that $1 \in P, S, Q$ by definition. Using (\ref{dp_def}) we find that
\begin{multline}
\label{eq:pi_alt}
\pi(x,y,\bold{\hat{x}}) = \min\{ 1,\,\, |x|,\,\,|y|,\,\, |x_j| : j \in Q^*,\,\, 2^{-1/2}|x - x_k| : k \in Q^c, \\2^{-1/2}|y - x_k| : k \in Q^c,\,\, 2^{-1/2}|x_j - x_k| : j \in Q^*, k \in Q^c \}.
\end{multline}
Using that $P^* \cap S^* = \emptyset$, we similarly find that
\begin{multline}
\label{eq:lam_alt}
\lambda(x,y,\bold{\hat{x}}) = \min\{ 1,\,\, |x|,\,\,|y|,\,\, |x_j| : j \in P^* \text{ or } S^*,\,\, 2^{-1/2}|x - x_k| : k \in P^c,\\ 
 2^{-1/2}|y - x_k| : k \in S^c,\,\, 2^{-1/2}|x_j - x_k| : j \in P^*, k \in P^c \text{ or } j \in S^*, k \in S^c\}. 
\end{multline}


\begin{lem}
\label{lem:lam_b2}
For any $x,y \in \real^3$ with $|x|, |y| \ge \epsilon$,
\begin{enumerate}[label=(\roman*)]
\item there exists $C$, independent of $\delta,\epsilon$ and $x,y,\bold{\hat{x}}$, such that when $\Phi(x,y,\hbx) \ne 0$,
\begin{align}
\label{eq:pi_bnd}
\pi(x, y, \bold{\hat{x}}) &\ge C \epsilon,\\
\label{eq:lam_bnd}
\lambda(x, y, \bold{\hat{x}}) &\ge C \delta,
\end{align}
\item for any $b \ge 0$ be such that $b \ne 3$, there exists $C$, independent of $\delta,\epsilon$ and $x,y$, such that 
\begin{multline}
\label{eq:lam_bnd2}
\int_{\textnormal{supp}\,\Phi(x,y,\,\cdot\,)} \lambda(x, y, \bold{\hat{x}})^{-b} f_{\infty}(x,\bold{\hat{x}}) f_{\infty}(y,\bold{\hat{x}}) \,d\bold{\hat{x}} \\\le C \big(\epsilon^{-b} + \delta^{\min\{0, 3-b \}}\big) \norm{\rho}^{1/2}_{L^1(B(x,1))}\norm{\rho}^{1/2}_{L^1(B(y,1))}.
\end{multline}
\end{enumerate}
\end{lem}
\begin{proof}
We can prove (i) using the expressions (\ref{eq:pi_alt}) and (\ref{eq:lam_alt}). The bound follows from Lemma \ref{lem:support} along with (\ref{eq:lam_b2}) of Lemma \ref{particle_ineqs}, since $P,S \subset Q$. For (ii) we use (\ref{eq:lam_alt}) to write
\begin{align*}
\lambda(x,y,\hbx)^{-b} \le C\big(\epsilon^{-b} + \sum_{k \in P^c}|x-x_k|^{-b} + \sum_{k \in S^c}|y-x_k|^{-b} + \sum_{\substack{j \in P^*,\, k \in P^c\, \text{or}\\ j \in S^*,\, k \in S^c}} |x_j - x_k|^{-b}\big)
\end{align*}
for some $C$ depending on $b$. The required inequality then follows from (\ref{eq:ff_int}) of Proposition \ref{prop:ff} for the first term above and from Lemma \ref{lem:lam_b} for the remaining terms.
\end{proof}

\begin{cor}
\label{cor:lam_b}
For any $x,y \in \real^3$ with $|x|, |y| \ge \epsilon$,
\begin{multline}
\label{eq:lam_int2}
\int_{\real^{3N-3}} \lambda(x, y, \bold{\hat{x}})^{-2} f_{\infty}(x,\bold{\hat{x}}) f_{\infty}(y,\bold{\hat{x}}) |\nabla \Phi(x,y,\bold{\hat{x}})| \,d\bold{\hat{x}} \\\le C \epsilon^{-3} \norm{\rho}^{1/2}_{L^1(B(x,1))}\norm{\rho}^{1/2}_{L^1(B(y,1))}
\end{multline}
for some $C$, independent of $\delta, \epsilon$ and $x,y$.
\end{cor}
\begin{proof}
Using Lemma \ref{lem:dphi} and (\ref{eq:lam_bnd}), the integral can be bounded by some constant multiplying
\begin{align*}
\int_{\textnormal{supp}\,\Phi(x,y,\,\cdot\,)} \big(\epsilon^{-1}\lambda(x, y, \bold{\hat{x}})^{-2} + \delta^{-3}M_{\delta}(x,y,\bold{\hat{x}}) \big) f_{\infty}(x,\bold{\hat{x}}) f_{\infty}(y,\bold{\hat{x}}) \,d\bold{\hat{x}}.
\end{align*}
Expanding, the first term is bounded by (\ref{eq:lam_bnd2}) and second is bounded using (\ref{eq:mt_int}) of Proposition \ref{prop:ff}.

\end{proof}

We now collect certain results which will be used throughout this section. We first consider cluster derivatives of the cutoff $\Phi$ of the form (\ref{eq:dxbc}),(\ref{eq:dxybc}). Such derivatives can be expanded as partial derivatives. Therefore, taking some $\eta \in \naturals_0^3$ and $\boldsymbol{\nu} = (\nu_1, \nu_2) \in \naturals_0^6$, there exists $C$ by Lemma \ref{lem:dphi}, depending on $\eta, \bnu$ but independent of $\delta, \epsilon$, such that
\begin{equation}
\label{eq:bc_1}
|D^{\eta}_{x,y,Q}D^{\nu_1}_{x,P}D^{\nu_2}_{y,S}\Phi_{\delta,\epsilon}(x,y,\bold{\hat{x}})| \le C \epsilon^{-|\eta|}\big( \epsilon^{-|\bnu|} + \delta^{-|\bnu|} M_{\delta}(x,y,\bold{\hat{x}})\big)
\end{equation}
for all $x,y \in \real^3$ and $\bold{\hat{x}} \in \real^{3N-3}$.

Now take any $x,y \in \real^3$ with $|x|, |y| \ge \epsilon$, and suppose $\Phi(x,y,\bold{\hat{x}}) \ne 0$. Then, as a consequence of Lemma \ref{lem:lam_b2} we have that
both $\pi(x,y,\bold{\hat{x}})$ and $\lambda(x,y,\bold{\hat{x}})$ are positive. Therefore, by the definitions (\ref{eq:lam_xyx}), (\ref{eq:mu_xyx}) and (\ref{dp_def}), (\ref{lam_def}),
\begin{align*}
(x, \bold{\hat{x}}) \in  \Sigma_{Q}^c \cap \Sigma_P^c \cap \Sigma_{S^*}^c \quad \text{and}\quad (y,\bold{\hat{x}}) \in \Sigma_{Q}^c \cap \Sigma_{P^*}^c \cap \Sigma_{S}^c.
\end{align*}
This allows us to apply Theorem \ref{thm:ddpsi} for cluster derivatives of clusters $Q,P,S^*$ on $\psi$ at the point $(x,\bold{\hat{x}})$, and of $Q,P^*,S$ on $\psi$ at the point $(y,\bold{\hat{x}})$. Indeed, using the definitions (\ref{eq:lam_xyx}) and (\ref{eq:mu_xyx}), we have for every $\eta \in \naturals_0^3$, $\boldsymbol{\nu} = (\nu_1, \nu_2) \in \naturals_0^6$ some $C$ such that
\begin{align}
\label{eq:ddpsi_1}
\sum_{k=0,1}\big|D_Q^{\eta} D^{\bnu}_{\{P,S^*\}} \nabla^k \psi(x, \bold{\hat{x}})\big| &\le C \pi(x,y,\bold{\hat{x}})^{-|\eta|} \lambda(x, y, \bold{\hat{x}})^{-|\bnu|} f_{\infty}(x, \bold{\hat{x}}),\\
\label{eq:ddpsi_2}
\sum_{k=0,1}\big|D_Q^{\eta} D^{\bnu}_{\{P^*,S\}}  \nabla^k \psi(y, \bold{\hat{x}})\big| &\le C\pi(x,y,\bold{\hat{x}})^{-|\eta|} \lambda(x, y, \bold{\hat{x}})^{-|\bnu|} f_{\infty}(y, \bold{\hat{x}}),
\end{align}
where we introduce here the notation $D^{\bnu}_{\{P,S^* \}} = D^{\nu_1}_PD^{\nu_2}_{S^*}$ and similarly $D^{\bnu}_{\{P^*,S \}} = D^{\nu_1}_{P^*}D^{\nu_2}_{S}$.

\begin{lem}
\label{lem:dglm}
Let $\mu_1,\mu_2,\mu_3 \in \naturals_0^3$ be arbitrary and $l,m \in \naturals_0^3$ be such that $|l|,|m| \le 1$.
%
The derivative $\partial^{\mu_1}_{x+y}\partial^{\mu_2}_x \partial^{\mu_3}_y \gamma_{l,m}(x,y;\Phi)$ is equal to a linear combination of integrals of the form
\begin{multline}
\label{eq:dglm_int}
\int_{\real^{3N-3}} \big( D_Q^{\alpha_1} D_{\{P, S^*\}}^{(\alpha_2, \alpha_3)} \partial_{x}^{l} \psi(x,\bold{\hat{x}})\big) \big( D_Q^{\beta_1} D_{\{P^*, S\}}^{(\beta_2, \beta_3)} \partial_{y}^{m} \psi(y,\bold{\hat{x}})\big) \big(D_{x,y,Q}^{\sigma_1} D^{\sigma_2}_{x,P} D_{y,S}^{\sigma_3}\Phi(x,y,\bold{\hat{x}})\big)\, d\bold{\hat{x}}\\
\text{where } |\alpha_i| + |\beta_i| + |\sigma_i| = |\mu_i|,\, i=1,2,3
\end{multline}
for all $|x|, |y| \ge \epsilon$. In particular, the derivative exists across the diagonal $x=y$ for $|x|, |y| \ge \epsilon$. 
\end{lem}
\begin{proof}
For each choice of $x$ and $y$ we define a $x$- and $y$-dependent change of variables for the integral $\g_{l,m}(x,y;\Phi)$. To start, we define two vectors $\bold{\hat{a}} = (a_2, \dots, a_N), \bold{\hat{b}} = (b_2, \dots, b_N) \in \real^{3N-3}$ by

\begin{align}
a_k &=
\begin{cases}
x \quad & \text{ if } k \in P^*\\
y \quad & \text{ if } k \in S^*\\
0 \quad & \text{ if } k \in P^c \cap S^c,
\end{cases}
&
b_{k} &=
\begin{cases}
(x+y)/2 \quad & \text{ if } k \in Q^* \\
0 \quad & \text{ if } k \in Q^{c}.
\end{cases}
\end{align}
By a translational change of variables using $\bold{\hat{a}}$ we can write
\begin{equation}
\label{glm_phi}
\g_{l,m}(x,y; \Phi) = \int_{\real^{3N-3}} \partial_{x}^l \psi(x, \bold{\hat{x}} + \bold{\hat{a}}) \overline{\partial_y^m \psi(y, \bold{\hat{x}}+\bold{\hat{a}})} \Phi(x,y, \bold{\hat{x}}+\bold{\hat{a}})  \,d\bold{\hat{x}}.
\end{equation}
%
%
We will then apply differentiation under the integral. Beforehand, we show how such derivatives will act on each function within the integrand. As an illustration, we take a function $f$ defined on $\real^{3N}$ and any $\eta \in \naturals_0^3$ to see by the chain rule that
\begin{equation*}
  \begin{split}
\partial^{\eta}_{x+y}[f(x, \bold{\hat{x}}+\bold{\hat{b}})] &= D^{\eta}_Q f(x, \bold{\hat{x}}+\bold{\hat{b}}),\\
\partial^{\eta}_{x+y}[f(y, \bold{\hat{x}}+\bold{\hat{b}})] &= D^{\eta}_Q f(y, \bold{\hat{x}}+\bold{\hat{b}}),\\
\partial^{\eta}_x[f(x, \bold{\hat{x}}+\bold{\hat{a}})] &= D^{\eta}_P f(x, \bold{\hat{x}}+\bold{\hat{a}}),
  \end{split}
\qquad
  \begin{split}
\partial^{\eta}_x[f(y, \bold{\hat{x}}+\bold{\hat{a}})] &= D^{\eta}_{P^*} f(y, \bold{\hat{x}}+\bold{\hat{a}}),\\
\partial^{\eta}_y[f(x, \bold{\hat{x}}+\bold{\hat{a}})] &= D^{\eta}_{S^*} f(x, \bold{\hat{x}}+\bold{\hat{a}}),\\
\partial^{\eta}_y[f(y, \bold{\hat{x}}+\bold{\hat{a}})] &= D^{\eta}_S f(y, \bold{\hat{x}}+\bold{\hat{a}}).
  \end{split}
\end{equation*}
Analogous expressions arise when the cutoff $\Phi(x,y,\bold{\hat{x}})$ is differentiated which involve cluster derivatives of the form (\ref{eq:dxbc}), (\ref{eq:dxybc}).

We now differentiate (\ref{glm_phi}) in $x$ and $y$, and use the product rule to get that the function $\partial^{\mu_2}_x\partial^{\mu_3}_y \g_{l,m}(x,y; \Phi)$ is a linear combination of terms of the form
\begin{align*}
&I_{(\al_2,\al_3,\beta_2,\beta_3,\sigma_2,\sigma_3)}\\ 
&= \int_{\real^{3N-3}} \big( D_{\{P, S^*\}}^{(\alpha_2, \alpha_3)} \partial_{x}^{l} \psi(x,\bold{\hat{x}}+\bold{\hat{a}})\big) \big( D_{\{P^*, S\}}^{(\beta_2, \beta_3)} \partial_{y}^{m} \psi(y,\bold{\hat{x}}+\bold{\hat{a}})\big) D^{\sigma_2}_{x,P} D_{y,S}^{\sigma_3}\Phi(x,y,\bold{\hat{x}}+\bold{\hat{a}})\big)\, d\bold{\hat{x}},
\end{align*}
for some $\al_2+\beta_2+\sigma_2 = \mu_2$ and $\al_3+\beta_3+\sigma_3 = \mu_3$.

We will now apply a $(x+y)$-derivative to an integral of this form. However, first we apply a change of variables to replace $\bold{\hat{a}}$ with $\bold{\hat{b}}$ as the translation of the integration variable $\bold{\hat{x}}$. Then, in the same way as before, we obtain that $\partial^{\mu_1}_{x+y}I_{(\al_2,\al_3,\beta_2,\beta_3,\sigma_2,\sigma_3)}$ is a linear combination of terms of the form
\begin{multline*}
\int_{\real^{3N-3}} \big( D^{\al_1}_Q D_{\{P, S^*\}}^{(\alpha_2, \alpha_3)} \partial_{x}^{l} \psi(x,\bold{\hat{x}}+\bold{\hat{b}})\big) \big(  D^{\beta_1}_Q D_{\{P^*, S\}}^{(\beta_2, \beta_3)} \partial_{y}^{m} \psi(y,\bold{\hat{x}}+\bold{\hat{b}})\big)\,\cdot\\ 
 D^{\sigma_1}_{x,y,Q} D^{\sigma_2}_{x,P} D_{y,S}^{\sigma_3}\Phi(x,y,\bold{\hat{x}}+\bold{\hat{b}})\big)\, d\bold{\hat{x}}.
\end{multline*}
for some $\al_1+\beta_1+\sigma_1 = \mu_1$. Finally, another translational change of variables removes $\bold{\hat{b}}$ to give the expression (\ref{eq:dglm_int}).

\end{proof}

We can now prove Lemma \ref{lem:glm} by bounding the integrals produced in the previous lemma.

\begin{proof}[Proof of Lemma \ref{lem:glm}]
Consider $\mu_1,\mu_2,\mu_3$ satisfying either (i) or (ii). By Lemma \ref{lem:dglm} it suffices to prove the required bound for integrals of the form (\ref{eq:dglm_int}). Define $a = |\al_1|+|\beta_1|+|\sigma_1|$, $b = |\al_2|+|\al_3|+|\beta_2|+|\beta_3|$ and $c = |\sigma_2|+|\sigma_3|$. Notice that $a = |\mu_1|$ and $b+c = |\mu_2|+|\mu_3|$. Then, using (\ref{eq:bc_1})-(\ref{eq:ddpsi_2}) and (\ref{eq:pi_bnd}) it is possible to bound (\ref{eq:dglm_int}) in absolute value by some constant multiplying
\begin{align*}
&\epsilon^{-a-c}\int_{\text{supp}\, \Phi(x,y,\cdot)} \lambda(x,y,\bold{\hat{x}})^{-b} f_{\infty}(x,\bold{\hat{x}}) f_{\infty}(y,\bold{\hat{x}}) \, d\bold{\hat{x}}\\
&\qquad +\epsilon^{-a}\delta^{-c} \int_{\text{supp}\, \Phi(x,y,\cdot)} M_{\delta}(x,y,\bold{\hat{x}}) \lambda(x,y,\bold{\hat{x}})^{-b} f_{\infty}(x,\bold{\hat{x}}) f_{\infty}(y,\bold{\hat{x}}) \, d\bold{\hat{x}} =: I_1 + I_2.
\end{align*}

For $I_2$ we use (\ref{eq:lam_bnd}) to bound this by some constant multiplying
\begin{align*}
\epsilon^{-a}\delta^{-c-b} \int_{\real^{3N-3}} M_{\delta}(x,y,\bold{\hat{x}}) f_{\infty}(x,\bold{\hat{x}}) f_{\infty}(y,\bold{\hat{x}}) \, d\bold{\hat{x}}\le C\epsilon^{-a} \delta^{3-b-c} \norm{\rho}^{1/2}_{L^{1}(B(x, 1))} \norm{\rho}^{1/2}_{L^{1}(B(y, 1))}
\end{align*}
for some $C$, where we used (\ref{eq:mt_int}). This gives (\ref{eq:dglm_int2}) for $I_2$.

For $I_1$, we first consider when $b \ne 3$. By (\ref{eq:lam_bnd2}), this is bounded by some constant multiplying
\begin{align*}
(\epsilon^{-a-c-b} + \epsilon^{-a-c}\delta^{\min\{0,3-b\}})\norm{\rho}^{1/2}_{L^{1}(B(x, 1))} \norm{\rho}^{1/2}_{L^{1}(B(y, 1))}
\end{align*}
which is then bounded by (\ref{eq:dglm_int2}).

Now we consider $I_1$ for $b = 3$. If we are in case (i) then we must have $c \ge 1$ since $b + c = |\mu_1|+|\mu_2| \ne 3$. Here, we use (\ref{eq:lam_bnd}) followed by (\ref{eq:lam_bnd2}) to bound $I_1$ by some constant multiplying
\begin{multline*}
\epsilon^{-a-c} \delta^{-1}\int_{\text{supp}\, \Phi(x,y,\cdot)} \lambda(x,y,\bold{\hat{x}})^{-2} f_{\infty}(x,\bold{\hat{x}}) f_{\infty}(y,\bold{\hat{x}}) \, d\bold{\hat{x}} \\\le C\epsilon^{-a-c-2}\delta^{-1}\norm{\rho}^{1/2}_{L^{1}(B(x, 1))} \norm{\rho}^{1/2}_{L^{1}(B(y, 1))}
\end{multline*}
for some $C$, which is then bounded by (\ref{eq:dglm_int2}). Finally, if we are in case (ii) for $b=3$ then $I_1$ cannot be bounded appropriately. Instead, we use that the integral (\ref{eq:dglm_int}) that we are attempting to bound must have $\al_1 = \beta_1 = 0$ and $\sigma_i = 0$, $i=1,2,3$. This allows us to use Lemma \ref{lem:ab3}, which is stated and proven below, to bound (\ref{eq:dglm_int}) directly.

\end{proof}
%


\subsection{Proof of Lemma \ref{lem:ab3}}
As before, $\Phi = \Phi_{\delta,\epsilon}$ is an arbitrary cutoff of the form (\ref{eq:pphi_def}) with corresponding clusters $P=P(\Phi), S=S(\Phi)$ and $Q = Q(\Phi)$. We assume throughout that $P^* \cap S^* = \emptyset$ and $0<2\delta \le \epsilon \le 1$.

\begin{lem}
\label{lem:ab3}
For any $\bal = (\al_1,\al_2), \bbeta=(\beta_1,\beta_2) \in \naturals_0^6$ with $|\bal| + |\bbeta| = 3$, and any $l,m \in \naturals_0^3$ with $|l|,|m| \le 1$ there exists $C$, independent of $\delta,\epsilon$, such that
\begin{multline}
\label{eq:ab3}
\Big|\int_{\real^{3N-3}} \big(D_{\{P, S^*\}}^{\bal}\partial_x^l \psi(x,\bold{\hat{x}})\big) \big(D_{\{P^*, S\}}^{\bbeta}\partial_y^m \psi(y,\bold{\hat{x}})\big) \Phi_{\delta,\epsilon}(x,y,\bold{\hat{x}})\, d\bold{\hat{x}}\Big| \\\le C\epsilon^{-3}\norm{\rho}^{1/2}_{L^{1}(B(x, 1))} \norm{\rho}^{1/2}_{L^{1}(B(y, 1))}
\end{multline}
for all $|x|, |y| \ge \epsilon$ with $|x-y| \le 2\delta$.
\end{lem}

First, we recall from Section 2.1 that
\begin{align}
F = F_c-F_s = F^{(en)}_c + F^{(ee)}_c - F_s
\end{align}
where $F_s$ is smooth and where
\begin{align}
\label{def:fee}
F^{(en)}_c(\bold{x}) = -\frac{Z}{2}\sum_{1 \le j \le N}|x_j|,\qquad F^{(ee)}_c(\bx) = \frac{1}{4}\sum_{1\le j<k \le N} |x_j - x_k|.
\end{align}

The proof of Lemma \ref{lem:ab3}, involves examining the cluster derivatives of $\psi$ using the refined estimates in (\ref{eq:gp_bnd}) of Theorem \ref{thm:ddpsi}. These allows us to write such derivatives in terms of a ``bad'' term, $F_c$, and a ``good'' term, $G^{\bal}_{\bold{P}}$, which is of higher regularity near certain singularities. The contributions involving $F_c$ are handled explicitly using integration by parts.

%
%
%
Let $\bal = (\alpha_1, \alpha_2) \in \naturals_0^6$, $\bbeta = (\beta_1, \beta_2) \in \naturals_0^6$ and let $l,m \in \naturals_0^3$ obey $|l|=|m|=1$. In the following, the cluster derivatives in (\ref{eq:d_x1xj_1}) are understood to act with respect to the ordered variables $(x,x_2, \dots, x_N)$ and the cluster derivatives in (\ref{eq:d_x1xj_2}) are understood to act with respect to the ordered variables $(y,x_2, \dots, x_N)$. For later convenience, on the right-hand side of both formulae all derivatives in the $x$- or $y$-variable are rewritten to act on $x_j$. Now assume $|\bal|,|\bbeta| \ge 1$, then
\begin{align}
\label{eq:d_x1xj_1}
D_{\{P, S^* \}}^{\bal} \partial_{x}^l |x-x_j| &=
\begin{cases}
(-1)^{|\alpha_1|+1} \partial_{x_j}^{\alpha_1 + \alpha_2 + l} |x - x_j| &\text{ if } |\alpha_2| \ge 1 \text{ and } j \in S^*\\
0 &\text{ if }|\alpha_2| \ge 1 \text{ and } j \in S^c\\
(-1)^{|\alpha_1|+1} \partial_{x_j}^{\alpha_1 + l} |x - x_j| &\text{ if } \alpha_2 = 0 \text{ and } j \in P^c\\
0 &\text{ if } \alpha_2 = 0 \text{ and } j \in P^*
\end{cases}\\
\label{eq:d_x1xj_2}
D_{\{P^*, S \}}^{\bbeta} \partial_{y}^m |y-x_j| &=
\begin{cases}
(-1)^{|\beta_2|+1} \partial_{x_j}^{\beta_1 + \beta_2 + m} |y - x_j| &\text{ if } |\beta_1| \ge 1 \text{ and } j \in P^*\\
0 &\text{ if }|\beta_1| \ge 1 \text{ and } j \in P^c\\
(-1)^{|\beta_2|+1} \partial_{x_j}^{\beta_2 + m} |y - x_j| &\text{ if } \beta_1 = 0 \text{ and } j \in S^c\\
0 &\text{ if } \beta_1 = 0 \text{ and } j \in S^*.
\end{cases}
\end{align}
where we made use of $P^* \cap S^* = \emptyset$ and the identity $(\nabla_t + \nabla_s)|t-s| \equiv 0$ for $t,s \in \real^n$. By the definition of $F_c^{(ee)}$, which is rewritten above, we have therefore proven the following lemma.
\begin{lem}
\label{lem:Fee}
For any $|l|=|m|=1$ and $|\bal|, |\bbeta| \ge 1$ we have
\begin{align}
\label{eq:fee_1}
D_{\{P, S^* \}}^{\boldsymbol{\alpha}} \partial_x^l F^{(ee)}_{c}(x,\bold{\hat{x}}) =
\begin{cases}
\frac{(-1)^{|\alpha_1|+1}}{4}\sum_{j \in S^*}  \partial_{x_j}^{\alpha_1 + \alpha_2 + l} |x - x_j|\quad &\text{if } |\alpha_2| \ge 1\\
\frac{(-1)^{|\alpha_1|+1}}{4}\sum_{j \in P^c}  \partial_{x_j}^{\alpha_1 + l} |x - x_j| \quad &\text{if } \alpha_2 = 0,\\
\end{cases}
\end{align}
\begin{align}
\label{eq:fee_2}
D_{\{P^*, S\}}^{\bbeta} \partial_y^m F^{(ee)}_{c}(y,\bold{\hat{x}}) = 
\begin{cases}
\frac{(-1)^{|\beta_2|+1}}{4}\sum_{j \in P^*}  \partial_{x_j}^{\beta_1 + \beta_2 + m} |y - x_j| \quad &\text{if } |\beta_1| \ge 1\\
\frac{(-1)^{|\beta_2|+1}}{4}\sum_{j \in S^c}  \partial_{x_j}^{\beta_2 + m} |y - x_j|\quad &\text{if } \beta_1 = 0.
\end{cases}
\end{align}
\end{lem}

We now consider bounds to these derivatives. We will use $\lambda(x,y,\bold{\hat{x}})$ as the quantity defined in (\ref{eq:lam_xyx}) for the clusters $P$ and $S$. By (\ref{eq:lam_alt}) and see, for example, the comment preceeding Lemma \ref{dcutoff_factors}, we have for each $|\alpha| \ge 1$ some $C$ such that
\begin{align}
\label{eq:modx_lam}
\big|\partial_{x_j}^{\alpha}|x-x_j|\big| + \big|\partial_{x_k}^{\alpha}|y-x_k|\big| &\le C \lambda(x, y, \bold{\hat{x}})^{1-|\alpha|}
\end{align}
for all $2 \le j,k \le N$ if $|\alpha|=1$, and all $j \in P^c$, $k \in S^c$ if $|\alpha| \ge 2$. Therefore it is clear from Lemma \ref{lem:Fee} and $P^* \cap S^* = \emptyset$, that for each $\bnu \in \naturals_0^6$ there exists $C$ such that
\begin{align}
\label{eq:ddfc}
\big|D^{\bnu}_{\{P,S^*\}}\nabla_x F_c^{(ee)}(x, \bold{\hat{x}})\big| + \big|D^{\bnu}_{\{P^*,S\}}\nabla_y F_c^{(ee)}(y, \bold{\hat{x}})\big| \le C \lambda(x, y, \bold{\hat{x}})^{-|\bnu|}.
\end{align}
In addition, direct differentiation shows that for each $\bnu \in \naturals_0^6$ there is a $C$ such that
\begin{align}
\label{eq:fen}
\big|D_{\{P, S^*\}}^{\bnu}\nabla_x F^{(en)}_c(x,\bold{\hat{x}})\big| + \big|D_{\{P^*, S\}}^{\bnu}\nabla_y F^{(en)}_c(y,\bold{\hat{x}})\big| \le C\epsilon^{-|\bnu|}
\end{align}
whenever $|x|,|y| \ge \epsilon$.

We now proceed with the proof of the main result of this subsection.

\begin{proof}[Proof of Lemma \ref{lem:ab3}]

We first consider the more difficult case of $|l|=|m|=1$. By Theorem \ref{thm:ddpsi} with $b=1/2$ we can write
\begin{align}
\label{eq:gps_def1}
D_{\{P,S^*\}}^{\bal}\nabla \psi &= G_{\{P,S^*\}}^{\bal} + \psi\big( D_{\{P,S^*\}}^{\bal}\nabla F_c\big) \quad &&\text{for } |\bal| \ge 1,\\
\label{eq:gps_def2}
D_{\{P^*,S\}}^{\bbeta}\nabla \psi &= G_{\{P^*,S\}}^{\bbeta} + \psi\big( D_{\{P^*,S\}}^{\bbeta}\nabla F_c\big)  \quad &&\text{for } |\bbeta| \ge 1,
\end{align}
and in these cases we have the estimates
\begin{align}
\label{eq:gps_bnd1}
\big|G^{\bal}_{\{P,S^*\}}(x, \bold{\hat{x}})\big| &\le C \lambda(x, y, \bold{\hat{x}})^{1/2-|\bal|} f_{\infty}(x, \bold{\hat{x}}),\\
\label{eq:gps_bnd2}
\big|G^{\bbeta}_{\{P^*,S\}}(y, \bold{\hat{x}})\big| &\le C \lambda(x, y, \bold{\hat{x}})^{1/2-|\bbeta|} f_{\infty}(y, \bold{\hat{x}}).
\end{align}

With these preliminary estimates, we can prove the result in the case where $|\bal|,|\bbeta| \ge 1$. By (\ref{eq:gps_def1}), (\ref{eq:gps_def2}) and $F_c = F^{(en)}_c + F^{(ee)}_c$ we bound the integral (\ref{eq:ab3}) by
\begin{align*}
&\int_{\real^{3N-3}} \big|G_{\{P, S^* \}}^{\bal}(x, \bold{\hat{x}})\big| \big|D_{\{P^*, S\}}^{\bbeta}\nabla_y \psi(y,\bold{\hat{x}})\big| \Phi(x,y,\bold{\hat{x}})\, d\bold{\hat{x}}\\ 
&\qquad+ \int_{\real^{3N-3}} \big|D_{\{P, S^*\}}^{\bal}\nabla_x F_c(x,\bold{\hat{x}})\big| | \psi(x,\bold{\hat{x}})| \big|G_{\{P^*, S \}}^{\bbeta}(y, \bold{\hat{x}})\big| \Phi(x,y,\bold{\hat{x}})\, d\bold{\hat{x}}\\ 
&\qquad+ \int_{\real^{3N-3}} \big|D_{\{P, S^*\}}^{\bal}\nabla_x F^{(en)}_c(x,\bold{\hat{x}})\big| |\psi(x,\bold{\hat{x}})| \big|D_{\{P^*, S\}}^{\bbeta}\nabla_y F_c(y,\bold{\hat{x}})\big| |\psi(y,\bold{\hat{x}})| \Phi(x,y,\bold{\hat{x}})\, d\bold{\hat{x}}\\
&\qquad+ \int_{\real^{3N-3}} \big|D_{\{P, S^*\}}^{\bal}\nabla_x F^{(ee)}_c(x,\bold{\hat{x}})\big| |\psi(x,\bold{\hat{x}})|  \big|D_{\{P^*, S\}}^{\bbeta}\nabla_y F^{(en)}_c(y,\bold{\hat{x}})\big| |\psi(y,\bold{\hat{x}})| \Phi(x,y,\bold{\hat{x}})\, d\bold{\hat{x}}\\
&\qquad+ \Big|\int_{\real^{3N-3}} \big(D_{\{P, S^*\}}^{\bal}\partial_x^l F^{(ee)}_c(x,\bold{\hat{x}})\big)  \big(D_{\{P^*, S\}}^{\bbeta}\partial_y^m F^{(ee)}_c(y,\bold{\hat{x}})\big) \psi(x,\bold{\hat{x}}) \overline{\psi(y,\bold{\hat{x}})} \Phi(x,y,\bold{\hat{x}})\, d\bold{\hat{x}}\Big|.
\end{align*}
By Lemma \ref{lem:Fee} the final term above can be expanded by terms of the form (\ref{eq:ijk}), since $S^* \subset P^c$ and $P^* \subset S^c$. Lemma \ref{lem:ff_int} then gives the appropriate bound. We will now show that each of the remaining integrals above can be bounded by a constant multiplying an integral of the form
\begin{align}
\label{eq:eps_lam}
\epsilon^{-a} \int_{\real^{3N-3}} \lambda(x,y,\bold{\hat{x}})^{-b}  f_{\infty}(x,\bold{\hat{x}}) f_{\infty}(y,\bold{\hat{x}}) \Phi(x,y,\bold{\hat{x}}) \,d\bold{\hat{x}}.
\end{align}
for some $a,b \ge 0$, $b \le 5/2$ and $a + b \le 3$. This follows from using the estimates (\ref{eq:ddfc}), (\ref{eq:fen}), (\ref{eq:gps_bnd1}) and (\ref{eq:gps_bnd2}). 
We also use (\ref{eq:ddpsi_2}) to bound the cluster derivatives of $\psi(y,\bold{\hat{x}})$ in the first term. The required bound then follows by Lemma \ref{lem:lam_b2}.

Next we consider the case where $|\bal|=3$ and $\bbeta = 0$. The other case where $|\bbeta|=3$ and $\bal = 0$ is similar with obvious modifications. By (\ref{eq:gps_def1}) and $F_c = F^{(en)}_c + F^{(ee)}_c$ we bound the integral (\ref{eq:ab3}) by
\begin{align*}
&\int_{\real^{3N-3}} \big|G_{\{P, S^* \}}^{\bal}(x, \bold{\hat{x}})\big| |\nabla \psi(y,\bold{\hat{x}})| \Phi(x,y,\bold{\hat{x}})\, d\bold{\hat{x}}\\ 
&\qquad+ \int_{\real^{3N-3}} \big|D_{\{P, S^*\}}^{\bal}\partial_x^l F^{(en)}_c(x,\bold{\hat{x}})\big| |\psi(x,\bold{\hat{x}})| |\nabla \psi(y,\bold{\hat{x}})| \Phi(x,y,\bold{\hat{x}})\, d\bold{\hat{x}}\\
&\qquad+ \Big|\int_{\real^{3N-3}} \big(D_{\{P, S^*\}}^{\bal}\partial_x^l F^{(ee)}_c(x,\bold{\hat{x}})\big) \psi(x,\bold{\hat{x}}) \overline{\partial_y^m \psi(y,\bold{\hat{x}})} \Phi(x,y,\bold{\hat{x}})\, d\bold{\hat{x}}\Big| =: I_1+I_2+I_3.
\end{align*}
The integral $I_1$ is bounded using (\ref{eq:gps_bnd1}), from which we can then apply Lemma \ref{lem:lam_b2}. Bounding $I_2$ is a simple application of (\ref{eq:fen}) and Proposition \ref{prop:ff}. For $I_3$ we use the expression $\psi = e^F\phi$, where $\phi$ was introduced in (\ref{phi_def}). Since $F = F^{(ee)}_c + F_c^{(en)} - F_s$ we can write
\begin{align*}
\partial_y^m \psi(y,\bold{\hat{x}}) = e^{F(y,\bold{\hat{x}})} \partial_y^m\phi(y,\bold{\hat{x}}) + \psi(y,\bold{\hat{x}}) \partial_y^m F^{(ee)}_c(y,\bold{\hat{x}}) + \psi(y,\bold{\hat{x}}) \big(\partial_y^m F^{(en)}_c(y,\bold{\hat{x}}) - \partial_y^m F_s(y,\bold{\hat{x}})\big).
\end{align*}
We expand $I_3$ according to this expression. By Lemma \ref{lem:Fee}, the first term can be expanded in terms of the form (\ref{eq:ff_holder1}) and therefore can bounded by Lemma \ref{lem:ff_holder}. Also by Lemma \ref{lem:Fee}, the second term can be expanded in terms of the form (\ref{eq:ijk}) for $\beta = m$ and $|\alpha|=4$, and therefore can bounded by Lemma \ref{lem:ff_int}. It therefore suffices to bound
\begin{align*}
\int_{\real^{3N-3}}  \big(D_{\{P, S^*\}}^{\boldsymbol{\bal}} \partial_x^l F^{(ee)}_c(x,\bold{\hat{x}})\big) \big(\partial_y^m F^{(en)}_c(y,\bx) - \partial_y^m F_s(y,\bx)\big) \psi(x,\bold{\hat{x}}) \overline{\psi(y,\bx)} \Phi(x,y,\bold{\hat{x}}) \,d\bold{\hat{x}}.
\end{align*}
We expand the derivative of $F_c^{(ee)}$ using Lemma \ref{lem:Fee}. A general term in this expansion will have the form
\begin{align*}
\int_{\real^{3N-3}}  \partial_{x_j}^{\alpha+l}|x-x_j| \big(\partial_y^m F^{(en)}_c(y,\bx) - \partial_y^m F_s(y,\bx)\big) \psi(x,\bold{\hat{x}}) \overline{\psi(y,\bx)} \Phi(x,y,\bold{\hat{x}}) \,d\bold{\hat{x}}
\end{align*}
for some $j \in P^c$ and $|\alpha| = 3$. We write this integral as follows, using integration by parts to remove one derivative from $|x-x_j|$,
\begin{align*}
&-\int_{\real^{3N-3}}  \partial_{x_j}^{\alpha}|x-x_j| \big(\partial_y^m F^{(en)}_c(y,\bx) - \partial_y^m F_s(y,\bx)\big) \partial_{x_j}^l \big(\psi(x,\bold{\hat{x}}) \overline{\psi(y,\bx)} \big) \Phi(x,y,\bold{\hat{x}}) \,d\bold{\hat{x}}\\
&\qquad- \int_{\real^{3N-3}}  \partial_{x_j}^{\alpha}|x-x_j| \big(\partial_y^m F^{(en)}_c(y,\bx) - \partial_y^m F_s(y,\bx)\big) \psi(x,\bold{\hat{x}}) \overline{\psi(y,\bx)} \partial_{x_j}^l \Phi(x,y,\bold{\hat{x}}) \,d\bold{\hat{x}}\\
&\qquad +\int_{\real^{3N-3}} \partial_{x_j}^{\alpha}|x-x_j| \partial_{x_j}^l\partial_y^m F_s(y,\bx) \psi(x,\bold{\hat{x}}) \overline{\psi(y,\bx)} \Phi(x,y,\bold{\hat{x}}) \,d\bold{\hat{x}}
\end{align*}
since $\partial_{x_j}^l\partial_y^m F^{(en)}_c(y,\bx) \equiv 0$. We now use (\ref{eq:modx_lam}), that $\nabla F^{(en)}_c \in L^{\infty}(\real^{3N})$ and that $\nabla^k F_s \in L^{\infty}(\real^{3N})$ for any integer $k \ge 1$, see (\ref{f2_linf}). Therefore, these terms can be bounded by some constant multiplying
\begin{align*}
\int_{\real^{3N-3}} \lambda(x,y,\bold{\hat{x}})^{-2}  f_{\infty}(x,\bold{\hat{x}}) f_{\infty}(y,\bold{\hat{x}}) \big(\Phi(x,y,\bold{\hat{x}}) + |\nabla\Phi(x,y,\hbx)|\big) \,d\bold{\hat{x}}
\end{align*}
which is bounded appropriately by Lemma \ref{lem:lam_b2} and Corollary \ref{cor:lam_b}.

Finally, it remains to consider the case where one or both of $l$, $m$ are zero. The strategy is analogous but simpler. We remark that we can use the following improved version of (\ref{eq:ddpsi_1}) and (\ref{eq:ddpsi_2}) in the case where $\eta = 0$ and $k=0$. For any $\boldsymbol{\nu} = (\nu_1, \nu_2) \in \naturals_0^6$ with $|\bnu| \ge 1$ there exists some $C$ such that
\begin{align*}
\big|D^{\bnu}_{\{P,S^*\}} \psi(x, \bold{\hat{x}})\big| &\le C \lambda(x, y, \bold{\hat{x}})^{1-|\bnu|} f_{\infty}(x, \bold{\hat{x}}),\\
\big|D^{\bnu}_{\{P^*,S\}}  \psi(y, \bold{\hat{x}})\big| &\le C\lambda(x, y, \bold{\hat{x}})^{1-|\bnu|} f_{\infty}(y, \bold{\hat{x}}).
\end{align*}
These inequalities follow immediately from (\ref{eq:ddpsi}) of Theorem \ref{thm:ddpsi}.


\end{proof}

We now state and prove two the two auxiliary lemmas which were used in the preceeding proof.
\begin{lem}
\label{lem:ff_int}
Let $|\alpha|,|\beta| \ge 1$ be such that $|\alpha| + |\beta| \le 5$. Then for any pair $2 \le j,k \le N$ such that $j \in P^c$ if $|\alpha| \ge 2$ and $k \in S^c$ if $|\beta| \ge 2$ we have for
\begin{align}
\label{eq:ijk}
I_{j,k} := \int_{\real^{3N-3}} \partial_{x_j}^{\alpha}|x-x_j| \partial_{x_k}^{\beta} |y-x_k| \psi(x,\bold{\hat{x}}) \overline{\psi(y,\bold{\hat{x}})}  \Phi(x,y,\bold{\hat{x}}) \,d\bold{\hat{x}},
\end{align}
that there exists $C$, independent of $\delta$ and $\epsilon$, such that
\begin{align*}
|I_{j,k}| \le C\epsilon^{2-|\al|-|\beta|}\norm{\rho}^{1/2}_{L^{1}(B(x, 1))} \norm{\rho}^{1/2}_{L^{1}(B(y, 1))}
\end{align*}
for all $|x|,|y| \ge \epsilon$ and $|x-y| \le 2\delta$.
\end{lem}
\begin{proof}
We suppose that $|\alpha| + |\beta| = 5$ otherwise the result follows immediately from (\ref{eq:modx_lam}) and Lemma \ref{lem:lam_b2}.

We first prove the case where $j \ne k$, where the integration by parts is particularly simple. Suppose without loss that $|\alpha| \ge |\beta|$. Then take any $l \le \al$ with $|l|=1$ and use integration by parts to obtain
\begin{align}
\label{eq:ijk2}
I_{j,k} &= -\int_{\real^{3N-3}} \partial_{x_j}^{\al-l}|x-x_j| \partial_{x_k}^{\beta} |y-x_k| \partial_{x_j}^{l} \big( \psi(x,\bold{\hat{x}}) \psi(y,\bold{\hat{x}})  \Phi(x,y,\bold{\hat{x}})\big) \,d\bold{\hat{x}}.
\end{align}
This can then be bounded in absolute value as required using (\ref{eq:modx_lam}) followed by either Lemma \ref{lem:lam_b2} or Corollary \ref{cor:lam_b}. This completes the proof where $j \ne k$.


For the remainder of the proof we consider the case where $j = k$. First, suppose $j \in P^c \cap S^c$. 
In the procedure that follows, we will use integration by parts to remove all derivatives from $|y-x_j|$ in $I_{j,j}$. 

Since $|\beta| \ge 1$ we can find some multiindex $\beta_1 \le \beta$ with $|\beta_1| = 1$. Integration by parts then gives
\begin{align}
\nonumber
I_{j,j} = &-\int_{\real^{3N-3}} \partial_{x_j}^{\al+\beta_1} |x-x_j| \partial_{x_j}^{\beta-\beta_1} |y-x_j| \psi(x,\bold{\hat{x}}) \psi(y,\bold{\hat{x}}) \Phi(x,y,\bold{\hat{x}}) \,d\bold{\hat{x}}\\
\label{eq:ikk_ibp}
&-\int_{\real^{3N-3}} \partial_{x_j}^{\al} |x-x_j| \partial_{x_j}^{\beta - \beta_1} |y-x_j| \partial_{x_j}^{\beta_1} \big(\psi(x,\bold{\hat{x}}) \psi(y,\bold{\hat{x}}) \Phi(x,y,\bold{\hat{x}})\big) \,d\bold{\hat{x}}.
\end{align}
We leave untouched the second integral above. For the first, if $|\beta| - |\beta_1| \ge 1$ we can remove another first-order derivative from $|y-x_j|$ by the same procedure. That is, using integration by parts to give two new terms as in (\ref{eq:ikk_ibp}). We retain the term where the derivative falls on $\psi(x,\bold{\hat{x}}) \psi(y,\bold{\hat{x}}) \Phi(x,y,\bold{\hat{x}})$. Whereas on the term where the derivative falls on $|x-x_j|$ we repeat the procedure, so long as there remains at least one derivative on $|y-x_j|$. Through this process, we obtain a formula for $I_{j,j}$. To express this formula, we first set $T=|\beta|$ and write $\beta = \sum_{s=1}^{T}\beta_s$ for some collection $|\beta_s| = 1$ where $1 \le s \le T$. Furthermore, we define
\begin{align*}
\beta_{<i} =
\begin{cases}
0 &\text{ if } i=1\\
\beta_1 &\text{ if } i=2\\
\beta_1 + \dots + \beta_{i-1} &\text{ if } i \ge 3
\end{cases}
\qquad \beta_{>i} =
\begin{cases}
0 &\text{ if } i=T\\
\beta_{T} &\text{ if } i=T-1\\
\beta_{i+1} + \dots + \beta_{T} &\text{ if } i \le T-2.
\end{cases}
\end{align*}
Then,
\begin{align}
\label{eq:ikk_iter}
I_{j,j} = (-1)^{T}\int_{\real^{3N-3}} \partial_{x_j}^{\alpha+\beta} |x-x_j| |y-x_j| \psi(x,\bold{\hat{x}}) \psi(y,\bold{\hat{x}}) \Phi(x,y,\bold{\hat{x}}) \,d\bold{\hat{x}}
+ \sum_{i=1}^{T}  (-1)^i I^{(i)}_{j,j}
\end{align}
where
\begin{equation*}
I^{(i)}_{j,j} = \int_{\real^{3N-3}} \partial_{x_j}^{\al+\beta_{<i}} |x-x_j| \partial_{x_j}^{\beta_{>i}} |y-x_j|
\partial_{x_j}^{\beta_i} \big(\psi(x,\bold{\hat{x}}) \psi(y,\bold{\hat{x}}) \Phi(x,y,\bold{\hat{x}})\big) \,d\bold{\hat{x}}.
\end{equation*}
We can bound $|I^{(i)}_{j,j}|$ for $1 \le i \le T-1$ in the same way as (\ref{eq:ijk2}) since $j \in P^c \cap S^c$. It remains to bound $I^{(T)}_{j,j}$, along with the first integral in formula (\ref{eq:ikk_iter}).

We begin by expanding $|y - x_j| = |x - x_j| + \big(|y-x_j| - |x-x_j| \big)$ and noticing that
\begin{equation}
\label{eq:delta_triangle}
\big||y-x_j| - |x- x_j| \big| \le |x-y| \le 2\delta.
\end{equation}
It is then straightforward to show (for example see the comment preceeding Lemma \ref{dcutoff_factors}) that for some $C,C'$,
\begin{align}
\label{eq:itkk}
\big|I^{(T)}_{j,j}\big| &\le C\int_{\real^{3N-3}} |x-x_j|^{-3} \big(2\delta + |x-x_j| \big)
\big|\partial_{x_j}^{\beta_T} \big(\psi(x,\bold{\hat{x}}) \psi(y,\bold{\hat{x}}) \Phi(x,y,\bold{\hat{x}})\big)\big| \,d\bold{\hat{x}}\\
\nonumber
&\le C'\int_{\real^{3N-3}} \lambda(x,y,\bold{\hat{x}})^{-2}  f_{\infty}(x,\bold{\hat{x}}) f_{\infty}(y,\bold{\hat{x}}) \big(\Phi(x,y,\bold{\hat{x}}) + |\nabla\Phi(x,y,\hbx)|\big) \,d\bold{\hat{x}}
\end{align}
where in the second step we used formula (\ref{eq:lam_alt}), since $j \in P^c$, along with (\ref{eq:lam_bnd}). This is then bounded as required by Lemma \ref{lem:lam_b2} and Corollary \ref{cor:lam_b}.
     
We now bound the first integral in (\ref{eq:ikk_iter}). Once again we use $|y - x_j| = |x - x_j| + \big(|y-x_j| - |x-x_j| \big)$, (\ref{eq:delta_triangle}) and (\ref{eq:lam_alt}) to show it suffices to bound the two expressions
\begin{align}
\label{eq:ikk_int1}
&\delta \int_{\real^{3N-3}} \lambda(x,y,\bold{\hat{x}})^{-4} f_{\infty}(x,\bold{\hat{x}}) f_{\infty}(y,\bold{\hat{x}}) \Phi(x,y,\bold{\hat{x}}) \,d\bold{\hat{x}},\\
\label{eq:ikk_int2}
&\int_{\real^{3N-3}} \partial_{x_j}^{\al+\beta} |x-x_j| |x-x_j| \psi(x,\bold{\hat{x}}) \psi(y,\bold{\hat{x}}) \Phi(x,y,\bold{\hat{x}}) \,d\bold{\hat{x}}.
\end{align}
The first is readily bounded as required using (\ref{eq:lam_bnd2}) of Lemma \ref{lem:lam_b2} with $b=4$.

It remains to bound (\ref{eq:ikk_int2}). We simplify the calculation by denoting $\sigma = \al + \beta$ and writing $\sigma = \sigma_1 + \dots + \sigma_5$ for some $|\sigma_s|=1$, $s=1, \dots, 5$. Define (in the same way as $\beta_{<i}$ and $\beta_{>i}$ above)
\begin{align*}
\sigma_{<i} =
\begin{cases}
0 &\text{ if } i=1\\
\sigma_1 &\text{ if } i=2\\
\sigma_1 + \dots + \sigma_{i-1} &\text{ if } 3 \le i \le 5,
\end{cases}
\qquad \sigma_{>i} =
\begin{cases}
0 &\text{ if } i=5\\
\sigma_5 &\text{ if } i=4\\
\sigma_{i+1} + \dots + \sigma_{5} &\text{ if } 1 \le i \le 3.
\end{cases}
\end{align*}
We now apply the same method used above, that is, we transfer successive first order derivatives via integration by parts. At each step we leave as a remainder the term where the derivative falls on $\psi(x,\bold{\hat{x}}) \psi(y,\bold{\hat{x}}) \Phi(x,y,\bold{\hat{x}})$. 
Since $|\sigma|=5$ is odd, the result after this procedure has occured five times is that (\ref{eq:ikk_int2}) is precisely equal to minus the same integral plus remainder terms. This explains the $\frac{1}{2}$-factor in the following formula,
\begin{multline*}
\int_{\real^{3N-3}} \partial_{x_j}^{\sigma} |x-x_j| |x-x_j| \psi(x,\bold{\hat{x}}) \psi(y,\bold{\hat{x}}) \Phi(x,y,\bold{\hat{x}}) \,d\bold{\hat{x}} \\
= \frac{1}{2}\sum_{i=1}^{5} (-1)^i \int_{\real^{3N-3}} \partial_{x_j}^{\sigma_{>i}} |x-x_j| \partial_{x_j}^{\sigma_{<i}} |x-x_j| \partial_{x_j}^{\sigma_i}\big(\psi(x,\bold{\hat{x}}) \psi(y,\bold{\hat{x}}) \Phi(x,y,\bold{\hat{x}})\big) \,d\bold{\hat{x}}.
\end{multline*}
By (\ref{eq:modx_lam}), each of these integrals can be bounded by the right-hand side of (\ref{eq:itkk}) for a new constant $C'$. This completes the proof in the case where $j \in P^c \cap S^c$.

Finally, we must bound (\ref{eq:ijk}) for when $j=k$ but where $j \in P^*$ or $j \in S^*$. We assume that $j \in P^*$, the case of $j \in S^*$ is similar. Hence we need only consider $|\alpha|=1$. To bound the integral, we first apply integration by parts to obtain
\begin{align*}
I_{j,j} &= -\int_{\real^{3N-3}} |x-x_j|\, \partial_{x_j}^{\beta+\alpha} |y-x_j|\, \psi(x,\bold{\hat{x}}) \psi(y,\bold{\hat{x}})  \Phi(x,y,\bold{\hat{x}}) \,d\bold{\hat{x}}
\\ 
&\qquad-\int_{\real^{3N-3}} |x-x_j|\, \partial_{x_j}^{\beta} |y-x_j|\, \partial_{x_j}^{\al}\big(\psi(x,\bold{\hat{x}}) \psi(y,\bold{\hat{x}})  \Phi(x,y,\bold{\hat{x}})\big) \,d\bold{\hat{x}}.
\end{align*}
By Lemma \ref{particle_ineqs} we have $|x-x_j| < \delta/2$ when $\Phi(x,y,\bold{\hat{x}}) \ne 0$. Along with (\ref{eq:modx_lam}) and (\ref{eq:lam_bnd}) we then get
\begin{align*}
|I_{j,j}| &\le C\delta \int_{\real^{3N-3}} \lambda(x,y,\bold{\hat{x}})^{-4} f_{\infty}(x,\bold{\hat{x}}) f_{\infty}(y,\bold{\hat{x}}) \Phi(x,y,\bold{\hat{x}}) \,d\bold{\hat{x}}\\
&\qquad+ C\int_{\real^{3N-3}} \lambda(x,y,\bold{\hat{x}})^{-2} f_{\infty}(x,\bold{\hat{x}}) f_{\infty}(y,\bold{\hat{x}}) \big(\Phi(x,y,\bold{\hat{x}}) + |\nabla\Phi(x,y,\bold{\hat{x}})|\big) \,d\bold{\hat{x}}
\end{align*}
for some $C$. The two terms are just (\ref{eq:ikk_int1}) and the right-hand side of (\ref{eq:itkk}) respectively, and hence can be suitably bounded. This completes all required cases.

\end{proof}

\begin{lem}
\label{lem:ff_holder}
Let $|\alpha| = 4$, $j \in P^c$ and $k \in S^c$. Then the integrals
\begin{align}
\label{eq:ff_holder1}
\int_{\real^{3N-3}} \partial_{x_j}^{\alpha}|x-x_j| \psi(x,\bold{\hat{x}})  e^{F(y,\bold{\hat{x}})} \overline{\nabla_y \phi(y,\bold{\hat{x}})}\Phi(x,y,\bold{\hat{x}}) \,d\bold{\hat{x}},
\end{align}
\begin{align}
\label{eq:ff_holder2}
\int_{\real^{3N-3}} e^{F(x,\bold{\hat{x}})} \nabla_x \phi(x,\bold{\hat{x}}) \partial_{x_k}^{\alpha}|y-x_k| \overline{\psi(y,\bold{\hat{x}})} \Phi(x,y,\bold{\hat{x}}) \,d\bold{\hat{x}}
\end{align}
can be bounded as in (\ref{eq:ab3}).
\end{lem}

\begin{proof}
We prove the bound for (\ref{eq:ff_holder1}), the case of (\ref{eq:ff_holder2}) is similar. We use the cutoff function $\xi \in C^{\infty}_c(\real)$, introduced in (\ref{eq:chi_def}), to define
\begin{align*}
\chi_1(z) := \xi(8|z|), \qquad \chi_2(z) := 1- \chi_1(z)
\end{align*}
$z \in \real^3$. Then $\chi_1(z) \ne 0$ implies $|z| < 1/4$ and $\chi_2(z) \ne 0$ implies $|z| > 1/8$.

We write the integral (\ref{eq:ff_holder1}) as
\begin{align}
\label{eq:ff_split2}
&\int_{\real^{3N-3}} \big(\chi_1(x-x_j) + \chi_2(x-x_j)\big) \partial_{x_k}^{\al}|x-x_j| \psi(x,\bold{\hat{x}}) e^{F(y,\bold{\hat{x}})} \overline{\partial_y^m \phi(y,\bold{\hat{x}})} \Phi(x,y,\bold{\hat{x}}) \,d\bold{\hat{x}}.
\end{align}
Expanding this and using, for example, the comment before Lemma \ref{dcutoff_factors}, the integral involving $\chi_2$ can be bounded in absolute value by some constant multiplying
\begin{align*}
\int_{\real^{3N-6}} \int_{\{x_j : |x-x_j|>1/8\}} |x-x_j|^{-3} |\psi(x,\bold{\hat{x}}) e^{F(y,\bold{\hat{x}})} \nabla \phi(y,\bold{\hat{x}})| \Phi(x,y,\bold{\hat{x}}) \,dx_j d\bold{\hat{x}}_{1,j}.
\end{align*}
We then use (\ref{eq:phi_er}), (\ref{eq:phi_psi}) and that $F$ is uniformly bounded on $\real^{3N}$ to bound this integral by some constant multiplying (\ref{eq:ff_int}) of Proposition \ref{prop:ff} with, say, $R=1/2$.

We now consider the $\chi_1$ term of (\ref{eq:ff_split2}). Firstly, we recall the notation introduced in (\ref{eq:bx1})-(\ref{eq:bx4}), namely we can write
\begin{equation*}
(y,x,\bold{\hat{x}}_{1,j}) = (y,x_2, \dots, x_{j-1}, x, x_{j+1}, \dots, x_N).
\end{equation*}
Take any $\theta \in (0,1)$. We know that $\phi \in C^{1,\theta}(\real^{3N})$. In particular, using (\ref{eq:phi_er}) and (\ref{eq:phi_psi}) there exists a constant $C$ such that when $|x-x_j| < 1/4$,
\begin{gather}
\label{eq:phi_linf}
|\partial_y^m \phi(y,x,\bold{\hat{x}}_{1,j})| \le \norm{\nabla \phi}_{L^{\infty}(B((y,\bold{\hat{x}}), 1/4))} \le Cf_{\infty}(y,\bold{\hat{x}}),\\
\label{eq:dphi_holder}
\big|\partial_y^m \phi(y,\bold{\hat{x}}) - \partial_y^m \phi(y,x,\bold{\hat{x}}_{1,j}) \big| \le |x-x_j|^{\theta}[\nabla\phi]_{\theta, B((y,\bold{\hat{x}}), 1/4)} \le C|x-x_j|^{\theta}f_{\infty}(y,\bold{\hat{x}}).
\end{gather}
The constant $C$ depends on $\theta$ but is independent of $x,y$ and $\bold{\hat{x}}$.

We now write the $\chi_1$ term in (\ref{eq:ff_split2}) as the sum of the following two integrals,
\begin{gather}
\label{eq:ff_split2a}
\int_{\R^{3N-3}} \chi_1(x-x_j) \partial_{x_j}^{\al}|x-x_j| \psi(x,\bold{\hat{x}}) e^{F(y,\bold{\hat{x}})} \overline{\partial_y^m \phi(y,x,\bold{\hat{x}}_{1,j})} \Phi(x,y,\bold{\hat{x}}) \, d\bold{\hat{x}},\\
\nonumber
\int_{\R^{3N-3}} \chi_1(x-x_j) \partial_{x_j}^{\al}|x-x_j| \psi(x,\bold{\hat{x}}) e^{F(y,\bold{\hat{x}})} \overline{\big(\partial_y^m \phi(y,\bold{\hat{x}}) - \partial_y^m \phi(y,x,\bold{\hat{x}}_{1,j}) \big)} \Phi(x,y,\bold{\hat{x}}) \, d\bold{\hat{x}}.
\end{gather}
The second integral can be bounded by some constant multiplying
\begin{align*}
\int_{\R^{3N-3}} \lambda(x,y,\bold{\hat{x}})^{-3+\theta} f_{\infty}(x,\bold{\hat{x}}) f_{\infty}(y,\bold{\hat{x}}) \Phi(x,y,\bold{\hat{x}}) \, d\bold{\hat{x}}
\end{align*}
using (\ref{eq:dphi_holder}) and that $F$ is uniformly bounded. Lemma \ref{lem:lam_b2} can then be used to obtain the required bound.

It suffices to bound (\ref{eq:ff_split2a}). Integration by parts in the variable $x_j$ is used in (\ref{eq:ff_split2a}) to remove a single derivative from $\partial_{x_j}^{\al}|x-x_j|$. 
Take any $l \le \al$ with $|l|=1$. Integral (\ref{eq:ff_split2a}) can therefore be rewritten as
\begin{multline*}
\int_{\R^{3N-3}} \partial^{l} \chi_1(x-x_j)\, \partial_{x_j}^{\al-l}|x-x_j| \, \psi(x,\bold{\hat{x}})\, e^{F(y,\bold{\hat{x}})}\, \overline{\partial_y^m \phi(y,x,\bold{\hat{x}}_{1,j})} \,\Phi(x,y,\bold{\hat{x}}) \,d\bold{\hat{x}}\\
-\int_{\R^{3N-3}} \chi_1(x-x_j)\, \partial_{x_j}^{\al-l}|x-x_j| \,\overline{\partial_y^m \phi(y,x,\bold{\hat{x}}_{1,j})}\, \partial_{x_j}^{l} \big(e^{F(y,\bold{\hat{x}})} \psi(x,\bold{\hat{x}}) \Phi(x,y,\bold{\hat{x}}) \big) \,d\bold{\hat{x}}.
\end{multline*}
We can use (\ref{eq:modx_lam}), (\ref{eq:phi_linf}) and the fact that both $F$ and $\nabla F$ are uniformly bounded, to bound this in absolute value by some constant multiplying
\begin{align*}
\int_{\real^{3N-3}} \lambda(x,y,\bold{\hat{x}})^{-2} f_{\infty}(x,\bold{\hat{x}}) f_{\infty}(y,\bold{\hat{x}}) \big(\Phi(x,y,\bold{\hat{x}})+|\nabla \Phi(x,y,\bold{\hat{x}})|\big) \,d\bold{\hat{x}}.
\end{align*}
The relevant bound then follows by Lemma \ref{lem:lam_b2} and Corollary \ref{cor:lam_b}.

\end{proof}

\appendix

\section{Second derivatives of $\phi$}
\label{chpt:app}
In \cite{fsho_c11} it was shown that $\phi$, as defined in (\ref{phi_def}), has improved smoothness upon multiplication by a certain exponential factor, depending only on $N$ and $Z$. It is the aim of this section to use these results to give pointwise bounds to the second derivatives of $\phi$ itself, as required in the proof of Theorem \ref{thm:ddpsi}.

The above authors introduce the following functions
\begin{align*}
\tilde{F} &= -\frac{Z}{2}\sum_{j=1}^N \xi(|x_j|)|x_j| + \frac{1}{4}\sum_{1\le j<k \le N}\xi(|x_j-x_k|)|x_j - x_k|,\\
\tilde{G} &=  C_0 Z \sum_{1\le j<k\le N} \xi(|x_j|)\xi(|x_k|)(x_j \cdot x_k) \log\big(|x_j|^2 + |x_k|^2 \big),
\end{align*}
where $\xi \in C^{\infty}_c(\real)$ is defined in (\ref{eq:chi_def}) and $C_0 = (2-\pi)/12\pi$. The authors obtain $e^{-\tilde{F}-\tilde{G}}\psi \in W_{loc}^{2,\infty}(\real^{3N})$, \cite[Theorem 1.5]{fsho_c11}. We note that this is an improvement upon the well-known fact that $\phi \in C^{1,\theta}(\real^{3N})$ for all $\theta \in (0,1)$, see Section \ref{chpt:2}.

For our purposes, it will be convenient to restate their results using a slightly different form of the exponential factor. We define
\begin{align*}
\phi' = e^{-F-G}\psi
\end{align*}
where $F = F_c-F_s$ was defined in (\ref{f_def}) and
\begin{align*}
G &= C_0 Z \sum_{j<k}(x_j \cdot x_k) \log\big(|x_j|^2 + |x_k|^2 \big) - C_0 Z \sum_{j<k}(x_j \cdot x_k) \log\big(|x_j|^2 + |x_k|^2 +1 \big)\\
&=: G_c - G_s.
\end{align*}


We can then write $\phi' = e^{H} e^{-\tilde{F}-\tilde{G}}\psi$ for
\begin{align*}
H= -F_c + F_s + \tilde{F} - G_c + G_s + \tilde{G}.
\end{align*}
It can be verified, with the help of (\ref{f2_linf}) and (\ref{f_linf}) and direct calculation, that $H$ is smooth and $\partial^{\al}H \in L^{\infty}(\real^{3N})$ for all $|\al| \le 2$, $\al \in \naturals_0^{3N}$. Therefore, we can restate \cite[Theorem 1.5]{fsho_c11} as follows.
\begin{thm}[S. Fournais, M. and T. Hoffmann-Ostenhof, T.\O. S\o rensen]
\label{thm:fs}
For all $0<r<R$ we have a constant $C$, depending on $r$ and $R$, such that
\begin{equation}
\norm{\phi'}_{W^{2,\infty}(B(\bold{x}, r))} \le C\norm{\phi'}_{L^{\infty}(B(\bold{x}, R))}
\end{equation}
for all $\bold{x} \in \real^{3N}$. The constant does not depend on $\bx$.
\end{thm}
\begin{remark}
We use the equivalence of the spaces $C^{1,1}(\overline{B})$ and $W^{2,\infty}(B)$ in the case of an open ball $B$, see for example \cite{evans_gariepy}. 
\end{remark}

The following lemma can be verified by straightforward calculations. Here, $e_i$ represents the $i$-th unit basis vector in $\real^3$, $i=1,2,3$.
\begin{lem}
\label{lem:ddg}
$G, \nabla G \in L^{\infty}(\real^{3N})$. Also, for all $j,k = 1, \dots N$ and $r,s \in \{1,2,3\}$ there exists some $C$ such that
\begin{align*}
|\partial^{e_r}_{x_j} \partial^{e_s}_{x_k} G(\bx)| \le 
\begin{cases}
C\big(1 - \log\big(\min\{1,\, |x_j|^2+|x_k|^2\}\big)\big) &\text{ if } j \ne k \text{ and } r=s\\
C &\text{ otherwise}.
\end{cases}
\end{align*}
\end{lem}

Since $\phi$ does not contain the additional exponential factor involving $G$, we do not expect it's second order derivatives to be bounded. However, we obtain the following corollary of Theorem \ref{thm:fs} to obtain pointwise bounds showing the second order derivatives of $\phi$ have only logarithmic singularities. We recall that $\nu_P$ was defined in (\ref{eq:nup}).
\begin{cor}
\label{cor:d2phi}
For all $0<r<R<1$ there exists $C$, depending on $r$ and $R$, such that for any non-empty cluster $P$ and any $\eta \in \naturals_0^3$ with $|\eta|=1$,
\begin{align*}
\norm{D_P^{\eta} \nabla \phi}_{L^{\infty}(B(\bold{x}, r \nu_P(\bold{x})))} \le C(1-\log\nu_P(\bx)) \norm{\phi}_{L^{\infty}(B(\bold{x}, R))}
\end{align*}
for all $\bold{x}$ with $\nu_P(\bx) > 0$.
\end{cor}
\begin{rem}
By the definitions, it is immediate that $\nu_P \ge \lambda_P$ and hence, in particular, the above inequality holds on $\Sigma^c_P$.
\end{rem}
\begin{proof}

By the definition of cluster derivatives, (\ref{eq:cd}), it follows from Lemma \ref{lem:ddg} that there is some $C$ such that
\begin{align}
\label{eq:ddg}
|D_P^{\eta}\nabla G(\bx)| \le C(1 - \log(\nu_P(\bx)))
\end{align}
for all $\bx$ with $\nu_P(\bx) > 0$.

By (\ref{lam_lip}), we get that $(1-r)\nu_P(\bold{x}) \le \nu_P(\bold{y})$ for all $\bold{y} \in B(\bold{x}, r\nu_P(\bold{x}))$. Therefore, for $C$ as in (\ref{eq:ddg}), we have for all $\bold{x}$ with $\nu_P(\bx)>0$,
\begin{align}
\label{eq:ddg2}
\norm{D_P^{\eta} \nabla G}_{L^{\infty}(B(\bold{x}, r\nu_P(\bold{x})))} &\le C(1 - \log(1-r) - \log(\nu_P(\bx)))\\
&\le C'(1 - \log(\nu_P(\bx)))
\end{align}
for some $C'$ depending on $r$.

We recall from the definition, (\ref{phi_def}), that $\phi = e^{-F}\psi$. Therefore $\phi = e^G\phi'$. Then, we have $\nabla \phi = e^G \phi' \nabla G + e^G \nabla \phi'$. And therefore the following formula holds,
\begin{align*}
D_P^{\eta} \nabla \phi = \big(D_P^{\eta} \nabla G + D_P^{\eta} G\, \nabla G\big)e^G\phi' + e^G D_P^{\eta} \phi'\, \nabla G +  e^G D_P^{\eta}G\, \nabla \phi' +  e^G D_P^{\eta}\nabla \phi'.
\end{align*}
Taking the norm and using that $G,\nabla G \in L^{\infty}$ by Lemma \ref{lem:ddg} we can then obtain $C$ such that
\begin{align}
\label{eq:ddphi_b1}
\norm{D_P^{\eta} \nabla \phi}_{L^{\infty}(B(\bold{x}, r\nu_P(\bold{x})))} &\le C(1-\log(\nu_P(\bx))) \norm{\phi'}_{W^{2,\infty}(B(\bold{x}, r\nu_P(\bold{x})))}.
\end{align}
Using that $\nu_P \le 1$, we then apply Theorem \ref{thm:fs} followed by a use of the equality $\phi = e^G\phi'$ to obtain the required result.

\end{proof}

\vskip 0.5cm

\textbf{Acknowledgments.} The author would like to thank A. V. Sobolev for helpful discussions in all matters of the current work.

\bibliographystyle{unsrt}
\bibliography{refs}

\begin{thebibliography}{10}

\bibitem{rs2}
M.~Reed and B.~Simon.
\newblock {\em Methods of Modern Mathematical Physics II: Fourier Analysis,
  Self-Adjointness}.
\newblock Academic Press, Inc., 1975.

\bibitem{kato}
T.~Kato.
\newblock On the eigenfunctions of many-particle systems in quantum mechanics.
\newblock {\em Comm. Pure Appl. Math.}, 10:151--177, 1957.

\bibitem{HOS_81}
M.~Hoffmann-Ostenhof and R.~Seiler.
\newblock Cusp conditions for eigenfunctions of n-electron systems.
\newblock {\em Phys. Rev. A.}, 23(1), 1981.

\bibitem{densities_atoms}
M.~Hoffmann-Ostenhof, T.~Hoffmann-Ostenhof, and T.{\O}. S{\o}rensen.
\newblock {Electron Wavefunctions and Densities for Atoms}.
\newblock {\em Ann. Henri. Poincare}, 2:77--100, 2001.

\bibitem{yse_02}
H.~Yserentant.
\newblock {On the regularity of the electronic Schr\"{o}dinger equation in
  Hilbert spaces of mixed derivatives}.
\newblock {\em Numer. Math.}, 98:731--759, 2004.

\bibitem{analytic_rep}
S.~Fournais, M.~Hoffmann-Ostenhof, T.~Hoffmann-Ostenhof, and T.{\O}.
  S{\o}rensen.
\newblock {Analytic Structure of Many-Body Coulombic Wave Functions}.
\newblock {\em Comm. Math. Phys.}, 289:291--310, 2009.

\bibitem{ammann12}
B.~Ammann, C.~Carvalho, and V.~Nistor.
\newblock {Regularity for Eigenfunctions of Schr\"{o}dinger Operators}.
\newblock {\em Lett. Math. Phys.}, 101:49--84, 2012.

\bibitem{coulomb_estimates}
S.~Fournais and T.{\O}. S{\o}rensen.
\newblock {Estimates on derivatives of Coulombic wave functions and their
  electron densities}.
\newblock {\em J. reine angew. Math.}, 2021(775):1--38, 2021.

\bibitem{hearn_sob2}
P.~Hearnshaw and A.V. Sobolev.
\newblock {The diagonal behaviour of the one-particle Coulombic density
  matrix}.
\newblock {\em Probability and Mathematical Physics}, 4(4):935--964, 2023.

\bibitem{cio20}
J.~Cioslowski.
\newblock Off-diagonal derivative discontinuities in the reduced density
  matrices of electronic systems.
\newblock {\em J. Chem. Phys}, 153(154108), 2020.

\bibitem{cio22}
J.~Cioslowski.
\newblock Reverse engineering in quantum chemistry: How to reveal the
  fifth-order off-diagonal cusp in the one-electron reduced density matrix
  without actually calculating it.
\newblock {\em Int J Quantum Chem}, 122(8), 2022.
\newblock e26651.

\bibitem{hearn_sob}
P.~Hearnshaw and A.V. Sobolev.
\newblock {Analyticity of the One-Particle Density Matrix}.
\newblock {\em Ann. Henri. Poincare}, 23:707--738, 2022.

\bibitem{jecko22}
T.~Jecko.
\newblock On the analyticity of electronic reduced densities for molecules.
\newblock {\em J. Math. Phys.}, 63(1), 2022.
\newblock 013509.

\bibitem{density_analytic}
S.~Fournais, M.~Hoffmann-Ostenhof, T.~Hoffmann-Ostenhof, and T.{\O}.
  S{\o}rensen.
\newblock Analyticity of the density of electronic wavefunctions.
\newblock {\em Ark. Mat.}, 42:87--106, 2004.

\bibitem{jecko10}
T.~Jecko.
\newblock {A New Proof of the Analyticity of the Electronic Density of
  Molecules}.
\newblock {\em Lett. Math. Phys.}, (93):73--83, 2010.

\bibitem{sob_estimates}
A.V. Sobolev.
\newblock Eigenvalue estimates for the one-particle density matrix.
\newblock {\em J. Spectr. Theory}, 12(2):857--875, 2022.

\bibitem{sob_asymptotics}
A.V. Sobolev.
\newblock Eigenvalue asymptotics for the one-particle density matrix.
\newblock {\em Duke Math. J.}, 171(17):3481--3513, 2022.

\bibitem{sob_kinetic}
A.V. Sobolev.
\newblock Eigenvalue asymptotics for the one-particle kinetic energy density
  operator.
\newblock {\em Journal of Functional Analysis}, 283(8), 2022.

\bibitem{jecko23}
T.~Jecko and C.~No\^{u}s.
\newblock Limited regularity of a specific electronic reduced density matrix
  for molecules.
\newblock 2023.
\newblock arXiv, 2309.06318.

\bibitem{jecko24}
T.~Jecko and C.~No\^{u}s.
\newblock {Regularity of the (N-1)-particle electronic reduced density matrix
  for molecules with fixed nuclei and N electrons}.
\newblock 2024.
\newblock arXiv, 2407.03706.

\bibitem{density_smooth}
S.~Fournais, M.~Hoffmann-Ostenhof, T.~Hoffmann-Ostenhof, and T.{\O}.
  S{\o}rensen.
\newblock The electron density is smooth away from the nuclei.
\newblock {\em Comm. Math. Phys.}, 228(3):401--415, 2002.

\bibitem{fsho_c11}
S.~Fournais, M.~Hoffmann-Ostenhof, T.~Hoffmann-Ostenhof, and T.{\O}.
  S{\o}rensen.
\newblock {Sharp Regularity Results for Coulombic Many-Electron Wave
  Functions}.
\newblock {\em Comm. Math. Phys.}, 255:183--227, 2005.

\bibitem{evans_gariepy}
L.C. Evans and R.F. Gariepy.
\newblock {\em Measure Theory and Fine Properties of Functions}.
\newblock CRC Press, 1st edition, 1992.

\end{thebibliography}

\end{document}